\def\spaceconstraints{false}

\documentclass[a4paper,english]{lipics-v2019}
\usepackage{nth}
\usepackage{microtype}
\usepackage{mathtools}
\usepackage{breqn}
\usepackage{thm-restate}
\usepackage{tikz,pgfplots}
\usepackage{ifthen}
\usepackage{placeins}
\usepackage{refcount}
\pgfplotsset{compat=1.14, compat/show suggested version=false}
\newcounter{reftmpcounter}
\newcounter{resttmpcounter}
\newtheorem{observation}[theorem]{Observation}
\newcommand{\restatethm}[2]{
	\setcounterref{reftmpcounter}{#2}
	\setcounter{resttmpcounter}{\thetheorem}
	\setcounter{theorem}{\thereftmpcounter}
	\addtocounter{theorem}{-1}
	#1
	\setcounter{theorem}{\theresttmpcounter}
}

\newcommand{\revsocg}[1]{{\color{black} #1}}

\hideLIPIcs
\nolinenumbers

\relatedversion{An extended abstract based on the content of this paper has been accepted for inclusion in the Symposium on Computational Geometry (SoCG) 2020 \cite{thesocgpaper}.}
\title{Worst-Case Optimal Covering of Rectangles \texorpdfstring{\newline}{} by Disks}
\titlerunning{Worst-Case Optimal Covering of Rectangles by Disks}

\author{Sándor P. Fekete}{Department of Computer Science, TU Braunschweig, Germany}{s.fekete@tu-bs.de}{https://orcid.org/0000-0002-9062-4241}{}
\author{Utkarsh Gupta}{Department of Computer Science \& Engineering, IIT Bombay, India}{utkarshgupta149@gmail.com}{https://orcid.org/0000-0002-5324-6499}{}
\author{Phillip Keldenich}{Department of Computer Science, TU Braunschweig, Germany}{p.keldenich@tu-bs.de}{https://orcid.org/0000-0002-6677-5090}{}
\author{Christian Scheffer}{Department of Computer Science, TU Braunschweig, Germany}{scheffer@ibr.cs.tu-bs.de}{https://orcid.org/0000-0002-3471-2706}{}
\author{Sahil Shah}{Department of Computer Science \& Engineering, IIT Bombay, India}{sahilshah00199@gmail.com}{https://orcid.org/0000-0001-7854-1585}{}

\authorrunning{S. P. Fekete and U. Gupta and P. Keldenich and C. Scheffer and S. Shah}
\Copyright{Sándor P. Fekete and Utkarsh Gupta and Phillip Keldenich and Christian Scheffer\newline and Sahil Shah}

\ccsdesc{Theory of computation $\rightarrow$ Packing and covering problems}
\ccsdesc{Theory of computation $\rightarrow$ Computational geometry}

\keywords{Disk covering, critical density, covering coefficient, tight worst-case bound, interval arithmetic, approximation}

\supplement{The code of the automatic prover can be found at \url{https://github.com/phillip-keldenich/circlecover}.
			Furthermore, there is a video contribution~\cite{thevideo}, video available at \url{https://www.ibr.cs.tu-bs.de/users/fekete/Videos/Cover_full.mp4}, illustrating the algorithm and proof presented in this paper.}

\funding{} 

\acknowledgements{We thank Sebastian Morr and Arne Schmidt for helpful discussions.
Major parts of the research by Utkarsh Gupta and Sahil Shah were carried out during a stay at TU Braunschweig.}

\EventEditors{Sergio Cabello and Danny Z. Chen}
\EventNoEds{2}
\EventLongTitle{36th International Symposium on Computational Geometry (SoCG 2020)}
\EventShortTitle{SoCG 2020}
\EventAcronym{arXiv}
\EventYear{2020}
\EventDate{June 23--26, 2020}
\EventLocation{Z\"{u}rich, Switzerland}
\EventLogo{socg-logo}
\SeriesVolume{164}
\ArticleNo{12}

\newboolean{havespaceconstraint}
\setboolean{havespaceconstraint}{\spaceconstraints}

\numberwithin{subcase}{case}

\DeclarePairedDelimiter\abs{\lvert}{\rvert}

\newcounter{sbsubseccounter}

\newcounter{wcsubseccounter}

\newcommand{\sbnewsubsec}[1]{\refstepcounter{sbsubseccounter}\label{#1}}
\newcommand{\wcnewsubsec}[1]{\refstepcounter{wcsubseccounter}\label{#1}}
\newcommand{\sbstratref}[2]{\mbox{S-\ref{#1}.{#2}}}
\newcommand{\psbstratref}[2]{\noindent\textbf{(\sbstratref{#1}{#2})}\ \ }
\newcommand{\wcstratref}[2]{\mbox{W-\ref{#1}.{#2}}}
\newcommand{\pwcstratref}[2]{\noindent\textbf{(\wcstratref{#1}{#2})}\ \ }
\newcommand{\mparagraph}[1]{\noindent\textbf{#1.}\ \;}

\newboolean{applemmaworstcasesrectangles}	\setboolean{applemmaworstcasesrectangles}{false}
\newboolean{applemmasizeboundlarge}			\setboolean{applemmasizeboundlarge}{false}
\newboolean{appproofwallbuilding}			\setboolean{appproofwallbuilding}{false}
\newboolean{appproofsizebounded}			\setboolean{appproofsizebounded}{false}
\newboolean{appproofmainthm}				\setboolean{appproofmainthm}{false}

\ifthenelse{\boolean{havespaceconstraint}}{
	\setboolean{applemmaworstcasesrectangles}{true}
	\setboolean{applemmasizeboundlarge}{false}
	\setboolean{appproofwallbuilding}{true}
	\setboolean{appproofsizebounded}{true}
	\setboolean{appproofmainthm}{true}
}{}

\begin{document}
\maketitle
\begin{abstract}
We provide the solution for a fundamental problem of geometric optimization by
giving a complete characterization of worst-case optimal disk coverings of
rectangles: For any $\lambda\geq 1$, the critical covering area $A^*(\lambda)$ is the minimum value for
which any set of disks with total area at least $A^*(\lambda)$ can cover a rectangle of
dimensions $\lambda\times 1$.
We show that there is a threshold value $\lambda_2 =  \sqrt{\sqrt{7}/2 - 1/4} \approx  1.035797\ldots$,
such that for $\lambda<\lambda_2$ the critical covering area $A^*(\lambda)$ is
$A^*(\lambda)=3\pi\left(\frac{\lambda^2}{16} +\frac{5}{32} + \frac{9}{256\lambda^2}\right)$, and for $\lambda\geq \lambda_2$, the critical area is 
$A^*(\lambda)=\pi(\lambda^2+2)/4$; these values are tight.
For the special case $\lambda=1$, i.e., for covering a unit square, the critical covering area is $\frac{195\pi}{256}\approx 2.39301\ldots$.
The proof uses a careful combination of manual and automatic analysis, demonstrating the power of the employed interval arithmetic technique.
\end{abstract}

\section{Introduction}
Given a collection of (not necessarily equal) disks, is it
possible to arrange them so that they completely cover a given region, such
as a square or a rectangle? 
Covering problems of this type are of fundamental theoretical interest, but also
have a variety of different applications, most notably in 
sensor networks, communication networks, wireless communication, surveillance, robotics, 
and even gardening and sports facility management, as shown in~Fig.~\ref{fig:soccer}.

\begin{figure}
  \begin{center}
      \includegraphics[height=4cm]{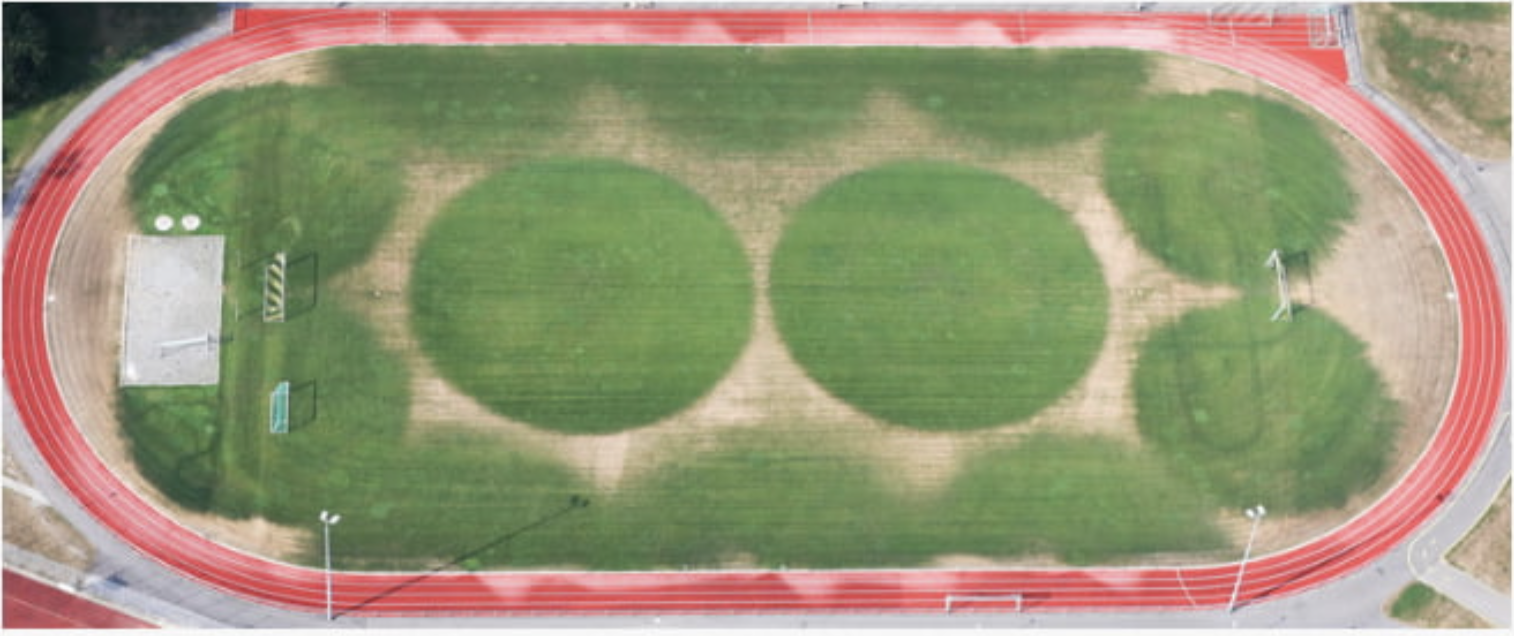}
  \end{center}
  \vspace*{-12pt}
  \caption{An incomplete covering of a rectangle by disks: Sprinklers on a soccer field during a drought. (Source: dpa~\cite{dpa}.)}
  \label{fig:soccer}
\end{figure}

If the total area of the disks is small, it is clear that completely covering the region is impossible.
On the other hand, if the total disk area is sufficiently large, finding a covering seems easy;
however, for rectangles with large aspect ratio, a major fraction of the
covering disks may be useless, so a relatively large total disk area may be
required.  The same issue is of clear importance for applications:
What fraction of the total cost of disks can be put to efficient use for covering?
This motivates the question of characterizing a critical threshold:
For any given $\lambda$, find the minimum value $A^*(\lambda)$ for which any collection of disks with total area at least
$A^*(\lambda)$ can cover a rectangle of dimensions $\lambda\times 1$.
What is the critical covering area of $\lambda \times 1$ rectangles?
In this paper we establish a complete and tight characterization.

\subsection{Related Work}
Like many other packing and covering problems, disk covering is typically quite
difficult, compounded by the geometric complications of dealing with
irrational coordinates that arise when arranging circular objects. This is 
reflected by the limitations of provably optimal results for the largest disk, square
or triangle that can be covered by $n$ unit disks, and hence, the ``thinnest'' disk covering,
i.e., a covering of optimal density.
As early as 1915, Neville~\cite{neville1915solution} computed the optimal arrangement for covering a disk by five unit disks, 
but reported a wrong optimal value; much later, Bezdek\cite{bezdek1979korok,bezdek1984einige} gave the correct value for $n=5,6$.
As recently as 2005, Fejes T\'{o}th~\cite{toth2005thinnest} established optimal values for $n=8,9,10$.
The question of incomplete coverings was raised 
in 2008 by Connelly, who asked how one should place $n$ small disks of radius 
$r$ to cover the largest possible area of a disk of radius $R > r$.
Szalkai~\cite{szalkai2016optimal} gave an optimal solution for $n=3$.
For covering rectangles by $n$ unit disks, Heppes and Mellissen~\cite{heppes1997covering} 
gave optimal solutions for $n\leq 5$; Melissen and Schuur~\cite{melissen2000covering} extended this
for $n=6,7$.  See Friedman~\cite{friedman1014} for 
the best known solutions for $n \leq 12$.
Covering equilateral triangles by $n$ unit disks 
has also been studied. Melissen~\cite{melissen1997loosest} gave optimality results for $n\leq 10$,
and conjectures for $n\leq 18$; the difficulty of these seemingly small problems is illustrated
by the fact that Nurmela~\cite{nurmela2000conjecturally} gave conjectured optimal solutions
for $n\leq 36$, improving \revsocg{the conjectured optimal covering for $n=13$ of} Melissen. 
Carmi et al.~\cite{carmi2007covering} considered algorithms for covering point sets by unit disks at fixed locations. 
There are numerous other related problems and results; for relevant surveys, see
Fejes T\'{o}th~\cite{toth1999recent} (Section 8),
Fejes T\'{o}th~\cite{toth20172} (Chapter 2),
Brass et al.~\cite{brass2005density} (Chapter 2) and the book by
B{\"o}r{\"o}czky~\cite{boroczky2004finite}.

Even less is known for covering by non-uniform disks, with most previous research focusing on
algorithmic aspects. Alt et al.~\cite{aab+-mccps-06} gave algorithmic results for minimum-cost covering of 
point sets by disks, where the cost function is $\sum_j r_j^{\alpha}$ for some $\alpha>1$, which 
includes the case of total disk area for $\alpha=2$.
Agnetis et al.~\cite{agnetis2009covering} discussed covering a line segment with variable radius disks.
Abu-Affash et al.~\cite{abu2011multi} studied covering a polygon minimizing the sum of areas; 
for recent improvements, see Bhowmick et al.~\cite{multi}.
B{\'a}nhelyi et al.~\cite{banhelyi2015optimal} gave algorithmic results for
the covering of polygons by variable disks with prescribed centers.

For relevant applications, we mention
the survey by Huang and Tseng~\cite{huang2005survey} for
wireless sensor networks, the work by Johnson et al.~\cite{johnson2012more} on
covering density for sensor networks, the algorithmic results 
for placing a given number of base stations to cover a 
square~\cite{das2006efficient} and a convex region by  Das et al.~\cite{das2008variations}.
For minimum-cost sensor coverage of planar regions, see Xu et al.~\cite{xu2008minimum};
for wireless communication coverage of a square, see Singh and Sengupta~\cite{singh2013efficient},
and Palatinus and B{\'a}nhelyi~\cite{palatinus2010circle} for the context of telecommunication networks.

The analogous question of {\em packing} unit disks into a square has also 
attracted attention. For $n=13$, the optimal value 
for the densest square covering was only established in 2003~\cite{13disks},
while the optimal value for 14 unit disks is still unproven; 
densest packings of $n$ disks in equilateral triangles are subject to a long-standing
conjecture by Erd\H{o}s and Oler from 1961~\cite{oler} 
that is still open for $n=15$.
Other mathematical work on densely packing relatively small numbers
of identical disks includes~\cite{goldberg71,melissen94,19disks,12disks}, and
\cite{reis75,lubachevsky97,graham98} for related experimental work.
The best known solutions for packing equal disks into
squares, triangles and other shapes are published on Specht's
website \url{http://packomania.com} \cite{specht2015packomania}.

Establishing the critical packing density for (not
necessarily equal) disks in a square was proposed by Demaine, Fekete, and
Lang~\cite{DFL2010circle} and solved by Morr, Fekete and
Scheffer~\cite{morr2017split,Fekete2018}.  Using a recursive procedure
for cutting the container into triangular pieces, they proved that the critical
packing density of disks in a square is $\frac{\pi}{3+2\sqrt{2}} \approx
0.539$. The critical density for (not necessarily equal) disks in a disk was recently proven to be
1/2 by Fekete, Keldenich and Scheffer~\cite{fks-pddow-19}; see the video~\cite{bfk+-pgoow-19}
for an overview and various animations.
The critical packing density of (not necessarily equal)
squares was established in 1967 by Moon and Moser~\cite{MM1967some}, who
used a shelf-packing approach to establish the value of 1/2 for packing into a square.

\subsection{Our Contribution}
We show that there is a threshold value $\lambda_2 =  \sqrt{\sqrt{7}/2 - 1/4} \approx  1.035797\ldots$,
such that for $\lambda<\lambda_2$ the critical covering area $A^*(\lambda)$ is
$A^*(\lambda)=3\pi\left(\frac{\lambda^2}{16} +\frac{5}{32} + \frac{9}{256\lambda^2}\right)$, and for $\lambda\geq \lambda_2$, the critical area is
$A^*(\lambda)=\pi(\lambda^2+2)/4$.
These values are tight: For any $\lambda$, any collection of disks of total area
$A^*(\lambda)$ can be arranged to cover a $\lambda\times 1$-rectangle, and for any $a(\lambda)<A^*(\lambda)$, there
is a collection of disks of total area $a(\lambda)$ such that a $\lambda\times 1$-rectangle cannot be covered.
(See Fig.~\ref{fig:critical-covering-coeff} for a graph showing the (normalized) critical covering density,
and Fig.~\ref{fig:worst-cases-rectangles} for examples of worst-case configurations.)
The point $\lambda = \lambda_2$ is the unique real number greater than $1$ for which the two bounds \(3\pi\left(\frac{\lambda^2}{16} + \frac{5}{32} + \frac{9}{256\lambda^2}\right)\) and \(\pi\frac{\lambda^2 + 2}{4}\) coincide; see Fig.~\ref{fig:critical-covering-coeff}.
At this so-called \emph{threshold value}, the worst case changes from three identical disks to two disks --- the circumcircle \(r_1^2 = \frac{\lambda^2+1}{4}\) and a disk \(r_2^2 = \frac{1}{4}\); see Fig.~\ref{fig:worst-cases-rectangles}.
For the special case $\lambda=1$, i.e., for covering a unit square, the critical covering area is $\frac{195\pi}{256}\approx 2.39301\ldots$.

The proof uses a careful combination of manual and automatic analysis, demonstrating the power of the employed interval arithmetic technique.

\begin{figure}
        \begin{center}
                \begin{tikzpicture}
                        \begin{axis}[%
                                xlabel=Skew $\lambda$,%
								ylabel=Critical covering density {$d^*(\lambda) = \frac{A^*(\lambda)}{\lambda}$},%
                                xmin=1.0,%
                                xmax=2.5,%
                                scale only axis,%
                                height=5.5cm,%
                                width=.75\linewidth,%
                                ytick = {2.2, 2.3, 2.4, 2.5, 2.6 },
                                extra x ticks = {1.035797,1.4142135623730951,2.089884158041382},%
                                extra x tick labels = { $\lambda_2$, $\sqrt{2}$, $\overline{\lambda} = (195+\sqrt{5257})/128$ },%
                                extra x tick style = {
                                        grid=none,
                                        tick style={draw=none},
                                        xticklabel pos = upper,
                                        xticklabel style = { anchor=south },
                                        font=\small
                                },
                                extra y ticks = {2.39301003, 2.22144147, 2.33002281},
                                extra y tick labels = {$\frac{195\pi}{256}$, $\frac{\pi}{\sqrt{2}}$, $\sqrt{\frac{469+182\sqrt{7}}{1728}}\pi$},
                                extra y tick style = {
                                        grid=none,
                                        tick style={draw=none},
                                        yticklabel pos = right,
                                        yticklabel style = { anchor=west },
                                        font=\small
                                },
                                y label style={at={(axis description cs:-0.05,.5)}},
                                y tick label style={font=\small},
                                x tick label style={font=\small}
                        ]%
                                \addplot[mark=none,color=blue,samples=1000,domain=0.99:2.51] {3.1415926 * max((x*x+2.0)/(4.0*x),(3.0/x)*(x*x/16.0 + 5.0/32 + 9.0/(256.0*x*x))};
                                \draw [densely dotted] ({axis cs:1.035797,0}|-{rel axis cs:0,0}) -- ({axis cs:1.035797,0}|-{rel axis cs:0,1});
                                \draw [densely dotted] ({axis cs:1.4142135623730951,0}|-{rel axis cs:0,0}) -- ({axis cs:1.4142135623730951,0}|-{rel axis cs:0,1});
                                \draw [densely dotted] ({axis cs:2.089884158041382,0}|-{rel axis cs:0,0}) -- ({axis cs:2.089884158041382,0}|-{rel axis cs:0,1});

                                \draw [densely dotted] ({rel axis cs:0,0}|-{axis cs:0,2.39301003}) -- ({rel axis cs:1,0}|-{axis cs:0,2.39301003});
                                \draw [densely dotted] ({rel axis cs:0,0}|-{axis cs:0,2.22144147}) -- ({rel axis cs:1,0}|-{axis cs:0,2.22144147});
                                \draw [densely dotted] ({rel axis cs:0,0}|-{axis cs:0,2.33002281}) -- ({rel axis cs:1,0}|-{axis cs:0,2.33002281});
                        \end{axis}
                \end{tikzpicture}%
        \end{center}
        \caption{
                The critical covering density $d^*(\lambda)$ depending on $\lambda$ and its values at the threshold value $\lambda_2$, the global minimum $\sqrt{2}$ and the skew $\overline{\lambda}$ at which the density becomes as bad as for the square.}
        \label{fig:critical-covering-coeff}
\end{figure}
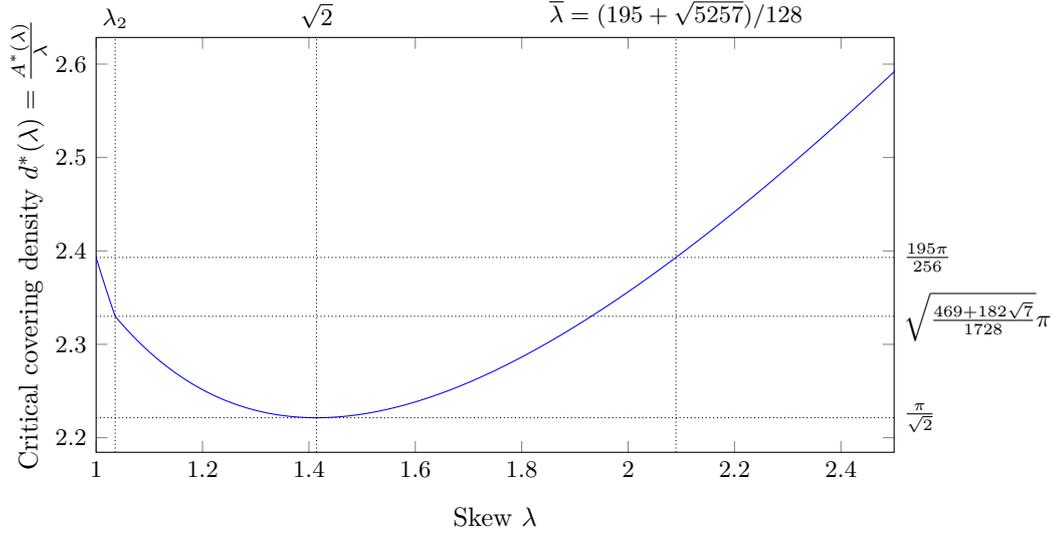

\begin{figure}
        \begin{center}
        \resizebox{.9\linewidth}{!}{\includegraphics{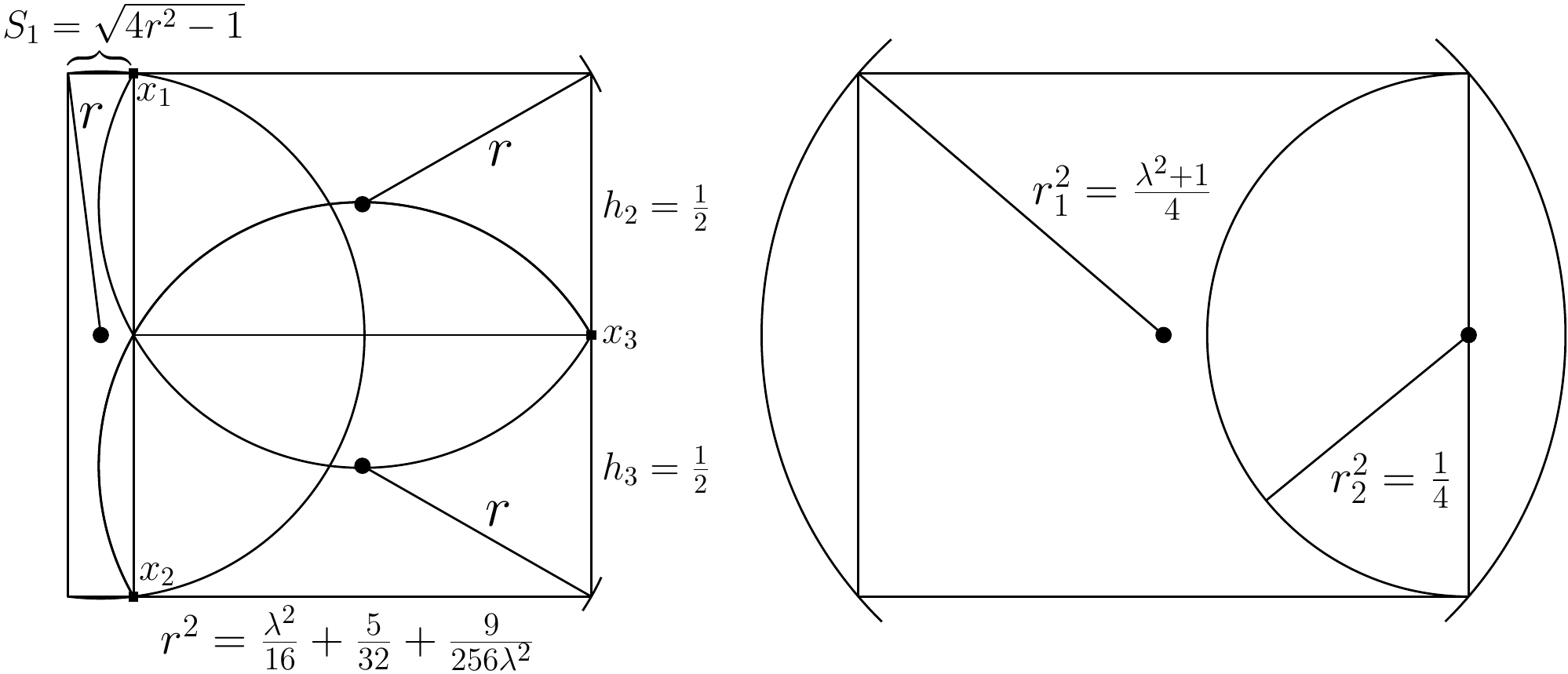}}
        \end{center}
		\caption{Worst-case configurations for small $\lambda \leq \lambda_2$ (left) and for large skew $\lambda \geq \lambda_2$ (right). \revsocg{Shrinking $r$ or $r_1$ by any $\varepsilon > 0$ in either configuration leads to an instance that cannot be covered.}}
        \label{fig:worst-cases-rectangles}
\end{figure}

\section{Preliminaries}
We are given a rectangular container $\mathcal{R}$,
which we assume w.l.o.g. to have height $1$ and some width $\lambda \geq 1$, 
which is called the \emph{skew} of $\mathcal{R}$.
For a collection $D = \{r_1,\ldots,r_n\}$ of radii $r_1 \geq r_2 \geq \cdots \geq r_n$,
we want to decide whether there is a placement of $n$ closed disks with radii $r_1,\ldots,r_n$ on $\mathcal{R}$, 
such that every point $x \in \mathcal{R}$ is covered by at least one disk.
\revsocg{Because we are only given radii and not center points, in a slight abuse of notation, we identify the disks with their radii and use $r_i$ to refer to both the disk and the radius.}

For any set $D$ of disks, the \emph{total disk area} is $A(D) \coloneqq \pi\sum_{r \in D} r^2$.
The \emph{weight} of a disk of radius $r$ is $r^2$, and $W(D) \coloneqq \frac{A(D)}{\pi}$ is the \emph{total weight} of $D$.
For any rectangle $\mathcal{R}$, the \emph{critical covering area}
$A^*(\mathcal{R})$ of $\mathcal{R}$ is the minimum value for which any set $D$
of disks with total area at least $A(D) \geq A^*(\mathcal{R})$ can cover
$\mathcal{R}$.  The \emph{critical covering weight} of $\mathcal{R}$ is
$W^*(\mathcal{R}) \coloneqq \frac{A^*(\mathcal{R})}{\pi}$.
For $\lambda \geq 1$, we define $A^*(\lambda) \coloneqq A^*(\mathcal{R})$ and $W^*(\lambda) \coloneqq W^*(\mathcal{R})$ for 
a $\lambda \times 1$  rectangle 
$\mathcal{R}$. 

For a placement $\mathcal{P}$ of the disks in $D$ \revsocg{fully} covering some area $A$,
the \emph{covering coefficient} of $\mathcal{P}$ is the ratio $\frac{W(D)}{A}$.
For $\lambda \geq 1$, the amount $E^*(\lambda) \coloneqq \frac{W^*(\lambda)}{\lambda}$ of total disk weight 
per unit of rectangle area that is necessary for 
guaranteeing a possible covering is the \emph{(critical) covering coefficient} of $\lambda$.
Analogously, $d^*(\lambda) \coloneqq \frac{A^*(\lambda)}{\lambda}$ is the \emph{(critical) covering density} of $\lambda$.

For proving our result, we use \textsc{Greedy Splitting}
for partitioning a collection of disks into two parts whose weight differs by at most the weight of the
smallest disk in the heavier part:
After sorting the disks by decreasing radius,
we start with two empty lists and continue to place the next disk in the list with smaller total weight.

\section{High-Level Description}\label{sec:high-level}%
Now we present and describe our main result: a theorem that characterizes the worst case for covering rectangles with disks.
This theorem gives a closed-form solution for the \emph{critical covering area} $A^*(\lambda)$ for any $\lambda \geq 1$;
in other words, for any given rectangle $\mathcal{R}$, we determine the total disk area that is \revsocg{(1) sometimes necessary and (2) always sufficient} to cover $\mathcal{R}$.

\begin{theorem}
	Let $\lambda \geq 1$ and let $\mathcal{R}$ be a rectangle of dimensions $\lambda \times 1$.
	Let \[\lambda_2 = \sqrt{\frac{\sqrt{7}}{2}-\frac{1}{4}} \approx 1.035797\ldots\text{, and }
		A^*(\lambda) = \begin{cases}
			3\pi\left(\frac{\lambda^2}{16} + \frac{5}{32} + \frac{9}{256\lambda^2}\right), & \text{if $\lambda < \lambda_2$,}\\
			\pi\frac{\lambda^2 + 2}{4}, & \text{otherwise.}
		\end{cases}
	\]
	\begin{enumerate}
		\item[\textup{(1)}] \revsocg{For any $a < A^*(\lambda)$, there is a set $D^-$ of disks with
			  $A(D^-) = a$ that cannot cover $\mathcal{R}$.}
		\item[\textup{(2)}] \revsocg{Let $D = \{ r_1,\ldots,r_n \} \subset \mathbb{R}$, $r_1 \geq r_2 \geq \ldots \geq r_n > 0$
			  be any collection of disks identified by their radii.
			  If $A(D) \geq A^*(\lambda)$, then $D$ can cover $\mathcal{R}$.}
	\end{enumerate}
	\label{thm:theorem-rectangle-covering-squares-tight}
	\label{thm:mainRectangles}
\end{theorem}
The critical covering area does not depend linearly on the area $\lambda$ of the rectangle; it also depends on the rectangle's skew.
Fig.~\ref{fig:critical-covering-coeff} shows a plot of the dependency of the covering density $d(\lambda)$ on $\lambda$.
In the following, to simplify notation, we factor out $\pi$ if possible; instead of working with the areas $A(D)$ or $A^*(\lambda)$ of the disks, we use their \emph{weight}, i.e., their area divided by $\pi$.
Similarly, we work with the covering coefficient $E^*(\lambda)$ instead of the density $d^*(\lambda)$; a lower covering coefficient corresponds to a more efficient covering.

As shown in Fig.~\ref{fig:critical-covering-coeff}, the critical covering coefficient \(E^*(\lambda)\) is monotonically decreasing from $\lambda = 1$ to $\sqrt{2}$ and monotonically increasing for $\lambda > \sqrt{2}$.
For a square, $E^*(1)=\frac{195}{256}$; the point $\lambda > 1$ for which the covering coefficient becomes as bad as for the square is $\overline{\lambda} \coloneqq \frac{195+\sqrt{5257}}{128} \approx 2.08988\ldots$; for all $\lambda \leq \overline{\lambda}$, the covering coefficient is at most $\frac{195}{256}$.

\subsection{Proof Components}
The proof of Theorem~\ref{thm:mainRectangles} uses a number of 
components. First is a lemma that describes the worst-case configurations and shows 
tightness\revsocg{, i.e., claim (1),} of Theorem~\ref{thm:mainRectangles} for all $\lambda$.
\begin{restatable}{lemma}{lemmaworstcasesrectangles}
	\label{lem:worst-cases-rectangles}
	\revsocg{Let $\lambda \geq 1$ and let $\mathcal{R}$ be a rectangle of dimensions $\lambda\times 1$.}
		(1)~Two disks of weight $r_1^2 = \frac{\lambda^2+1}{4}$ and $r_2^2 = \frac{1}{4}$ suffice to cover $\mathcal{R}$.
		(2)~For any $\varepsilon > 0$, two disks of weight $r_1^2-\varepsilon$ and $r_2^2$ do not suffice to cover $\mathcal{R}$.
		(3)~Three identical disks of weight $r^2 = \frac{\lambda^2}{16} + \frac{5}{32} + \frac{9}{256\lambda^2}$ suffice to cover a rectangle $\mathcal{R}$ of dimensions $\lambda \times 1$.
		(4)~For $\lambda \leq \lambda_2$ and any $\varepsilon > 0$, three identical disks of weight $r^2_- \coloneqq r^2-\varepsilon$ do not suffice to cover $\mathcal{R}$.
\end{restatable}
For large $\lambda$, the critical covering coefficient $E^*(\lambda)$ of Theorem~\ref{thm:mainRectangles} becomes worse, as large disks cannot be used to cover the rectangle efficiently.
If the weight of each disk is bounded by some $\sigma \geq r_1^2$, 
we provide the following lemma achieving a better covering coefficient $E(\sigma)$ with $E^*(\overline{\lambda}) \leq E(\sigma) \leq E^*(\lambda)$.
This coefficient is independent of the skew of $\mathcal{R}$.
\begin{restatable}{lemma}{lemmasizeboundlarge}
	\label{lem:size-bound-large}
	\revsocg{
		Let $\hat{\sigma} \coloneqq \frac{195\sqrt{5257}}{16384} \approx 0.8629$.
		Let $\sigma \geq \hat{\sigma}$ and $E(\sigma) \coloneqq \frac{1}{2}\sqrt{\sqrt{\sigma^2+1}+1}$.
		Let $\lambda \geq 1$ and $D = \{r_1,\ldots,r_n\}$ be any collection of disks with $\sigma \geq r_1^2 \geq \ldots \geq r_n^2$ and $W(D) = \sum\limits_{i=1}^{n}r_i^2 \geq E(\sigma)\lambda$.
		Then $D$ can cover a rectangle $\mathcal{R}$ of dimensions $\lambda \times 1$.
	}
\end{restatable}
\revsocg{
	Note that $E(\hat{\sigma}) = \frac{195}{256}$, i.e. the best covering coefficient established by Lemma~\ref{lem:size-bound-large},
	coinciding with the critical covering coefficient of the square established by Theorem~\ref{thm:mainRectangles}.
	Thus, we can cover any rectangle with covering coefficient $\frac{195}{256}$ if the largest disk satisfies $r_1^2 \leq \hat{\sigma}$.
}

The final component is the following Lemma~\ref{lem:rectanglesSB},
which also gives a better covering coefficient 
\revsocg{if the size of the largest disk is bounded.}
The bound required for Lemma~\ref{lem:rectanglesSB} is smaller than for Lemma~\ref{lem:size-bound-large};
in return, the covering coefficient that Lemma~\ref{lem:rectanglesSB} yields is better.
Note that the result of Lemma~\ref{lem:rectanglesSB} is not tight.
\newcommand{\rectanglesSBEff}{0.61}
\newcommand{\rectanglesSBRB}{0.375}
\newcommand{\rectanglesSBWB}{0.140625}
\begin{restatable}{lemma}{lemmarectanglessb}
	\label{lem:rectanglesSB}
	Let $\lambda \geq 1$ and let $\mathcal{R}$ be a rectangle of dimensions $\lambda \times 1$.
	Let $D = \{ r_1, \ldots, r_n \}$, $\rectanglesSBRB \geq r_1 \geq \ldots \geq r_n > 0$ be a collection of disks.
	If $W(D) \geq \rectanglesSBEff\lambda$, or equivalently $A(D) \geq \rectanglesSBEff\pi\lambda \approx 1.9164\lambda$, then $D$ suffices to cover $\mathcal{R}$.
\end{restatable}

\subsection{Proof Overview}
The proofs of Theorem~\ref{thm:mainRectangles} and Lemmas~\ref{lem:size-bound-large}~and~\ref{lem:rectanglesSB} work by induction on the number of disks.
For proving Lemma~\ref{lem:size-bound-large} for $n$ disks, we use Theorem~\ref{thm:mainRectangles} for $n$ disks.
For proving Theorem~\ref{thm:mainRectangles} for $n$ disks, we use Lemma~\ref{lem:rectanglesSB} for $n$ disks; Lemma~\ref{lem:size-bound-large} is only used for fewer than $n$ disks; \revsocg{see Fig.~\ref{fig:proof_flow}.}
\begin{figure}
	\centering
	\resizebox{.5\linewidth}{!}{\includegraphics{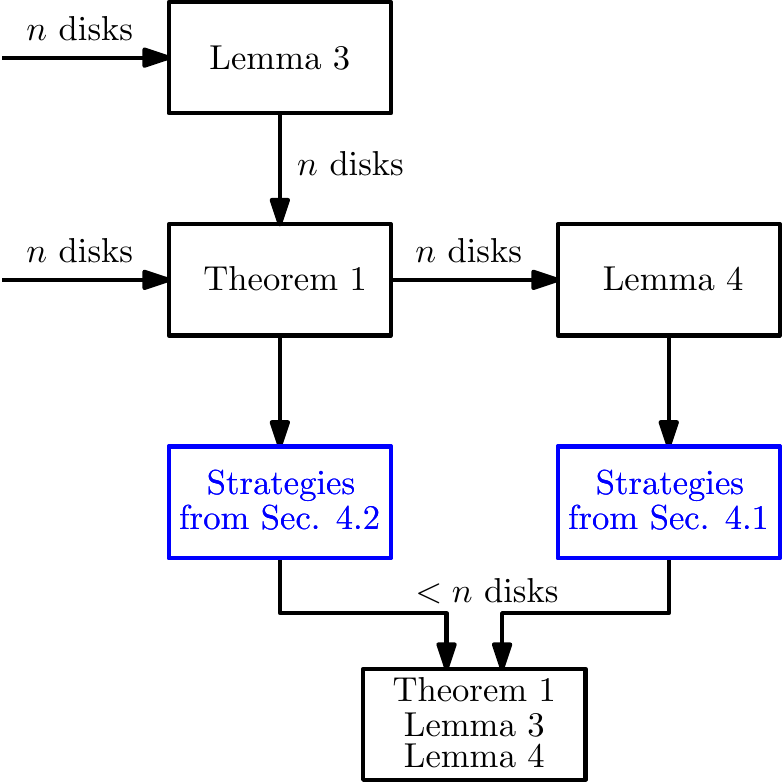}}
	\caption{The inductive structure of the proof; the blue parts are computer-aided.}
	\label{fig:proof_flow}
\end{figure}
For proving Lemma~\ref{lem:rectanglesSB} for $n$ disks, we only use Theorem~\ref{thm:mainRectangles} and Lemma~\ref{lem:size-bound-large} for fewer than $n$ disks.
Therefore, there are no cyclic dependencies in our argument; however, we have to perform the induction for Theorem~\ref{thm:mainRectangles} and Lemmas~\ref{lem:size-bound-large}~and~\ref{lem:rectanglesSB} simultaneously.

\mparagraph{Routines} The proofs of Theorem~\ref{thm:mainRectangles} and Lemma~\ref{lem:rectanglesSB} are constructive; they are based on an efficient recursive algorithm that uses a set of simple \emph{routines}.
\revsocg{We go through the list of rountines in some fixed order.
For each routine, we check a sufficient criterion for the routine to work.
We call these criteria \emph{success criteria}.
They only depend on the total available weight and a constant number of largest disks.
If we cannot guarantee that a routine works by its success criterion, we simply disregard the routine;
this means that our algorithm does not have to backtrack.
We prove that, regardless of the distribution of the disks' weight,
at least one success criterion is met, implying that we can always apply at least one routine.
The number of routines and thus success criteria is large;
this is where the need for automatic assistance comes from.}

\mparagraph{Recursion}
\revsocg{Typical routines are recursive;
they consist of splitting the collection of disks into smaller parts,
splitting the rectangle accordingly, and recursing, or recursing after fixing the position
of a constant number of large disks.

In the entire remaining proof, the criterion we use to guarantee that recursion works is as follows.
Given a collection $D' \subsetneq D$ and a rectangular region $\mathcal{R}' \subsetneq \mathcal{R}$,
we check whether the preconditions of Theorem~\ref{thm:mainRectangles} or Lemma~\ref{lem:size-bound-large}~or~\ref{lem:rectanglesSB} are met
after appropriately scaling and rotating $\mathcal{R}'$ and the disks.
Note that, due to the scaling, the radius bounds of Lemmas~\ref{lem:size-bound-large}~and~\ref{lem:rectanglesSB} depend on the length of the shorter side of $\mathcal{R}'$.
In some cases where we apply recursion, we have more weight than necessary to satisfy the weight requirement for recursion according to Lemma~\ref{lem:size-bound-large}~or~\ref{lem:rectanglesSB},
but these lemmas cannot be applied due to the radius bound.
In that case, we also check whether we can apply Lemma~\ref{lem:size-bound-large}~or~\ref{lem:rectanglesSB} after increasing the length of the shorter side of $\mathcal{R}'$ as far as the disk weight allows.
This excludes the case that we cannot recurse on $\mathcal{R}'$ due to the radius bound, but there is some $\mathcal{R}'' \supset \mathcal{R}'$ on which we could recurse.
}

\subsection{Interval Arithmetic}\label{sec:interval-arithmetic}
We use interval arithmetic to prove that there always is a successful routine.
In interval arithmetic, operations like addition, multiplication or taking a square root are performed on intervals $[a,b] \subset \mathbb{R}$ instead of numbers.
Arithmetic operations on intervals are derived from their real counterparts as follows.
The result of an operation $\circ$ in interval arithmetic is \[[a_1,b_1] \circ [a_2,b_2] \coloneqq \left[\min\limits_{x_1 \in [a_1,b_1], x_2 \in [a_2,b_2]}x_1 \circ x_2,\max\limits_{x_1 \in [a_1,b_1], x_2 \in [a_2,b_2]} x_1 \circ x_2\right].\]
Thus, the result of an operation is the smallest interval that contains all possible results of $x \circ y$ for $x \in [a_1,b_1], y \in [a_2,b_2]$.
Unary operations are defined analogously.
For square roots, division or other operations that are not defined on all of $\mathbb{R}$, a result is undefined 
iff the input interval(s) contain values for which the real counterpart of the operation is undefined.

\mparagraph{Truth values} In interval arithmetic, inequalities such as $[a_1,b_1] \leq [a_2,b_2]$ can have three possible truth values.
An inequality can be \emph{definitely true}; this means that the inequality holds for any value of $x \in [a_1,b_1], y \in [a_2,b_2]$.
In the example $[a_1,b_1] \leq [a_2,b_2]$, this is the case if $b_1 \leq a_2$.
An inequality can be \emph{indeterminate}; this means that there are some values $x,x' \in [a_1,b_1], y,y' \in [a_2,b_2]$ such that the inequality holds for $x,y$ and does not hold for $x',y'$.
In the example $[a_1,b_1] \leq [a_2,b_2]$, this is the case if $a_1 \leq b_2$ and $b_1 > a_2$.
Otherwise, an inequality is \emph{definitely false}.
An inequality that is either \emph{definitely true} or \emph{indeterminate} is called \emph{possibly true}; an inequality that is either \emph{indeterminate} or \emph{definitely false} is called \emph{possibly false}.
These truth values can also be interpreted as intervals $[0,0],[0,1],[1,1]$.

\mparagraph{Using interval arithmetic} \revsocg{We apply interval arithmetic in our proof as follows.
Recall that for each routine, we have a \emph{success criterion}.
These criteria only consider $\lambda \geq 1$ and the largest $k \in \mathcal{O}(1)$ disks $r_1 \geq \cdots \geq r_k$ as well as the remaining weight
$R_{k+1} \coloneqq \sum_{i=k+1}^{n} r_i^2$, which can be computed from $\lambda$ and $r_1,\ldots,r_k$, assuming w.l.o.g.\ that the total disk weight $W(D)$ is exactly $W^*(\lambda)$.

If we can manually perform induction base and induction step of our result for all $\lambda \geq \hat{\lambda}$ for some finite value $\hat{\lambda}$,
we can also provide an upper bound $\hat{r}_1$ for $r_1$ such that all cases that remain to be considered (in our induction base and induction step) correspond to a point in the $(k+1)$-dimensional space $\Psi$ given by 
\[\lambda \in [1,\hat{\lambda}], r_1 \in [0,\hat{r}_1],r_2 \in [0,r_1],\ldots,r_k \in [0,r_{k-1}],\sum\limits_{i=1}^kr_i^2 \leq W^*(\lambda).\]
This is due to the fact that there is nothing to prove if $r_1$ can cover $\mathcal{R}$ on its own; $r_1$ can have no more than the total disk weight $W(D)$ and $r_k \leq \cdots \leq r_2 \leq r_1$.
Furthermore, observe that the induction base is just a special case with $r_i = r_{i+1} = \cdots = 0$ for some $1 < i \leq k$.

This allows subdividing (a superset of) $\Psi$ into a large finite number of \emph{hypercuboids}
by splitting the range of each of the variables $\lambda,r_1,\ldots,r_k$ into a number of smaller intervals.
For each hypercuboid, we then use interval arithmetic to verify that there is a routine whose success criterion is met.
If we find such a routine, we have eliminated all points in that hypercuboid from further consideration.
Hypercuboids for which this does not succeed are called \emph{critical} and must be resolved manually;
note that, in particular, hypercuboids containing (tight) worst-case configurations cannot be handled by interval arithmetic.
The restriction to critical hypercuboids makes the overall analysis feasible, 
while a manual analysis of the entire space is impractical due to the large number of routines and variables.
}

\mparagraph{Implementation} \revsocg{
We implemented the subdivision outlined above and all success criteria of our routines using interval arithmetic\footnote{The source code of the implementation is available online:\\\url{https://github.com/phillip-keldenich/circlecover} .}.
Because most of our success criteria use the squared radii $r_i^2$ instead of the radii $r_i$, we use $\lambda$ and $r_i^2$ instead of $r_i$ as variables.
Moreover, for efficiency reasons, instead of the simple grid-like subdivision outlined above, we use a search-tree-like subdivision strategy where we begin by subdividing the range of $\lambda$, continue by subdividing $r_1^2$, followed by $r_2^2$, and so on.
Whenever a success criterion only needs the first $i < k$ disks, we can check this criterion farther up in the tree, thus potentially avoiding visits to large parts of the search tree; see Fig.~\ref{fig:ia_overview} for a sketch of this procedure.
Even with this pruning in place, the number of hypercuboids that we have to consider is still very large; this is a result of the fact that, depending on the claim at stake, we have $5$ or even $8$ dimensions.
Therefore, we implemented the checks for our success criteria on a CUDA-capable GPU to perform them in a massively parallel fashion.

Moreover, to provide a finer subdivision where necessary, we run our search in several \emph{generations} (our proof uses 11 generations).
Each generation yields a set of critical hypercuboids that could not be handled automatically.
After each generation, for each subinterval of $\lambda$, we collect all critical hypercuboids and merge those for which the
$r_1^2$-subintervals are overlapping by taking the smallest hypercuboid containing all points in the merged hypercuboids.
This procedure typically yields only 1-3 hypercuboids per subinterval of $\lambda$.
The next generation is run on each of these, starting with the bounds given by these hypercuboids.
\begin{figure}[p]%
	\centering
	\resizebox{.99\linewidth}{!}{\includegraphics{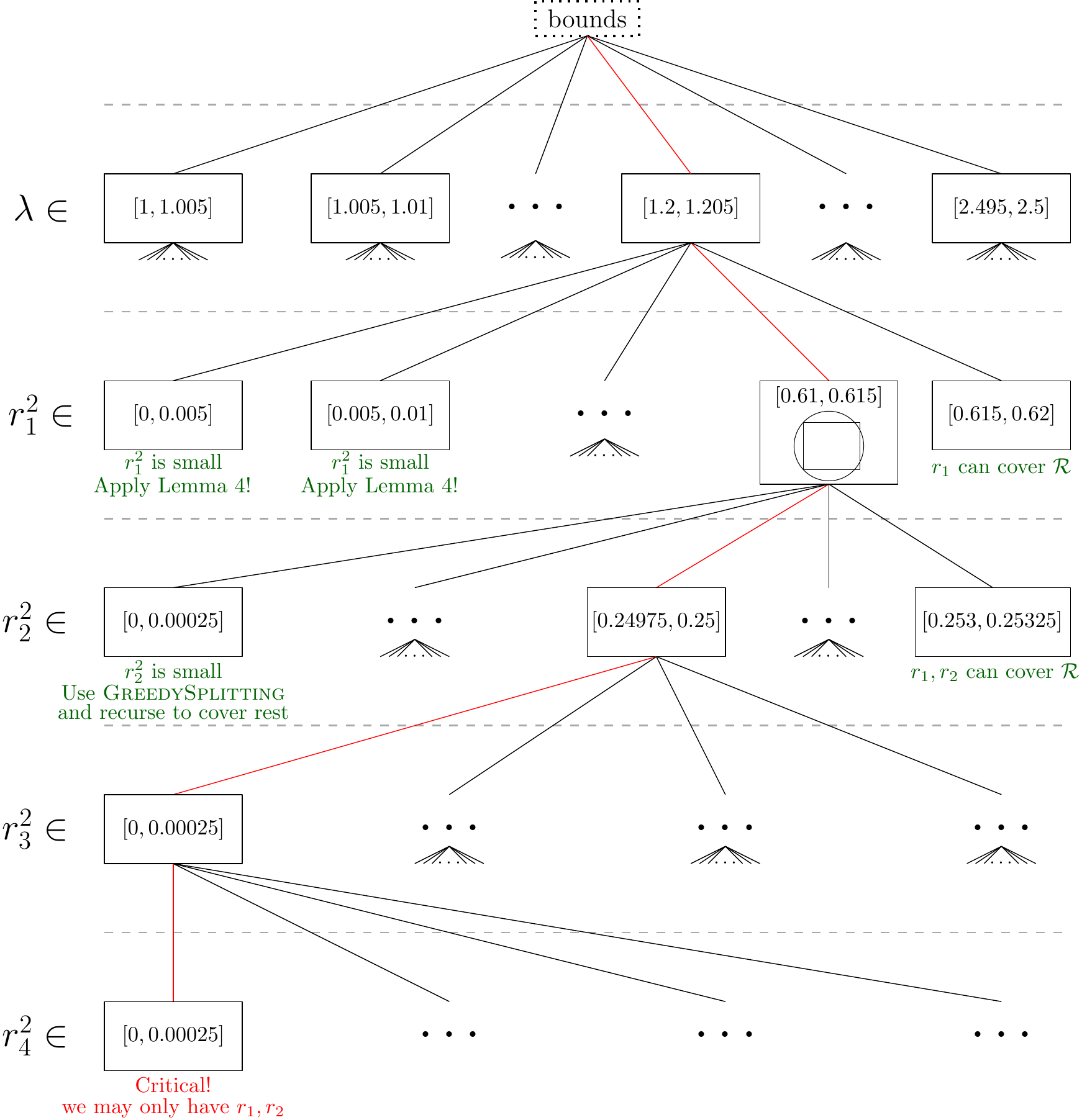}}
	\caption{\revsocg{Sketch of our interval arithmetic-based search procedure.
			 The red edges denote a path leading to a critical cuboid containing a tight two-disk worst-case configuration.
			 Green text indicates that the children of the corresponding node do not have to be considered.}}
	\label{fig:ia_overview}
\end{figure}
}

\mparagraph{Numerical issues} When performing computations on a computer with limited-precision floating-point numbers instead of real numbers, there can be rounding errors, underflow errors and overflow errors.
Our implementation of interval arithmetic performs all operations using appropriate rounding modes; this technique is also used by the implementation of interval arithmetic in the well-known Computational Geometry Algorithms Library (CGAL)~\cite{cgal}.
This means that any operation $\circ$ on two intervals $A,B$ yields an interval $I \supseteq A \circ B$ to ensure that the result of any operation contains all values that are possible outcomes of $x \circ y$ for $x,y \in A,B$.
This guarantees soundness of our results in the presence of numerical errors.

\FloatBarrier
\section{Proof Structure}\label{sec:proof-structure}
In this section, we give an overview of the structure of the proofs of Theorem~\ref{thm:mainRectangles} and
Lemmas~\ref{lem:worst-cases-rectangles},~\ref{lem:size-bound-large}~and~\ref{lem:rectanglesSB}.
\ifthenelse{\boolean{applemmaworstcasesrectangles}}{%
	For the proof of Lemma~\ref{lem:worst-cases-rectangles}, we refer to the full version~\cite{fullversion} of our paper.
}{%
	We prove Lemma~\ref{lem:worst-cases-rectangles} in Section~\ref{sec:proof-lemma-worst-cases-rectangles} using a straightforward argument and simple case analysis.
}%
Lemma~\ref{lem:size-bound-large} is proven in Section~\ref{sec:proof-lemma-size-bound-large} using a simple recursive algorithm;
basically, we show that we can always split the disks using \textsc{Greedy Splitting}, split the rectangle accordingly, and recurse using Theorem~\ref{thm:mainRectangles}.
The proofs of Theorem~\ref{thm:mainRectangles} and Lemma~\ref{lem:rectanglesSB} involve a larger number of routines and
make use of an automatic prover based on interval arithmetic as described in Section~\ref{sec:interval-arithmetic}.

\subsection{Proof Structure for Lemma~\ref{lem:rectanglesSB}}
Proving Lemma~\ref{lem:rectanglesSB} means proving that, for any skew $\lambda$,
any collection $D$ of disks of radius $r_1 \leq \rectanglesSBRB$ and with total weight $W(D) = E\lambda$ suffices to cover $\mathcal{R}$,
where $E = \rectanglesSBEff$ is the covering coefficient guaranteed by Lemma~\ref{lem:rectanglesSB}.
We first reduce the number of cases that we have to consider in our induction base and induction step to a finite number.
As described in Section~\ref{sec:interval-arithmetic}, this requires handling the case of arbitrarily large skew $\lambda$.
Finding a bound $\hat{\lambda}$ and reducing Lemma~\ref{lem:rectanglesSB} for $\lambda \geq \hat{\lambda}$ to the case of $\lambda < \hat{\lambda}$ yields bounds for $\lambda$ and $r_1,\ldots,r_k$ that allow a reduction to finitely many cases using interval arithmetic.
\begin{lemma}\label{lem:rectanglesSB-handling-large-skew}
	Let $\hat{\lambda} = 2.5$.
	Given disks $D$ according to the preconditions of Lemma~\ref{lem:rectanglesSB} and $\lambda \geq \hat{\lambda}$, we can cover $\mathcal{R}$ using a simple recursive routine.
\end{lemma}
\begin{proof}
	\revsocg{
	The routine works as follows.
	We build a list of disks $D_1$ by adding disks in decreasing order of radius until $W\left(D_1\right) \geq E$.
	Due to the radius bound, this procedure always stops before all disks are used, i.e., $D_1 \subsetneq D$.
	Let $D_2 \coloneqq D\setminus D_2$ be the remaining disks.
	We then place a vertical rectangular strip $\mathcal{R}_1$ of height $1$ and width
	$\beta_{\mathcal{R}_1} \coloneqq \frac{W\left(D_1\right)}{E} \geq 1$ at the left side of $\mathcal{R}$.
	By induction, we can recurse on $\mathcal{R}_1$ using Lemma~\ref{lem:rectanglesSB} and the disks from $D_1$,
	because both side lengths are at least $1$ and the efficiency we require is \emph{exactly} $E$.
	Note that, due to adapting the width $\beta_{\mathcal{R}_1}$ according to the actual weight $W(D_1)$, we actually achieve an efficiency of $E$;
	in other words, there is no \emph{waste} of disk weight.
	This means that we also require an efficiency of exactly $E$ on the remaining rectangle $\mathcal{R}_2 \coloneqq \mathcal{R}\setminus\mathcal{R}_1$.
	Therefore, provided that the largest disk in $D_2$ satisfies the size bound of Lemma~\ref{lem:rectanglesSB}, we can inductively apply Lemma~\ref{lem:rectanglesSB} to $\mathcal{R}_2$ and $D_2$ and are done.
	This can be guaranteed by proving that the shorter side of $\mathcal{R}_2$ is at least $1$ as well.
	We have $W(D_1) \leq E+r_1^2 \leq E + \rectanglesSBRB^2$ which implies $\beta_{\mathcal{R}_1} \leq 1 + \frac{\rectanglesSBRB^2}{E} < 1.5$;
	therefore, $\lambda \geq 2.5$ ensures that the width of $\mathcal{R}_2$ is at least $1$.
	}
\end{proof}

\revsocg{As outlined in Section~\ref{sec:interval-arithmetic}, the remainder of the proof of Lemma~\ref{lem:rectanglesSB}
is based on a list of simple covering routines and their success criteria.
We prove that there always is a working routine in that list using an automatic prover based on interval arithmetic,
as described in Section~\ref{sec:interval-arithmetic}.
This automatic prover considers the 8-dimensional space spanned by the variables $\lambda$ and $r_1^2,\ldots,r_7^2$ and
subdivides it into a total of more than $2^{46}$ hypercuboids in order to prove that there always is a working routine,
i.e., no critical hypercuboids remain to be analyzed manually; this only works because the result of Lemma~\ref{lem:rectanglesSB} is not tight.}

In the following, we give a brief description of the routines that we use.
\ifthenelse{\boolean{appproofsizebounded}}{
	\sbnewsubsec{strat:rectanglesSB-eurs}%
	\sbnewsubsec{strat:rectanglesSB-bas}%
	\sbnewsubsec{strat:rectanglesSB-wb}%
	\sbnewsubsec{strat:rectanglesSB-r1c}%
	\sbnewsubsec{strat:rectanglesSB-r12oc}%
	\sbnewsubsec{strat:rectanglesSB-3ld}%
	\sbnewsubsec{strat:rectanglesSB-4ld}%
	\sbnewsubsec{strat:rectanglesSB-5ld}%
	\sbnewsubsec{strat:rectanglesSB-6ld}%
	\sbnewsubsec{strat:rectanglesSB-7ld}%
	Due to space constraints, for a detailed description of the routines, we refer to the full version of our paper~\cite{fullversion}.
}{
	The routines are described in detail in Section~\ref{sec:proof-rectanglesSB}. 
}%

\mparagraph{Recursive splitting}
Routines~(\sbstratref{strat:rectanglesSB-eurs}{1})~and~(\sbstratref{strat:rectanglesSB-eurs}{2}) work by splitting $D$ into two parts, splitting $\mathcal{R}$ accordingly, and recursing on the two sub-rectangles.
This split is either performed as balanced as possible using \textsc{Greedy Splitting}, or in an unbalanced manner; in the latter case, we choose an unbalanced split to accommodate large disks that violate the radius bound of Lemma~\ref{lem:rectanglesSB} w.r.t.\ a rectangle of half the width of $\mathcal{R}$.

\mparagraph{Building a strip}
Routine~(\sbstratref{strat:rectanglesSB-bas}{1}) works by either covering the left or the bottom side of a rectangular strip $\mathcal{R}$; see Fig.~\ref{fig:rectangleSB-vertical-strip}.
This strip uses a subset of the largest six disks and tries several configurations for placing the disks.
The remaining area is covered by recursion.
\begin{figure}%
	\centering
	\resizebox{.85\linewidth}{!}{\includegraphics{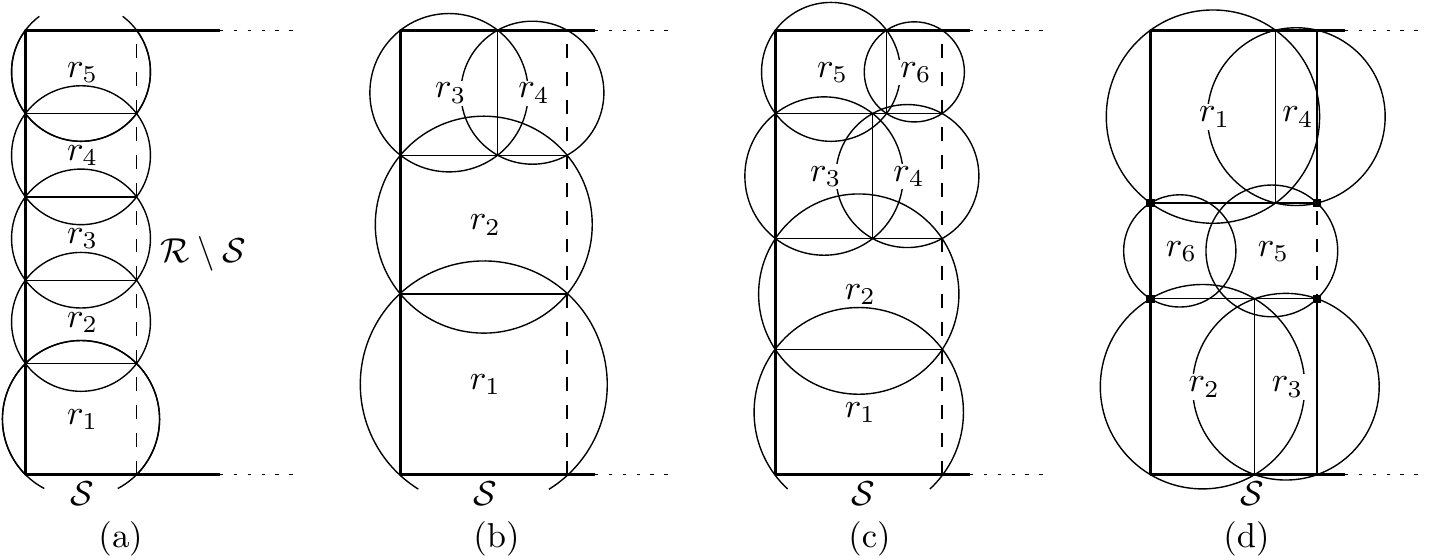}}
	\caption{Some placements considered by Routine~\sbstratref{strat:rectanglesSB-bas}{1} to build a vertical strip; horizontal strips are analogous.
	(a) Simply stacking a subset $T$ of the six largest disks on top of each other.
	(b) Stacking $r_1,r_2$ on top of each other, and placing $r_3,r_4$ horizontally next to each other on top.
	(c) Same as (b), but with an additional row built from $r_5,r_6$.
	(d) Building two rows at the top and the bottom consisting of $r_1,r_4$ and $r_2,r_3$, and covering the remaining region by $r_5,r_6$.
	The points on the boundary defining the position of $r_5$ and $r_6$ are marked by squares.
	Note that $r_5$ and $r_6$ are not big enough to cover the entire rectangular area between the top and the bottom row.
	}
	\label{fig:rectangleSB-vertical-strip}
\end{figure}%

\mparagraph{Wall building} Routines~(\sbstratref{strat:rectanglesSB-wb}{1})~and~(\sbstratref{strat:rectanglesSB-r1c}{1}) are based on the idea of covering a rectangular strip of fixed length $\ell$ and variable width $b$ with covering coefficient exactly $E$.
We call this \emph{wall building}.
To achieve this covering coefficient, we stack disks of similar size on top of (or horizontally next to) each other; each disk placed in this way covers a rectangle of variable height, but width $b$.
We provide sufficient conditions \ifthenelse{\boolean{appproofsizebounded}}{for}{(see Lemma~\ref{lem:rectangleSB-wall-building}) for}
this procedure to result in a successful covering of a strip of length $\ell$.
Routine~(\sbstratref{strat:rectanglesSB-wb}{1}) uses this idea to build a column of stacked disks at the left side of $\mathcal{R}$; see Fig.~\ref{fig:rectangleSB-wall-building}.
\begin{figure}%
		\centering
		\resizebox{.75\linewidth}{!}{\includegraphics{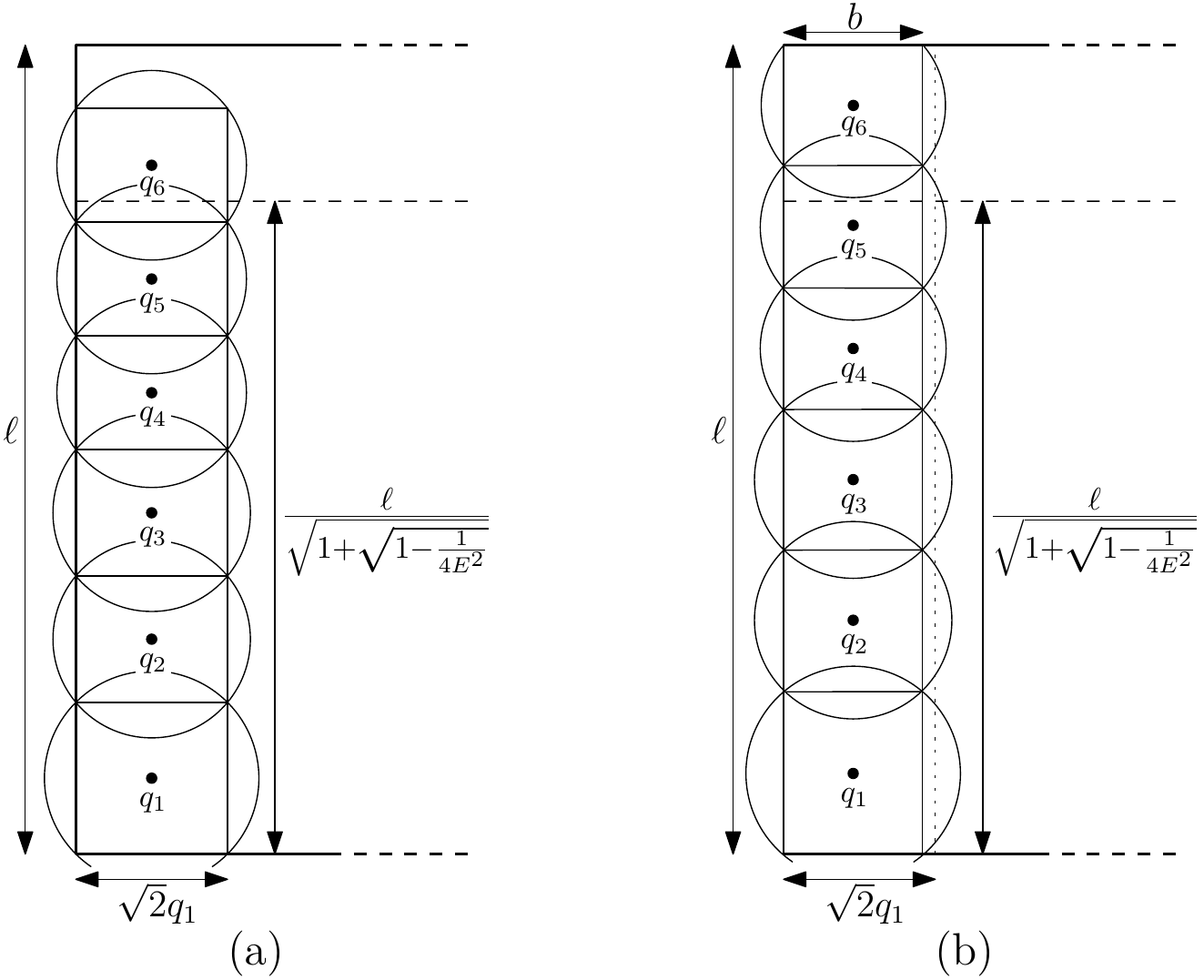}}
		\caption{The wall-building procedure.
		(a) Using an initial guess of $b = \sqrt{2}q_1$ as width, where $q_1$ is the largest disk that we use, we stack disks until they exceed a certain fraction of the length $\ell$.
		(b) We decrease $b$ until the disks exactly cover a strip of length $\ell$.}
		\label{fig:rectangleSB-wall-building}
\end{figure}%
Routine~(\sbstratref{strat:rectanglesSB-r1c}{1}) uses this idea by placing $r_1$ in the bottom-left corner of $\mathcal{R}$ and filling the area above $r_1$ with horizontal rows of disks; see Fig.~\ref{fig:rectanglesSB-r1-wall-building}.
\begin{figure}%
	\centering
	\resizebox{.8\linewidth}{!}{\includegraphics{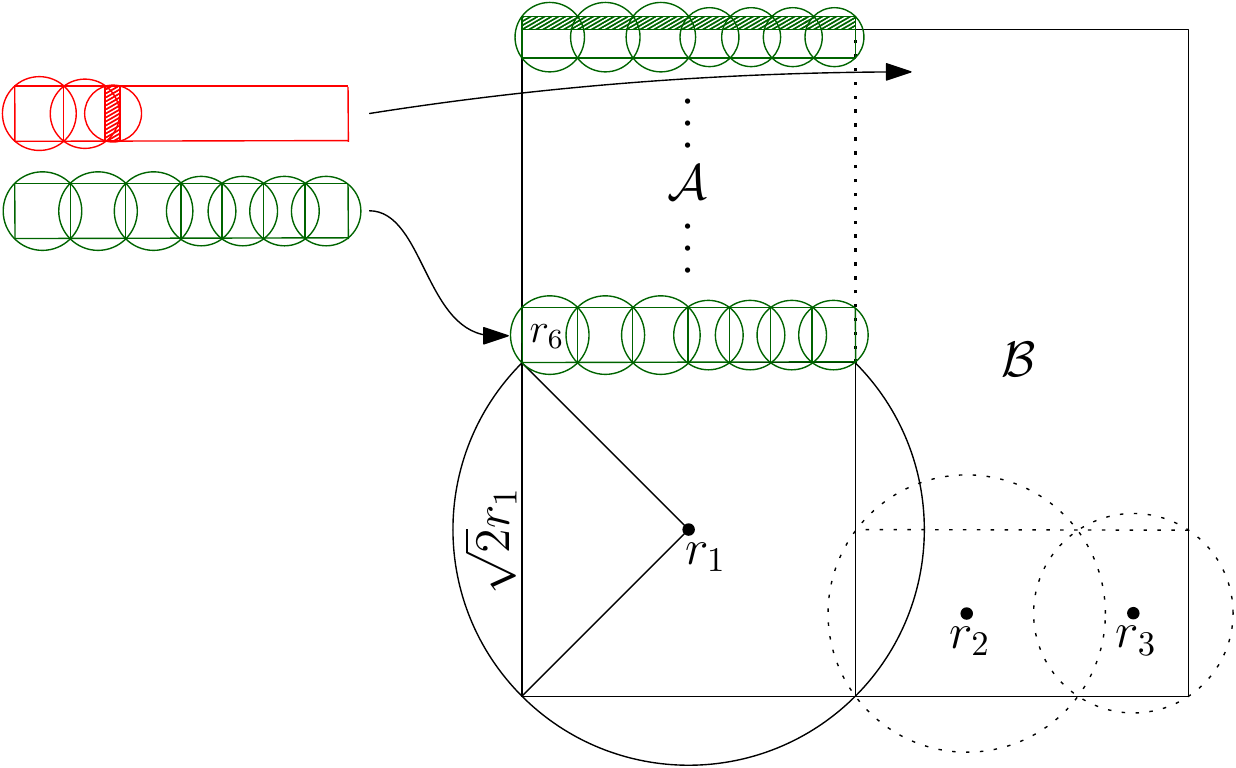}}
	\caption{
		Routine~\sbstratref{strat:rectanglesSB-r1c}{1} places $r_1$ in the bottom-left corner and tries to cover $\mathcal{A}$ using either recursion or wall building.
		In the latter case, whenever the disk radius drops too much while building a row of length $\ell = \sqrt{2}r_1$, we move the disks constituting this incomplete row to $\mathcal{B}$ (red).
		Otherwise, a complete row is built (green) and we continue with the next row.
		This process stops once the entire area $\mathcal{A}$ is covered, including some potential overhead (shaded green region).
		We compensate for the overhead by the area gained by placing $r_1$ covering a square.
		In case $r_2$ does not fit into $\mathcal{B}$ recursively, we try placing $r_2,r_3$ (or $r_2,r_3,r_4$) at the bottom of $\mathcal{B}$ (dotted outline).
	}
	\label{fig:rectanglesSB-r1-wall-building}
\end{figure}%
Intuitively speaking, these routines are necessary to handle cases in which there are large disks that interfere with recursion, but small disks, for which we do not know the weight distribution, significantly contribute to the total weight.

\mparagraph{Using the two largest disks}
Routine~(\sbstratref{strat:rectanglesSB-r12oc}{1}) places the two largest disks in diagonally opposite corners, each disk covering its inscribed square; see Fig.~\ref{fig:rectanglesSB-r1_r2_opposite_corners}.
The remaining area is subdivided into three rectangular regions; we cover these regions recursively, considering several ways to split the remaining disks.
\begin{figure}%
	\centering
	\resizebox{.80\linewidth}{!}{\includegraphics{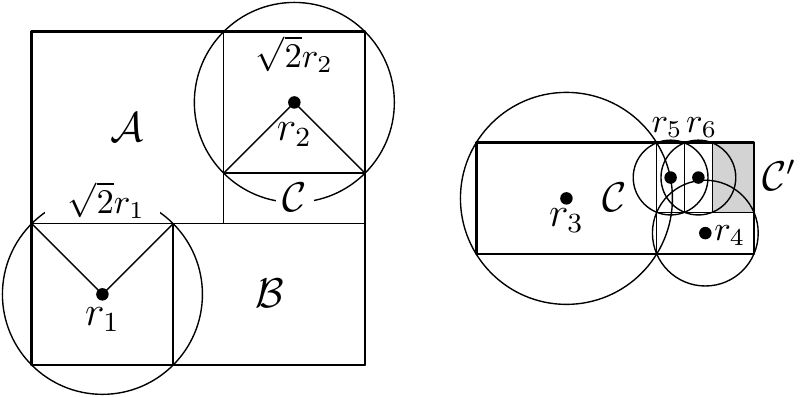}}
	\caption{The routine~\sbstratref{strat:rectanglesSB-r12oc}{1} places $r_1$ and $r_2$ in diagonally opposite corners, each covering a square.
	We cover three remaining rectangular areas $\mathcal{A}, \mathcal{B}, \mathcal{C}$ using the remaining disks (left).
	Regions $\mathcal{A}, \mathcal{B}$ and $\mathcal{C}$ are covered by recursion; we also consider using disks $r_3,\ldots,r_6$ to reduce $\mathcal{C}$ to $\mathcal{C}'$ (light gray) before recursing.}
	\label{fig:rectanglesSB-r1_r2_opposite_corners}
\end{figure}%

\mparagraph{Using the three largest disks} Routines~(\sbstratref{strat:rectanglesSB-3ld}{1})~and~(\sbstratref{strat:rectanglesSB-3ld}{2}) consider two different placements of the largest three disks as shown in Fig.~\ref{fig:rectanglesSB-l-shaped-recursion}.
\begin{figure} 
	\centering
	\includegraphics[width=.9\linewidth]{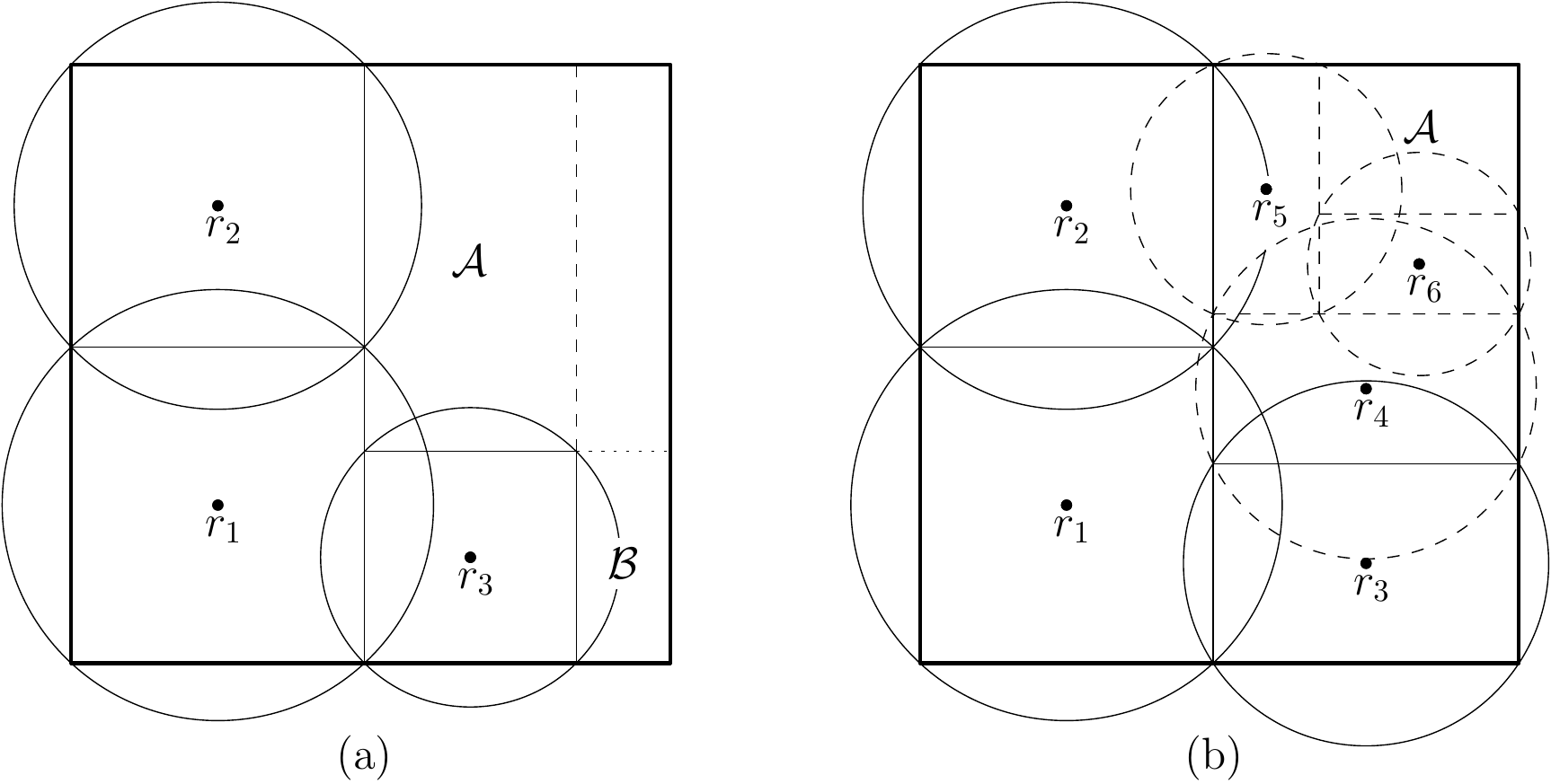}
	\caption{Two routines based on using the three largest disks. We use $r_1$ and $r_2$ to cover a vertical strip of height $1$ and maximal width.
	(a) In Routine~\sbstratref{strat:rectanglesSB-3ld}{1}, we place $r_3$ to the right of $r_1$, covering its inscribed square at the lower left corner of the remaining rectangle; the remaining region can be subdivided into two rectangles $\mathcal{A},\mathcal{B}$ in two ways (dashed and dotted line).
	(b) In Routine~\sbstratref{strat:rectanglesSB-3ld}{2}, we cover a horizontal strip of the remaining rectangle using $r_3$; we either recurse on the remaining rectangle directly or place some of the disks $r_4,r_5,r_6$ to cut off pieces of the longer side of the remaining rectangle (dashed outlines).}
	\label{fig:rectanglesSB-l-shaped-recursion}
\end{figure}

\mparagraph{Using the four largest disks} Routines~(\sbstratref{strat:rectanglesSB-4ld}{1})--(\sbstratref{strat:rectanglesSB-4ld}{3}) consider different placements of the four largest disks and recursion to cover $\mathcal{R}$; see Fig.~\ref{fig:rectanglesSB-r1_in_corner_right_34_and_2}.
\begin{figure}
	\centering
	\includegraphics[width=.60\textwidth]{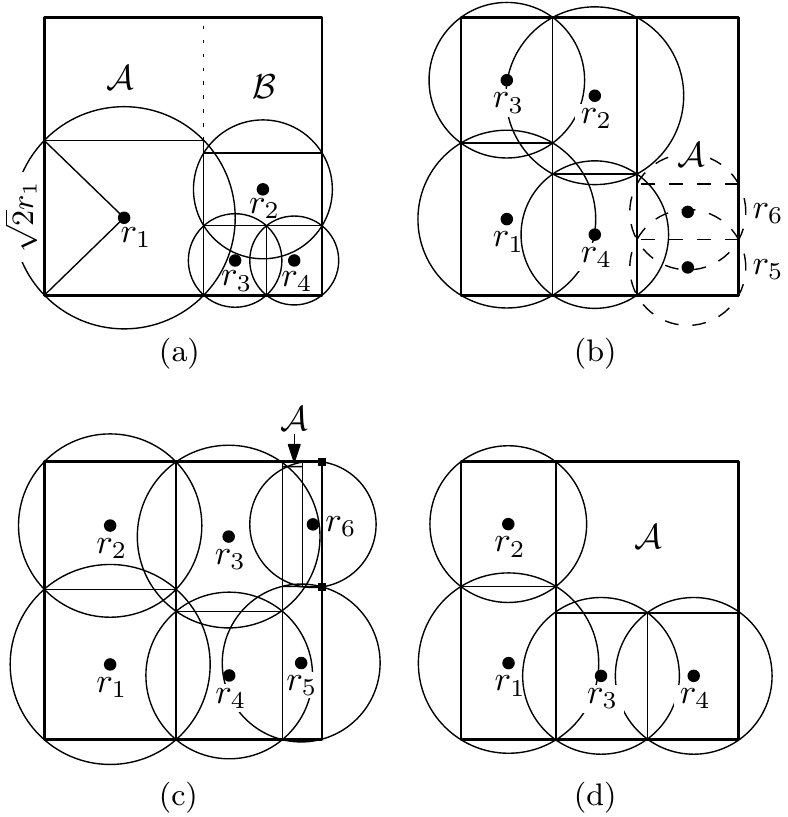}
	\caption{Covering routines that mainly rely on the four largest disks to cover $\mathcal{R}$.
	(a) In Routine~\sbstratref{strat:rectanglesSB-4ld}{1}, we place $r_1$ in the bottom-left corner, covering its inscribed square; use $r_2,r_3$ and $r_4$ to cover as much of the vertical strip remaining to the left of $r_1$, and recurse on $\mathcal{A}$ and $\mathcal{B}$.
	(b) In Routine~\sbstratref{strat:rectanglesSB-4ld}{2}, in the first case, we cover a rectangular strip using $r_1,\ldots,r_4$.
	Either use recursion immediately on the remainder $\mathcal{A}$, or recurse after placing $r_5$ and possibly $r_6$ covering a rectangle at the bottom of $\mathcal{A}$.
	(c) In Routine~\sbstratref{strat:rectanglesSB-4ld}{2}, in the second case, we cover a rectangular strip using $r_1,\ldots,r_4$ and place $r_5$ at the bottom of the remainder as in (b);
	however, we change the placement of $r_6$ to cover the remaining part of the right side of $\mathcal{R}$.
	The points that determine the position of $r_6$ are marked by black squares in the figure.
	We use recursion to cover the bounding box $\mathcal{A}$ of the area that remains uncovered.
	(d) In Routine~\sbstratref{strat:rectanglesSB-4ld}{3}, we cover an $L$-shaped region of $\mathcal{R}$ using the four largest disks, and recurse on the remaining region $\mathcal{A}$.}
	\label{fig:rectanglesSB-r1_in_corner_right_34_and_2}
\end{figure}

\mparagraph{Using the five largest disks}
Routines~(\sbstratref{strat:rectanglesSB-5ld}{1})~and~(\sbstratref{strat:rectanglesSB-5ld}{2}) consider different placements of the five largest disks and recursion to cover $\mathcal{R}$; see Fig.~\ref{fig:rectanglesSB-five-disks}.
\begin{figure}%
	\begin{center}
		\resizebox{.75\linewidth}{!}{\includegraphics{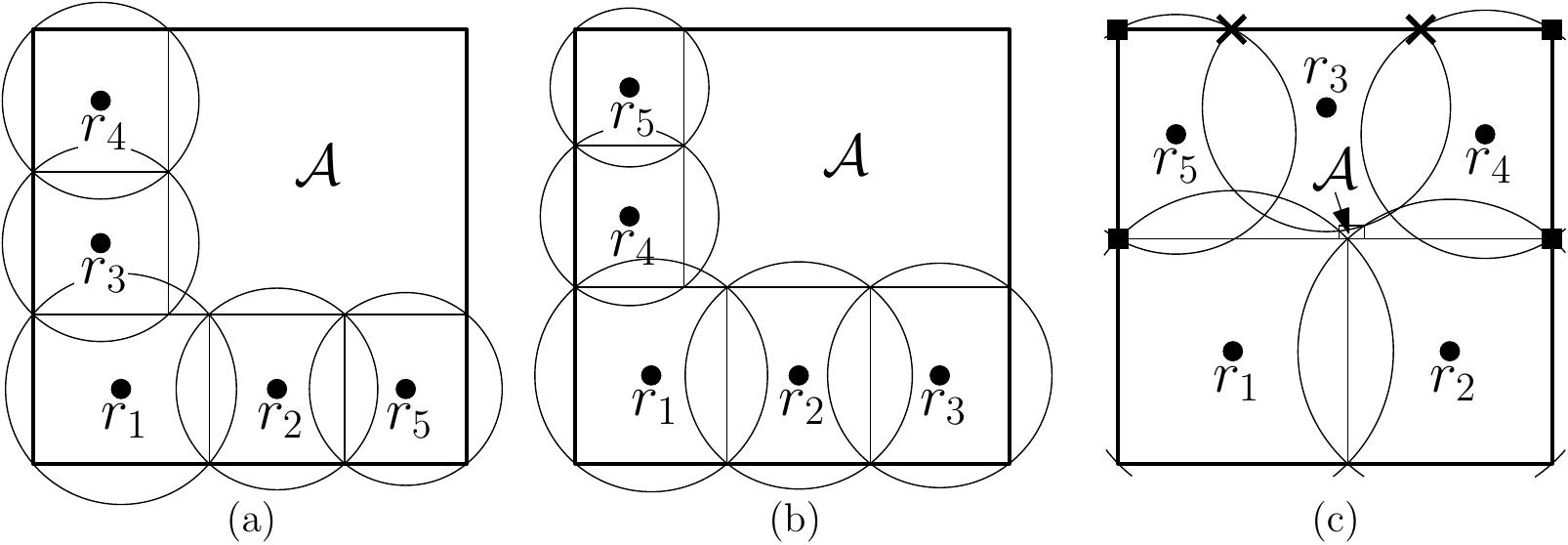}}
	\end{center}
	\caption{Routines for covering $\mathcal{R}$ using the five largest disks and recursion.
		According to Routine~\sbstratref{strat:rectanglesSB-5ld}{1}, in (a) and (b), we first cover a horizontal strip of maximum height using three disks ($r_1,r_2,r_3$ or $r_1,r_2,r_5$) and then cover a vertical strip using the other two disks.
		(c) Routine~\sbstratref{strat:rectanglesSB-5ld}{2} places the five largest disks such that everything but a small region $\mathcal{A}$ is covered.
		The points that define the placement of $r_4$ and $r_5$ in are marked by boxes; those that define the placement of $r_3$ are marked $\times$.
		All three routines use recursion to cover $\mathcal{A}$.}
	\label{fig:rectanglesSB-five-disks}
\end{figure}%

\mparagraph{Using the six largest disks}
Routines~(\sbstratref{strat:rectanglesSB-6ld}{1})--(\sbstratref{strat:rectanglesSB-6ld}{3}) consider different placements of the six largest disks and recursion to cover $\mathcal{R}$; see Fig.~\ref{fig:rectanglesSB-six-disks}.
\begin{figure}%
	\begin{center}
		\resizebox{.94\linewidth}{!}{\includegraphics{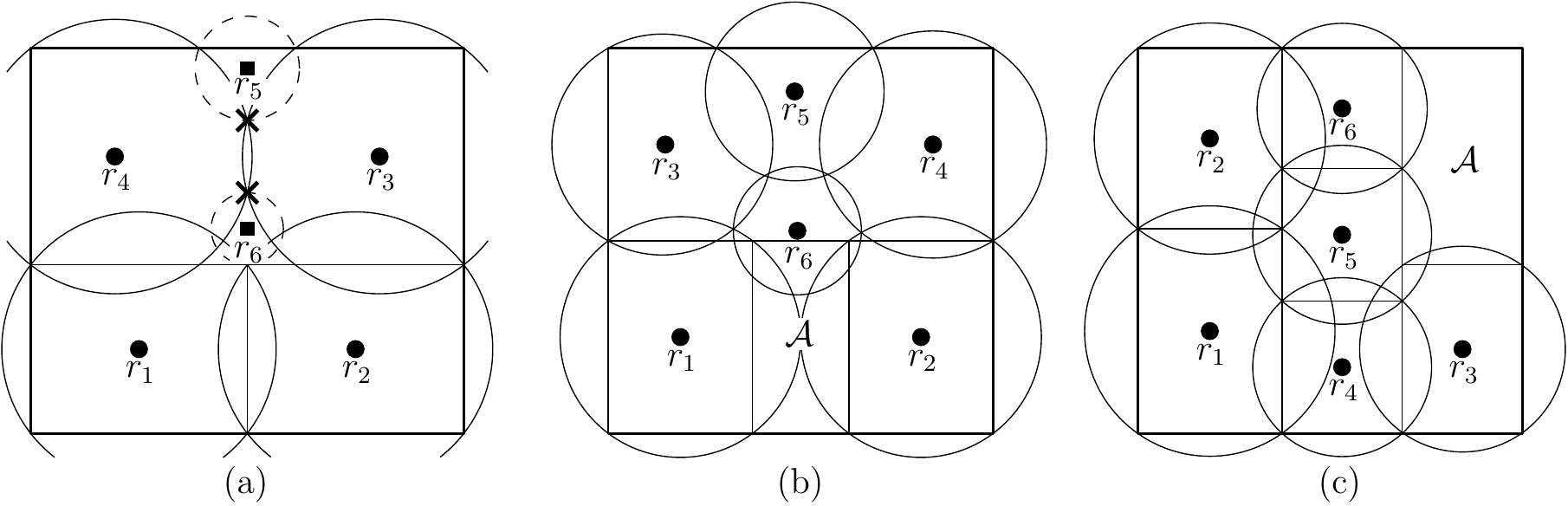}}
	\end{center}
	\caption{(a) Routine~\sbstratref{strat:rectanglesSB-6ld}{1} covers $\mathcal{R}$ using the six largest disks.
	(b) In Routine~\sbstratref{strat:rectanglesSB-6ld}{2}, we also use recursion on the remaining disks to cover an additional rectangular region $\mathcal{A}$.
	(c) In Routine~\sbstratref{strat:rectanglesSB-6ld}{3}, we cover two vertical strips using $r_1,r_2$ and $r_4,r_5,r_6$, using $r_3$ and recursion to cover the remaining strip.}
	\label{fig:rectanglesSB-six-disks}
\end{figure}%

\mparagraph{Using the seven largest disks} 
Routines~(\sbstratref{strat:rectanglesSB-7ld}{1})--(\sbstratref{strat:rectanglesSB-7ld}{8}) consider different placements of the seven largest disks, together with recursion, to cover $\mathcal{R}$; see Figs.~\ref{fig:rectanglesSB-six-disks-max-width},~\ref{fig:rectanglesSB-six_disk_222_with_recursion}~and~\ref{fig:rectanglesSB-seven-disk-strategies}.
\begin{figure}[t]
	\centering
	\resizebox{.95\linewidth}{!}{\includegraphics{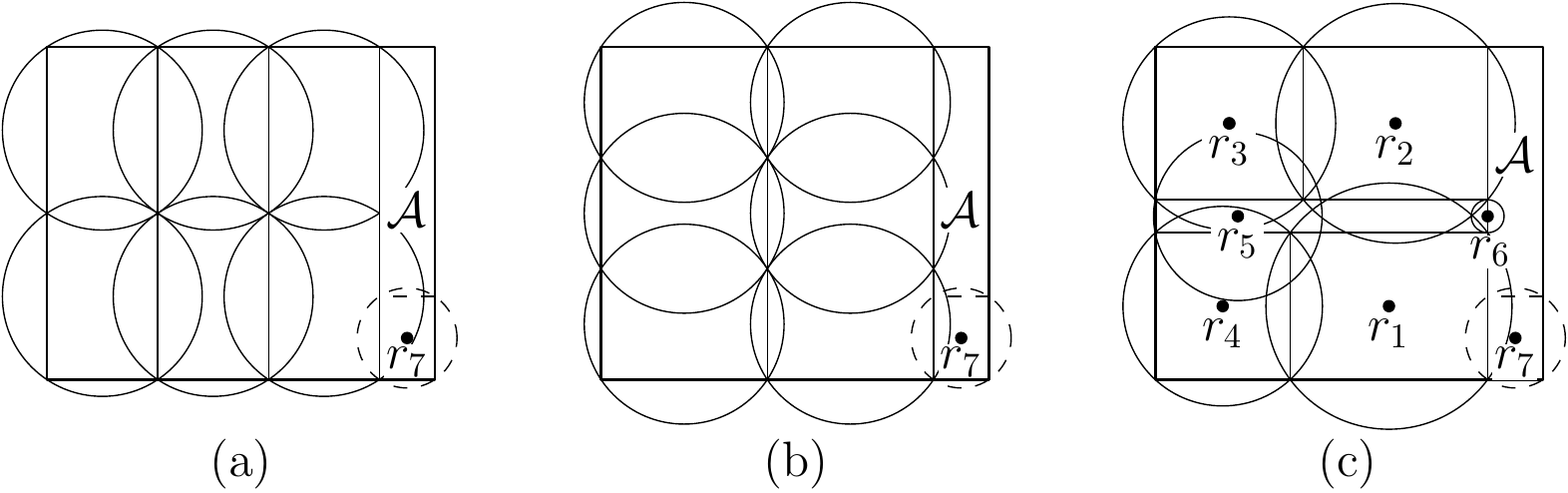}}
	\caption{Routine~\sbstratref{strat:rectanglesSB-7ld}{1} considers the following three configurations to cover a strip of maximal width $w$.
	(a) Using any partition of $r_1,\ldots,r_6$ into three groups of two disks, each covering a strip of height $1$ and maximal width,
	(b) using any partition of $r_1,\ldots,r_6$ into two groups of three disks, each covering a strip of height $1$ and maximal width, or
	(c) using the disks $r_1,r_4$ and $r_2,r_3$ to cover strips of width $w$ and maximal height and covering the uncovered pockets using $r_5$ and $r_6$.}
	\label{fig:rectanglesSB-six-disks-max-width}
\end{figure}%
\begin{figure}%
	\begin{center}
		\resizebox{.85\linewidth}{!}{\includegraphics{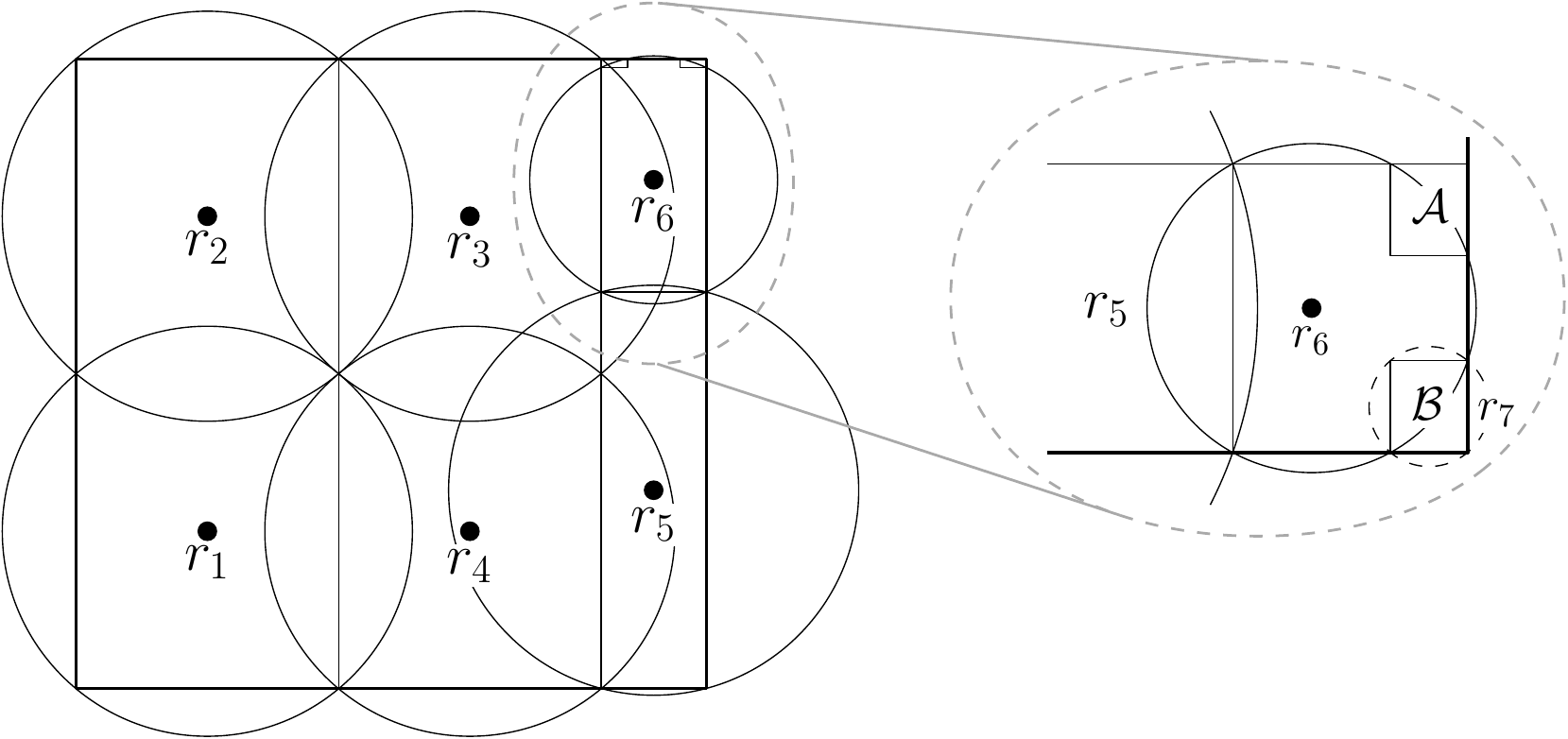}}
	\end{center}
	\caption{Routine~\sbstratref{strat:rectanglesSB-7ld}{2} covers as much width as possible using disks $r_1,\ldots,r_4$, using $r_5,r_6$ and $r_7$ on the remaining strip.}
	\label{fig:rectanglesSB-six_disk_222_with_recursion}
\end{figure}%
\begin{figure}%
	\centering
	\includegraphics[width=.99\linewidth]{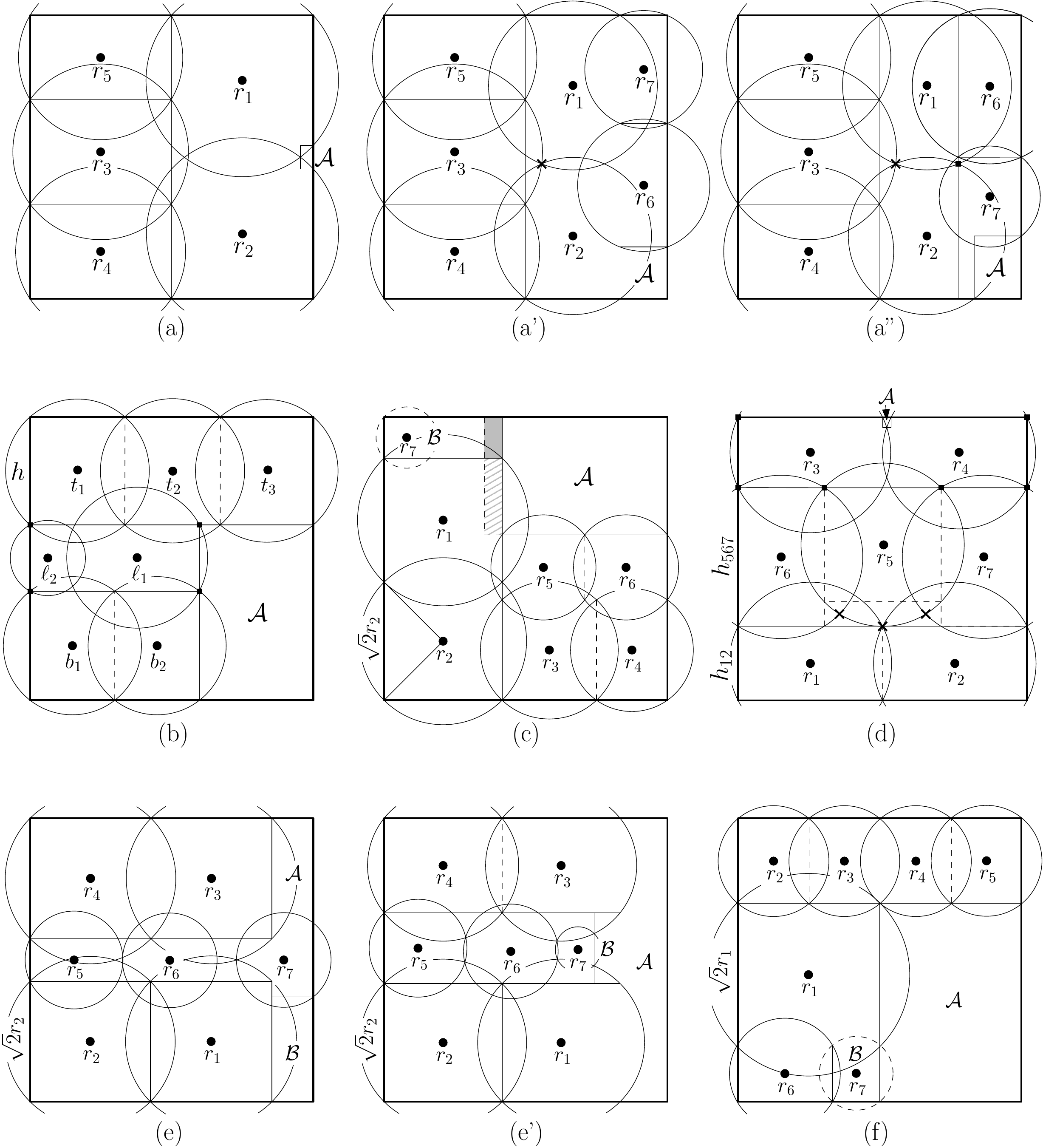}
	\caption{Routines~\sbstratref{strat:rectanglesSB-7ld}{3}--\sbstratref{strat:rectanglesSB-7ld}{8} using disks $r_1,\ldots,r_7$ and recursion to cover $\mathcal{R}$.}
	\label{fig:rectanglesSB-seven-disk-strategies}
\end{figure}%

\FloatBarrier
\subsection{Proof Structure for Theorem~\ref{thm:mainRectangles}}
Tightness of the result claimed by Theorem~\ref{thm:mainRectangles} is proved by Lemma~\ref{lem:worst-cases-rectangles}.
Therefore, proving Theorem~\ref{thm:mainRectangles} means proving that, for any skew $\lambda$, any collection of disks $D$ with $A(D) = A^*(\lambda)$ suffices to cover $\mathcal{R}$.
As in the proof of Lemma~\ref{lem:rectanglesSB}, we begin by reducing the number of cases we have to consider to a finite number.
Again, we begin by proving our result for all rectangles with sufficiently large skew.
\begin{restatable}{lemma}{lemmathmlargeskew}\label{lem:main-large-skew}
	Let $\lambda \geq \overline{\lambda}$ and let $D$ be a collection of disks with $W(D) = W^*(\lambda)$.
	We can cover $\mathcal{R}$ using the disks from $D$.
\end{restatable}
\ifthenelse{\boolean{appproofmainthm}}%
{%
	Due to space restrictions, for the full proof we refer to the full version~\cite{fullversion} of our paper.
}{%
	This lemma is proved in Section~\ref{sec:proof-lemma-thm-large-skew}.
}%
The proof is manual and uses the two simple routines \textsc{Split Cover} (\wcstratref{strat:main-clr}{1}) and \textsc{Large Disk} (\wcstratref{strat:main-clr}{2});
see Fig.~\ref{fig:mainRectangles-large-skew}.%
\begin{figure}[b]
	\centering
	\vspace{-0.5cm}
	\includegraphics[width=.75\linewidth]{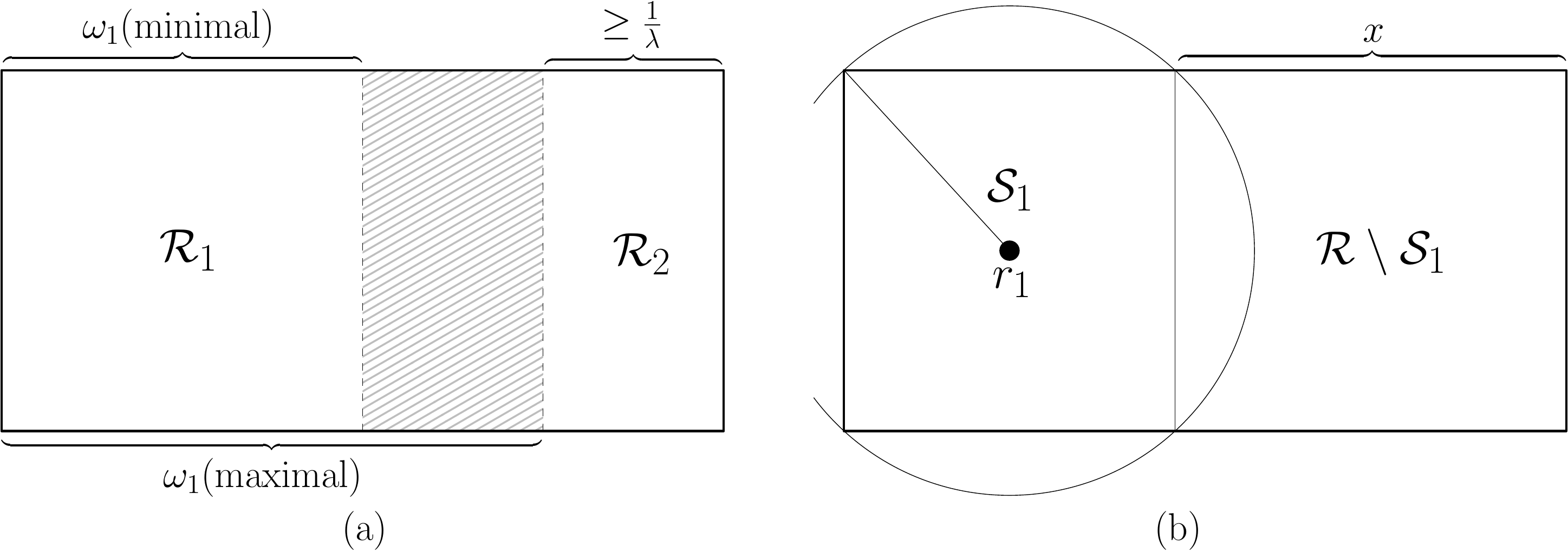}
	\caption{
	(a) The routine \textsc{Split Cover} (\wcstratref{strat:main-clr}{1}) applies \textsc{Greedy Splitting} to the input disks, splits $\mathcal{R}$ into $\mathcal{R}_1,\mathcal{R}_2$ according to the split and recurses.
	The resulting split must not be too unbalanced for this routine to succeed.
	(b) The routine \textsc{Large Disk} (\wcstratref{strat:main-clr}{2}) places $r_1$ covering a rectangle $\mathcal{S}_1$ at the right border of $\mathcal{R}$ and recurses on the remaining rectangle.}
	\label{fig:mainRectangles-large-skew}
\end{figure}
Intuitively speaking, if $r_1$ is small, we split $D$ using \textsc{Greedy Splitting}, split $\mathcal{R}$ accordingly, and recurse on the two resulting regions.
On the other hand, if $r_1$ is big, we cover the left side of $\mathcal{R}$ using $r_1$ and recurse on the remaining region.

The remainder of the proof of Theorem~\ref{thm:mainRectangles} is again based on a list of simple covering routines, which our algorithm tries to apply until it finds a working routine.
We prove that there always is a working routine in the list using an automatic prover based on interval arithmetic as described in Section~\ref{sec:interval-arithmetic}.
After automatic analysis, several critical cases remain.
\ifthenelse{\boolean{appproofmainthm}}%
{%
	We complete our proof by manually analyzing these critical cases.
}{%
	In Section~\ref{sec:manual-analysis}, we perform a manual analysis of these critical cases in order to complete our proof.
}%
In the following, we give a brief description of the routines we use.
\ifthenelse{\boolean{appproofmainthm}}{%
	Due to space constraints, for details, we refer to the full version of our paper~\cite{fullversion}.%
	\wcnewsubsec{strat:main-clr}%
	\wcnewsubsec{strat:main-hsd}%
	\wcnewsubsec{strat:main-cld}%
	\wcnewsubsec{strat:main-2ld}%
	\wcnewsubsec{strat:main-3ld}%
	\wcnewsubsec{strat:main-4ld}%
}{%
	The routines are described in detail in Section~\ref{sec:proof-mainthm}.
}%

\mparagraph{Small disks} Because the covering coefficient guaranteed by Lemma~\ref{lem:rectanglesSB} is always better than $E^*(\lambda)$, Routine~(\wcstratref{strat:main-hsd}{1}) attempts to apply Lemma~\ref{lem:rectanglesSB} directly; this works if the largest disk is not too big.

\mparagraph{Using the largest disk} Routines~(\wcstratref{strat:main-cld}{1})--(\wcstratref{strat:main-cld}{3}) try several placements for the largest disk $r_1$; see Fig.~\ref{fig:mainRectangles-largest-disk}.
\begin{figure}
	\begin{center}
		\resizebox{.98\textwidth}{!}{\includegraphics{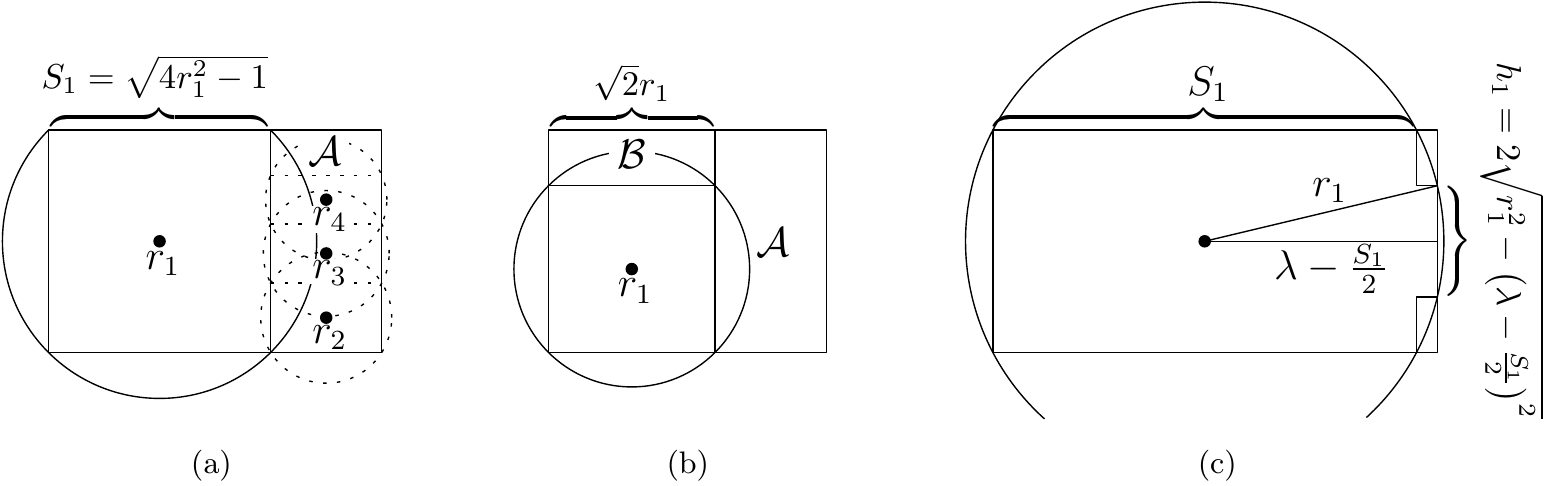}}
	\end{center}
	\caption{(a) In Routine~\wcstratref{strat:main-cld}{1}, we place $r_1$ covering a strip at the left side of $\mathcal{R}$ and try to recurse on $\mathcal{A}$.
	If that does not work, we also try to place $r_2,r_3$ and potentially $r_4$ covering horizontal strips at the bottom of the remaining rectangle before we try recursing.
	(b) In Routine~\wcstratref{strat:main-cld}{2}, we place $r_1$ covering its inscribed square at the bottom-left corner of $\mathcal{R}$, covering the two remaining regions $\mathcal{A},\mathcal{B}$ recursively.
	(c) In Routine~\wcstratref{strat:main-cld}{3}, we place $r_1$ covering a strip at the left side of $\mathcal{R}$; if placed like this, $r_1$ intersects the right border of $\mathcal{R}$, only leaving two small uncovered pockets.}
	\label{fig:mainRectangles-largest-disk}
\end{figure}

\mparagraph{Using the two largest disks} 
Routines~(\wcstratref{strat:main-2ld}{1})~and~(\wcstratref{strat:main-2ld}{2}) try several placements for the largest two disks $r_1,r_2$; see Fig.~\ref{fig:mainRectangles-two-disks}.
\begin{figure}
	\begin{center}
		\resizebox{.99\linewidth}{!}{\includegraphics{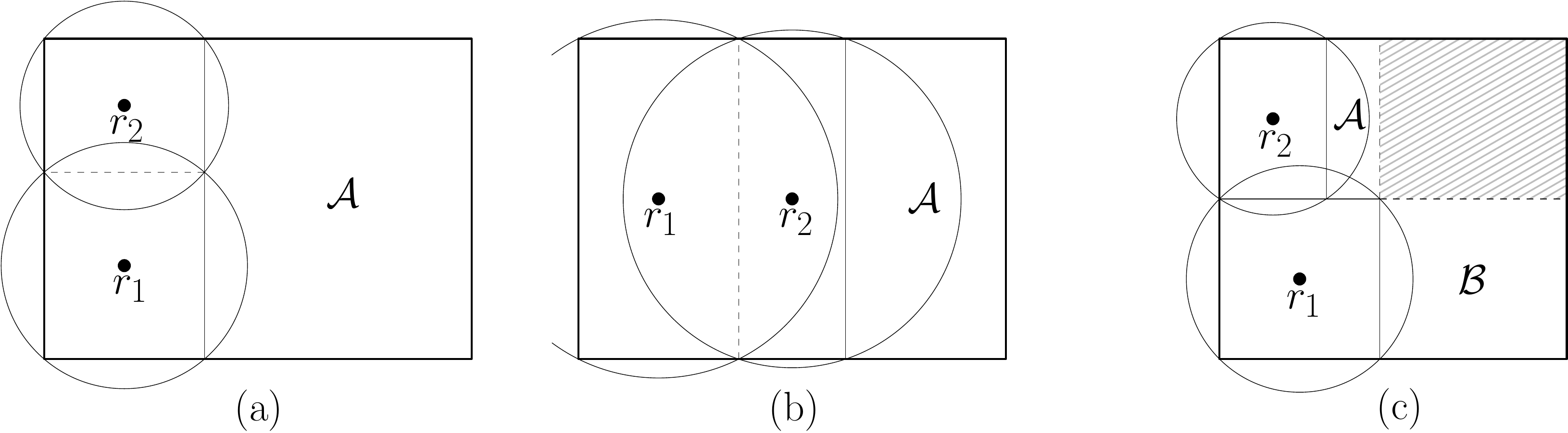}}
	\end{center}
	\caption{(a) and (b) depict Routine~\wcstratref{strat:main-2ld}{1}. The two largest disks are used to cover as wide a strip as possible at the left side of $\mathcal{R}$; the remaining disks are used for recursion on $\mathcal{A}$.
	(c)~Routine~\wcstratref{strat:main-2ld}{2} places $r_1$ covering its inscribed square and covers the remaining part of $\mathcal{R}$'s left boundary using $r_2$. Two regions $\mathcal{A}$ and $\mathcal{B}$ remain.
	The shaded area can be added to either $\mathcal{A}$ or $\mathcal{B}$; we try both options.}
	\label{fig:mainRectangles-two-disks}
\end{figure}

\mparagraph{Using the three largest disks}
Routines~(\wcstratref{strat:main-3ld}{1})--(\wcstratref{strat:main-3ld}{5}) consider several placements for the largest three disks; see Figs.~\ref{fig:strategy_cover_with_three_disks},~\ref{fig:mainRectangles-three-disks},~\ref{fig:strategy_three_disk_recursion_small_corner},~and~\ref{fig:strategy_three_disk_recursion_r1_r3_pocket}.
\begin{figure}
	\begin{center}
		\resizebox{.5\linewidth}{!}{\includegraphics{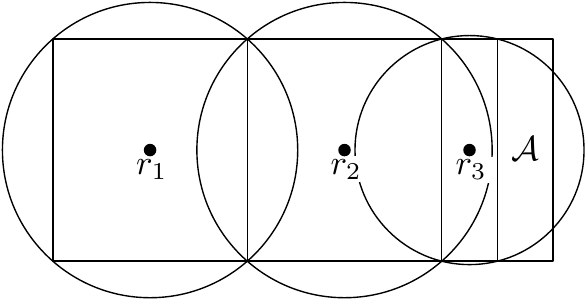}}
	\end{center}
	\caption{
		Routine~\wcstratref{strat:main-3ld}{1} places the three largest disks next to each other, each covering a vertical strip of height $1$.
		If this does not cover the entire rectangle, we recurse on the bounding box $\mathcal{A}$ of the remaining area.
	}
	\label{fig:strategy_cover_with_three_disks}
\end{figure}
\begin{figure}
	\begin{center}
		\resizebox{.9\linewidth}{!}{\includegraphics{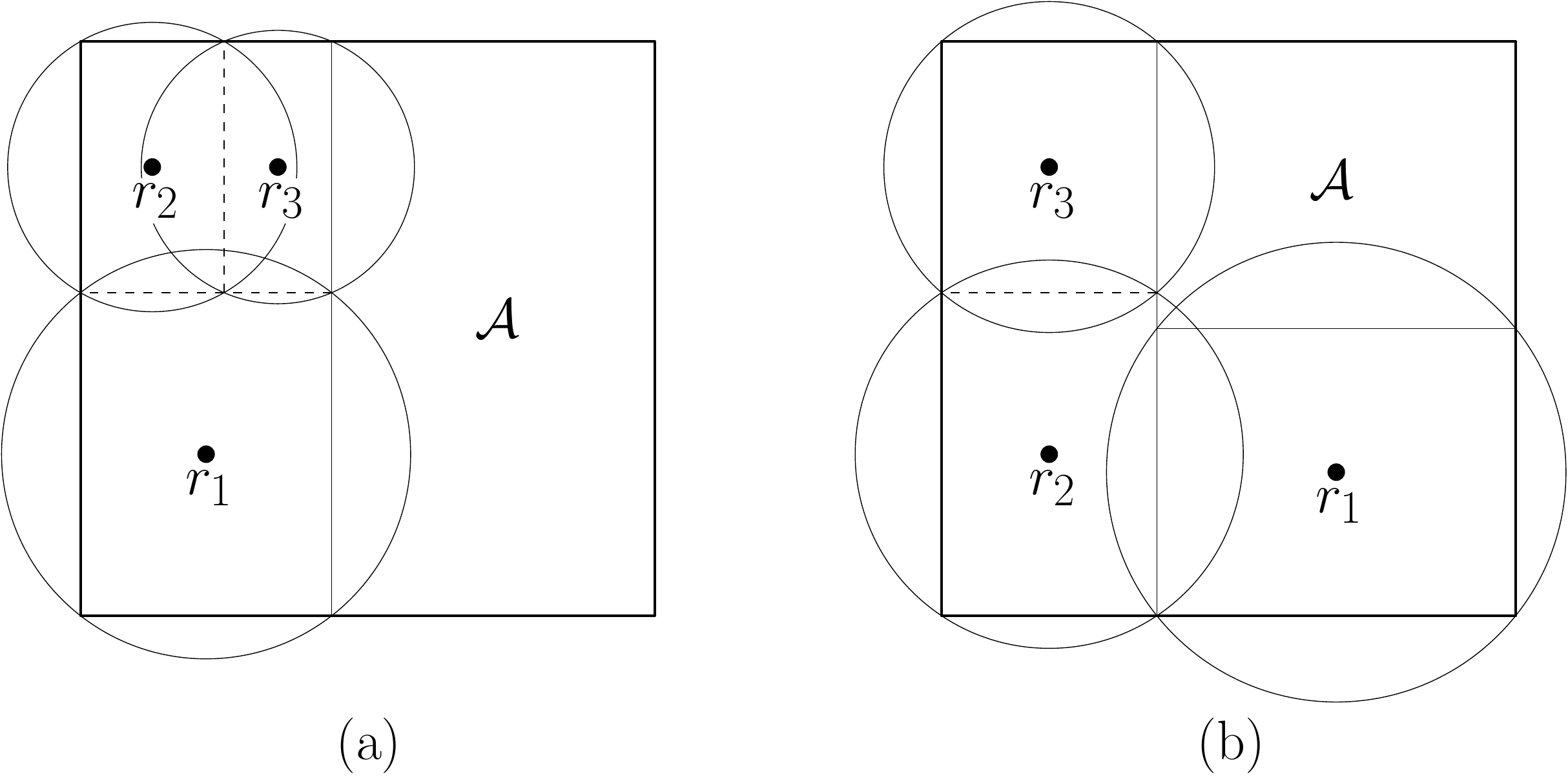}}
	\end{center}
	\caption{
		(a) Routine~\wcstratref{strat:main-3ld}{2} builds a strip of maximum possible width by placing $r_1$ at the bottom and $r_2$ besides $r_3$ on top.
		(b) Routine~\wcstratref{strat:main-3ld}{3} builds a vertical strip of maximum possible width by placing $r_2,r_3$ on top of each other, and covers the remaining part of the lower boundary using $r_1$.}
	\label{fig:mainRectangles-three-disks}
\end{figure}
\begin{figure}
	\begin{center}
		\resizebox{.95\linewidth}{!}{\includegraphics{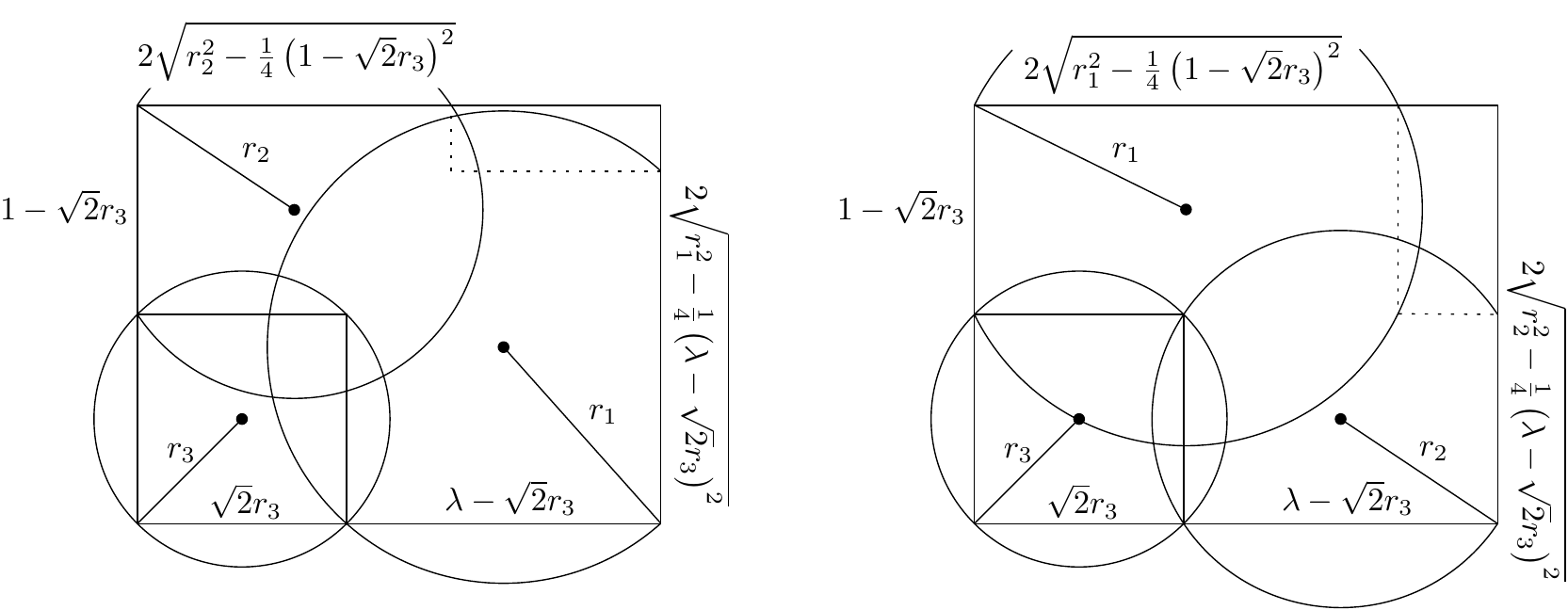}}
	\end{center}
	\caption{
		Routine~\wcstratref{strat:main-3ld}{4} covers the rectangle using the third-largest disk to cover a square at the bottom-left corner.
		The remaining rectangle that we recurse on is drawn with dashed outline.
		{\bf Left: }Placing the largest disk to the right of the third-largest disk and the second-largest disk on top of the third-largest disk.
		{\bf Right: }Placing the largest disk on top of the third-largest disk and the second-largest disk to the right of the third-largest disk.
	}
	\label{fig:strategy_three_disk_recursion_small_corner}
\end{figure}
\begin{figure}
	\begin{center}
		\resizebox{.95\linewidth}{!}{\includegraphics{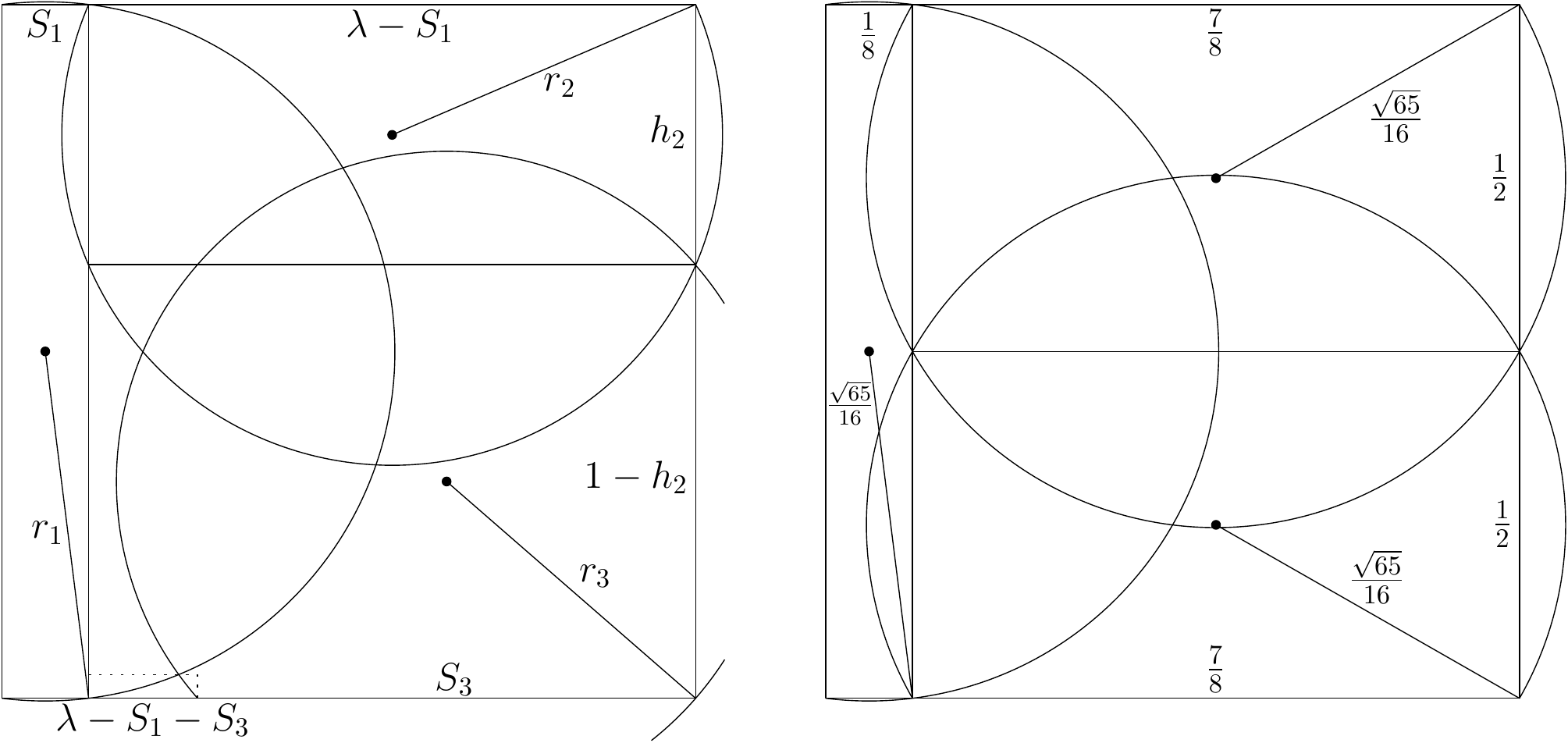}}
	\end{center}
	\caption{
		{\bf Left:} Routine~\wcstratref{strat:main-3ld}{5} covers the rectangle using the largest disk to cover a strip of width $S_1$, using the second- and third-largest disks to cover the remaining corners.
		The bounding box of the uncovered pocket between the largest and third-largest disk is drawn with dashed outline.
		{\bf Right:} The worst-case example for a square, consisting of three equal disks with radius $\frac{\sqrt{65}}{16}$.
		The covering of Routine~\wcstratref{strat:main-3ld}{5} converges to this covering for disks converging to this worst-case example.
	}
	\label{fig:strategy_three_disk_recursion_r1_r3_pocket}
\end{figure}

\mparagraph{Using the four largest disks}
Routines~(\wcstratref{strat:main-4ld}{1})--(\wcstratref{strat:main-4ld}{3}) consider several placements for the largest four disks; see Fig.~\ref{fig:mainRectangles-four-disks}.
\begin{figure}
	\begin{center}
		\resizebox{.99\linewidth}{!}{\includegraphics{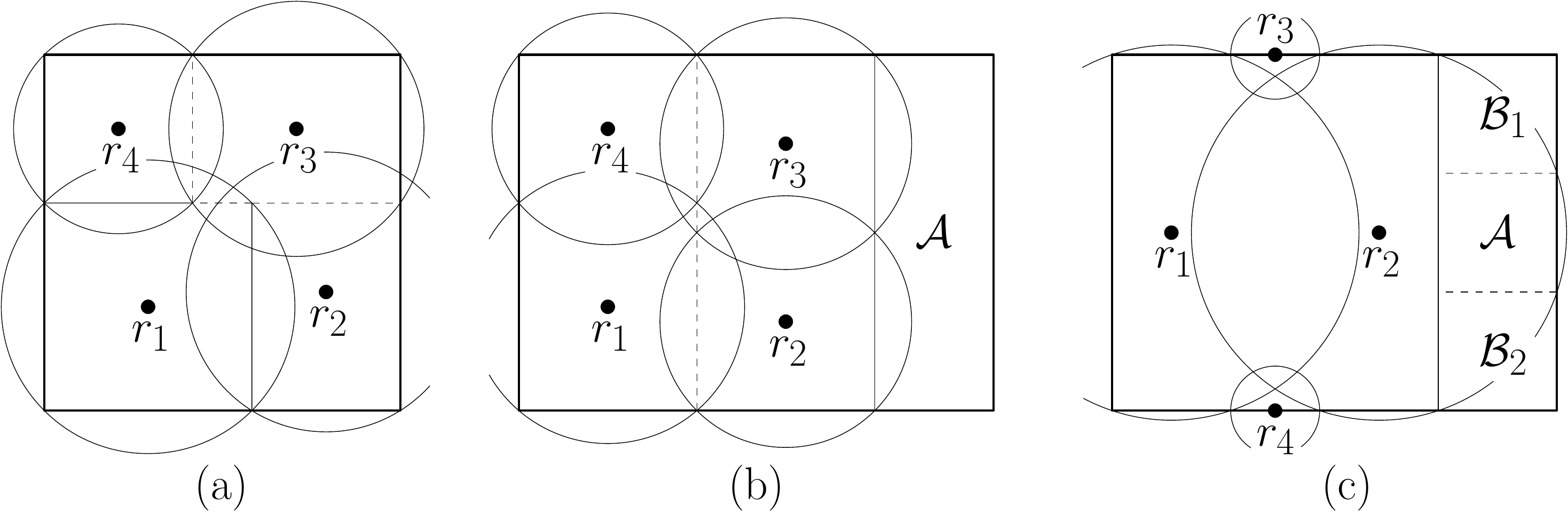}}
	\end{center}
	\caption{(a)~Routine~\wcstratref{strat:main-4ld}{1} covers $\mathcal{R}$ using only the four largest disks. The dashed outline depicts the rectangle that $r_3$ has to be able to cover.
	(b)~Routine~\wcstratref{strat:main-4ld}{2} covers a strip of maximum possible width using two groups of two disks and recurses on the remaining rectangle $\mathcal{A}$.
	(c)~Routine~\wcstratref{strat:main-4ld}{3} covers $\mathcal{R}$ by placing the two largest disks besides each other, filling the gaps between the disks using $r_3,r_4$.
	If this does not cover $\mathcal{R}$, we either recurse on the remaining strip $\mathcal{A}$ or on the bounding box of two pockets $\mathcal{B}_1,\mathcal{B}_2$ if $r_2$ intersects $\mathcal{R}$'s right border.}
	\label{fig:mainRectangles-four-disks}
\end{figure}

\def\proofdetailsmaintheorem{%
\subsection{Covering Without Weight Bound --- Proof of Theorem~\ref{thm:mainRectangles}}\label{sec:proof-mainthm}
It remains to prove Theorem~\ref{thm:mainRectangles}.
Similar to the proof of Lemma~\ref{lem:rectanglesSB}, the proof is based on an algorithm that tries to apply a sequence of simple routines until it finds a working one.
As input, the algorithm receives a rectangle $\mathcal{R}$ of dimensions $\lambda \times 1$ and a collection $D$ of disks $r_1 \geq \cdots \geq r_n$ with $W(D) \coloneqq \sum_{i=1}^{n} r_i^2 \geq W^*(\lambda)$; w.l.o.g., we assume $W(D) = W^*(\lambda)$.
If $D$ contains only one disk $r_1$, that disk has greater weight than $\mathcal{R}$'s circumcircle, and we can use it to cover $\mathcal{R}$ completely.
Therefore, in the following, we may assume that we are given at least two disks.
Recall that $E^*(\lambda) \coloneqq \frac{W^*(\lambda)}{\lambda}$ is the covering coefficient the algorithm has to achieve.

\subsubsection{Covering Long Rectangles}\label{sec:proof-lemma-thm-large-skew}\wcnewsubsec{strat:main-clr}
In this section, we prove Lemma~\ref{lem:main-large-skew}, i.e., the result of Theorem~\ref{thm:mainRectangles} for large skew $\lambda \geq \overline{\lambda} \approx 2.08988$.
\restatethm{\lemmathmlargeskew*}{lem:main-large-skew}

In this situation, we have $W^*(\lambda) = \frac{\lambda^2 + 2}{4}$, corresponding to the weight $\frac{\lambda^2+1}{4}$ of the circumcircle of $\mathcal{R}$ and another circle of radius $\frac{1}{2}$ that suffices to cover one of $\mathcal{R}$'s shorter sides; see Fig.~\ref{fig:worst-cases-rectangles}.
In this case, our algorithm uses two simple routines to cover $\mathcal{R}$; see Fig.~\ref{fig:mainRectangles-large-skew}.

\pwcstratref{strat:main-clr}{1} If $r_1 \leq \overline{r} \coloneqq \sqrt{\frac{195\overline{\lambda}}{128} - \frac{11}{4}} \approx 0.6586$, we apply the routine \textsc{Split Cover}.
This routine works by applying \textsc{Greedy Splitting} to $D$, which results in two non-empty collections $D_1,D_2$.
We partition $\mathcal{R}$ into two rectangles $\mathcal{R}_1,\mathcal{R}_2$ such that $\frac{W(D_1)}{W(D_2)} = \frac{\abs*{\mathcal{R}_1}}{\abs*{\mathcal{R}_2}}$ by dividing the longer side of $\mathcal{R}$ in that ratio.
We then recursively cover $\mathcal{R}_1$ and $\mathcal{R}_2$.

\pwcstratref{strat:main-clr}{2} Otherwise, if $r_1 > \overline{r}$, we apply the routine \textsc{Large Disk}.
It works by placing $r_1$ covering a vertical rectangular strip $\mathcal{S}_1$ of height $1$ and maximum possible width at the left border of $\mathcal{R}$.
After placing $r_1$ in this manner, we recurse on the remaining $\mathcal{R} \setminus \mathcal{S}_1$ using all remaining disks.
We prove the following two lemmas stating that these routines suffice to cover $\mathcal{R}$.
\begin{lemma}\label{lem:mainRectangles-split-cover}
	Let $\lambda \geq \overline{\lambda}$ and $r_1 \leq \overline{r}$.
	Then, \textsc{Split Cover} can be used to cover $\mathcal{R}$ completely.
\end{lemma}
\begin{lemma}\label{lem:mainRectangles-large-disk}
	Let $\lambda \geq \overline{\lambda}$ and $r_1 \geq \overline{r}$.
	Then, \textsc{Large Disk} can be used to cover $\mathcal{R}$ completely.
\end{lemma}
\begin{proof}[Proof of Lemma~\ref{lem:mainRectangles-split-cover}]
	The routine \textsc{Split Cover} cuts $\mathcal{R}$ into two rectangles $\mathcal{R}_1,\mathcal{R}_2$ of width $\omega_1 = \lambda \cdot \frac{W(D_1)}{W(D)}, \omega_2 = \lambda \cdot \frac{W(D_2)}{W(D)}$ and height $1$.
	Because $\lambda \geq \overline{\lambda}$, we know that $E^*(\lambda') \leq E^*(\lambda)$ for all $\lambda' \leq \lambda$, i.e., according to Theorem~\ref{thm:mainRectangles}, any rectangle with some skew $\lambda' \leq \lambda$ can be covered at least as efficiently as $\mathcal{R}$.
	Therefore, to prove that we can recurse on $\mathcal{R}_1,\mathcal{R}_2$, it suffices to prove that the skew of $\mathcal{R}_1,\mathcal{R}_2$ is at most $\lambda$.
	W.l.o.g., let $\omega_1 \geq \omega_2$.
	Because $\omega_1 < \lambda$, the skew of $\mathcal{R}_1,\mathcal{R}_2$ can only become larger than $\lambda$ if the width $\omega_2$ is too small, i.e., $\omega_2 < \frac{1}{\lambda}$.
	
	The weight $W(D_2)$ assigned to $\mathcal{R}_2$ is at least $\frac{W(D) - r_1^2}{2} = \frac{\lambda^2 + 2}{8} - \frac{r_1^2}{2}$.
	Hence, the width of $\mathcal{R}_2$ satisfies \[\omega_2 = \frac{\lambda}{W(D)} \cdot W(D_2) \geq \frac{\lambda}{2} - \frac{\lambda r_1^2}{2W(D)} \geq \frac{\lambda}{2} - \frac{2\lambda\left(\frac{195\overline{\lambda}}{128}-\frac{11}{4}\right)}{\lambda^2+2} \eqqcolon \hat{\omega}_2.\]
	As a function of $\lambda \geq \overline{\lambda}$, $\hat{\omega}_2$ is monotonically increasing and $\frac{1}{\lambda}$ is monotonically decreasing.
	Therefore, $\hat{\omega}_2 - \frac{1}{\lambda}$ is monotonically increasing in $\lambda$.
	To prove $\omega_2 \geq \frac{1}{\lambda}$, it therefore suffices to observe that, for $\lambda = \overline{\lambda}$, $\hat{\omega}_2 - \frac{1}{\lambda} \approx 0.7601 - 0.4784 > 0$.
\end{proof}
\begin{proof}[Proof of Lemma~\ref{lem:mainRectangles-large-disk}]
	Because $\overline{r} \approx 0.6586 > \frac{1}{2}$, we can always place $r_1$ such that a vertical rectangular strip $\mathcal{S}_1$ of height $1$ and positive width $\omega_1$ is covered.
	Therefore, to prove that \textsc{Large Disk} works, it suffices to prove that the remaining weight $W(D) - r_1^2$ suffices to apply Theorem~\ref{thm:mainRectangles} to the remaining rectangle $\mathcal{R} \setminus \mathcal{S}_1$.
	Let $x \coloneqq \lambda - \omega_1$ be the width of the remaining rectangle; we have $r_1 = \frac{1}{2}\sqrt{1+\left(\lambda-x\right)^2}$ and $\lambda = x + 2\sqrt{r_1^2-\frac{1}{4}}$.
	We consider the three cases (a) $x \geq \overline{\lambda}$, (b) $\frac{1}{\overline{\lambda}} \leq x < \overline{\lambda}$ and (c) $x < \frac{1}{\overline{\lambda}}$.

	In case (a), in order to apply Theorem~\ref{thm:mainRectangles}, we have to prove that the remaining weight $R_2 \coloneqq W(D)-r_1^2$ is at least $\frac{x^2 + 2}{4}$.
	We have
	\[\frac{\lambda^2+2}{4} - r_1^2 = \frac{1+2\lambda x-x^2}{4} \geq \frac{x^2+2}{4} \Leftrightarrow 2\underbrace{\left(x+2\sqrt{r_1^2-\frac{1}{4}}\right)}_{=\lambda}x \geq 2x^2+1 \Leftrightarrow 4x\sqrt{r_1^2-\frac{1}{4}} \geq 1,\]
	which follows from $x \geq \overline{\lambda} > 2$ and $r_1 \geq \overline{r}$.

	In case (b), let $\lambda_x \coloneqq \max\{\frac{1}{x},x\}$ be the skew of the remaining rectangle and observe that $E^*(\lambda_x) \leq \frac{195}{256}$.
	Therefore, it suffices to show that the remaining weight is at least $\frac{195x}{256}$.
	We have
	\[\frac{\lambda^2+2}{4} - r_1^2 \geq \frac{195x}{256} \Leftrightarrow 1+2\lambda x-x^2 \geq \frac{195x}{64} \Leftrightarrow 1 \geq -x^2 + x\left(\frac{195}{64} - 4\sqrt{r_1^2-\frac{1}{4}}\right).\]
	To prove this inequality, observe that $r_1^2 \geq \overline{r}$ yields $-x^2 + x\left(\frac{195}{64} - 4\sqrt{r_1^2-\frac{1}{4}}\right) \leq -x^2 + xc$ for $c \coloneqq \left(\frac{195}{64} - 4\sqrt{\overline{r}^2-\frac{1}{4}}\right)$.
	The function $-x^2 + xc$ attains its global maximum $\frac{c^2}{4}$ at $x = \frac{c}{2}$.
	Because $\frac{c^2}{4} \approx 0.4435 < 1$, the inequality holds and the remaining weight suffices to recurse on $\mathcal{R} \setminus \mathcal{S}_1$.

	In case (c), $x$ is the length of the shorter side of $\mathcal{R} \setminus \mathcal{S}_1$.
	Because the skew $\lambda_x = \frac{1}{x}$ is at least $\overline{\lambda}$, in order to apply Theorem~\ref{thm:mainRectangles}, the remaining weight must be at least $\frac{1+2x^2}{4}$.
	Moreover, we have $\frac{3}{2}x \leq \frac{3}{2}\cdot\frac{1}{\overline{\lambda}} < 1 \leq \lambda$.
	This yields \(\frac{\lambda^2+2}{4} - r_1^2 = \frac{1+2\lambda x - x^2}{4} \geq \frac{1+3x^2-x^2}{4} = \frac{1+2x^2}{4}\).
\end{proof}
This concludes the proof of Theorem~\ref{thm:mainRectangles} for rectangles with large skew $\lambda \geq \overline{\lambda}$; in the following, we may assume $\lambda < \overline{\lambda}$.

\subsubsection{Handling Small Disks}\wcnewsubsec{strat:main-hsd}
\pwcstratref{strat:main-hsd}{1} For the case of skew $\lambda < \overline{\lambda}$, the algorithm begins by checking whether it can use Lemma~\ref{lem:rectanglesSB}; because $E^*(\lambda) > E = \rectanglesSBEff$ for all $\lambda$, this only depends on the size of the largest disk.
We use the following lemma as success criterion.
\begin{lemma}
	If the largest disk $r_1 \in D$ satisfies $r_1^2 \leq \frac{E^*(\lambda)}{E} \cdot \rectanglesSBRB^2$, $D$ suffices to cover $\mathcal{R}$.
\end{lemma}
\begin{proof}
	We distinguish two cases.
	If $\lambda \geq \sqrt{\frac{E^*(\lambda)}{E}}$, we apply Lemma~\ref{lem:rectanglesSB} to a rectangle $\mathcal{R}' \supseteq \mathcal{R}$ of dimensions $\lambda \times \frac{E^*(\lambda)}{E}$ instead of $\mathcal{R}$.
	The total disk weight is $\lambda \cdot E^*(\lambda)$ and the area of $\mathcal{R}'$ is $\lambda \cdot \frac{E^*(\lambda)}{E}$;
	therefore, the covering coefficient we have to achieve for this rectangle is $\lambda \cdot E^*(\lambda) / \frac{\lambda \cdot E^*(\lambda)}{E} = E$, which is what Lemma~\ref{lem:rectanglesSB} guarantees.

	For a rectangle whose shorter side has length $h$, Lemma~\ref{lem:rectanglesSB} requires the weight of the largest disk to satisfy $r_1^2 \leq \rectanglesSBRB^2 \cdot h^2$.
	In this case, $h = \lambda$ if $\sqrt{\frac{E^*(\lambda)}{E}} \leq \lambda \leq \frac{E^*(\lambda)}{E}$ and $h = \frac{E^*(\lambda)}{E}$ otherwise.
	In either case, we have $h \geq \sqrt{\frac{E^*(\lambda)}{E}}$; therefore, $r_1^2 \leq \frac{E^*(\lambda)}{E}\rectanglesSBRB^2$ suffices to apply Lemma~\ref{lem:rectanglesSB} to $\mathcal{R}'$, and we are done.

	Otherwise, we have $\lambda < \sqrt{\frac{E^*(\lambda)}{E}}$.
	In this case, we apply Lemma~\ref{lem:rectanglesSB} to a square $\mathcal{R}' \supseteq \mathcal{R}$ of side length $\sqrt{\frac{E^*(\lambda)}{E}}$.
	The area of $\mathcal{R}'$ is $\frac{E^*(\lambda)}{E}$ and we have at least weight $E^*(\lambda) \cdot \lambda \geq 1 \cdot E^*(\lambda)$, thus the covering coefficient $E$ guaranteed by Lemma~\ref{lem:rectanglesSB} suffices.
	Therefore, because $r_1^2 \leq \frac{E^*(\lambda)}{E} \cdot \rectanglesSBRB^2$, we can apply Lemma~\ref{lem:rectanglesSB} to $\mathcal{R}'$, thus concluding the proof.
\end{proof}

\subsubsection{Covering Using the Largest Disk}\label{sec:two-pockets}\wcnewsubsec{strat:main-cld}
\pwcstratref{strat:main-cld}{1} If $r_1$ is too large to apply Lemma~\ref{lem:rectanglesSB} directly, we consider the following routines.
The first, depicted in Fig.~\ref{fig:mainRectangles-largest-disk}~(a), places $r_1$ covering a vertical rectangular strip of width $S_1 = \sqrt{4r_1^2-1}$ at the left side of $\mathcal{R}$.
We disregard this routine if such a placement is impossible, i.e., if $r_1^2 \leq \frac{1}{2}$.
Afterwards, we check whether we can guarantee successful recursion on the remaining rectangle $\mathcal{A}$.
If this does not work, we also consider covering horizontal rectangular strips at the bottom of $\mathcal{A}$ using disks $r_2,r_3$ and $r_4$, checking whether we can guarantee successful recursion after each additional disk.

\pwcstratref{strat:main-cld}{2} If $\sqrt{2}r_1 < 1$, we also consider placing $r_1$ covering its inscribed square of side length $\sqrt{2}r_1$ in the bottom-left corner of $\mathcal{R}$, as depicted in Fig.~\ref{fig:mainRectangles-largest-disk}~(b).
The remaining two rectangular regions $\mathcal{A}$ to the right of $r_1$ and $\mathcal{B}$ above $r_1$ are covered recursively as follows.
We construct a collection $D_{\mathcal{A}}$ of disks for $\mathcal{A}$ by collecting disks in decreasing order of radius, starting with $r_2$, until the weight in $D_{\mathcal{A}}$ suffices to guarantee successful recursion on $\mathcal{A}$, which may use Theorem~\ref{thm:mainRectangles} or Lemmas~\ref{lem:size-bound-large}~or~\ref{lem:rectanglesSB}.
We bound the cost of this split, i.e., the amount of weight in $D_{\mathcal{A}}$ that exceeds the weight requirement for recursion, as follows.
If $r_2$ or $r_2,r_3$ are sufficient, we can directly compute the cost.
Otherwise, we use $r_4^2$ as an upper bound.
Using this, we check whether we can guarantee successful recursion on $\mathcal{B}$ using the remaining disks.

\pwcstratref{strat:main-cld}{3} If $r_1$ is large enough to intersect the right border of $\mathcal{R}$ when placed covering the left border as depicted in Fig.~\ref{fig:mainRectangles-largest-disk}~(c), we attempt to apply the routine \textsc{Two Pockets}.
This routine places $r_1$ on the left border such that two identical pockets remain to be covered at the right border.
If $r_3$ is big enough to cover one of the pockets, we cover the pockets using $r_2$ and $r_3$.
Otherwise, if $r_2$ is big enough to cover both pockets simultaneously, we cover both pockets using $r_2$.
Otherwise, if $r_2$ is big enough to cover one pocket on its own, we cover one pocket using $r_2$ and check whether we can guarantee successful recursion on the bounding box $\mathcal{B}$ of the remaining pocket.
Finally, if $r_2$ is not big enough to cover one pocket, we subdivide the disks into two groups $D_{\mathcal{A}},D_{\mathcal{B}}$ using \textsc{Greedy Splitting} and check whether we can guarantee successful recursion on the bounding boxes $\mathcal{A},\mathcal{B}$.

\subsubsection{Covering Using the Two Largest Disks}\wcnewsubsec{strat:main-2ld}
\pwcstratref{strat:main-2ld}{1} We continue with routines that mainly rely on the two largest disks $r_1,r_2$.
The first routine, depicted in Fig.~\ref{fig:mainRectangles-two-disks}~(a)~and~(b), uses $r_1$ and $r_2$ to cover a vertical strip of height $1$ and maximum possible width at the left side of $\mathcal{R}$.
To achieve this, the disks are either placed on top of each other or horizontally next to each other.
If this covers $\mathcal{R}$, we are done; otherwise, we check whether we can guarantee successful recursion on the remaining rectangle $\mathcal{A}$ using the remaining disks.

\pwcstratref{strat:main-2ld}{2} If this does not work and if $\sqrt{2}r_1 < 1$, we continue using the following routine, depicted in Fig.~\ref{fig:mainRectangles-two-disks}~(c).
We place $r_1$ covering its inscribed square of side length $\sqrt{2}r_1$ in the bottom-left corner of $\mathcal{R}$.
On top of $r_1$, we place $r_2$ such that it covers the remaining part of $\mathcal{R}$'s left border.
The remaining area can be subdivided into two rectangles $\mathcal{A},\mathcal{B}$ either horizontally or vertically; we consider both options separately.
For both options, we build collections $D_{\mathcal{A}}$ and $D_{\mathcal{B}}$ to recurse on the rectangles as follows.
We begin building either $D_{\mathcal{A}}$ or $D_{\mathcal{B}}$, again considering both options, by collecting disks in decreasing order of radius, starting with $r_3$, until the collected weight suffices for recursion;
this may, depending on $r_3, r_4$ and the dimensions of $\mathcal{A}$ and $\mathcal{B}$, use Theorem~\ref{thm:mainRectangles} or Lemmas~\ref{lem:size-bound-large}~or~\ref{lem:rectanglesSB}.
We check that there is enough weight for this.
Finally, we use $r_3$ or $r_4$ to bound the cost of the split into $D_{\mathcal{A}}$ and $D_{\mathcal{B}}$, and verify that the remaining disks can be used to successfully recurse on the remaining region.

\subsubsection{Covering Using the Three Largest Disks}\label{sec:three-disk-pocket}\wcnewsubsec{strat:main-3ld}
In this section, we describe several routines that mainly rely on the three largest disks and recursion to cover $\mathcal{R}$.

\pwcstratref{strat:main-3ld}{1} If $r_3 > \frac{1}{2}$ is large enough to cover a strip of height $1$ and positive width, we consider placing $r_1,r_2,r_3$ horizontally next to each other as depicted in Fig.~\ref{fig:strategy_cover_with_three_disks}.
If this covers the entire rectangle, we are done; otherwise, we check whether we can guarantee successful recursion on the bounding box $\mathcal{A}$ of the remaining uncovered region.

\pwcstratref{strat:main-3ld}{2} We also consider placing $r_1,r_2,r_3$ as depicted in Fig.~\ref{fig:mainRectangles-three-disks}~(a), together covering a strip of maximum possible width, using recursion on the remaining rectangle.

\pwcstratref{strat:main-3ld}{3} If that does not work, we consider covering a vertical rectangular strip of height $1$ and maximum possible width at the left of $\mathcal{R}$ using $r_2$ and $r_3$; see Fig.~\ref{fig:mainRectangles-three-disks}~(b).
After placing $r_2$ and $r_3$ in that manner, we cover the remaining part of $\mathcal{R}$'s bottom side using $r_1$.
We discard this routine if any of these placements are infeasible and check whether we can guarantee successful recursion on the remaining rectangular region $\mathcal{A}$.

\pwcstratref{strat:main-3ld}{4} In the next routine, we cover the bottom-left corner of $\mathcal{R}$ using $r_3$.
We place $r_3$ such that its inscribed square is contained in $\mathcal{R}$, covering $\sqrt{2}r_3$ of the bottom and left side of $\mathcal{R}$; see Fig.~\ref{fig:strategy_three_disk_recursion_small_corner}.
We then consider two placements for $r_1$ and $r_2$ as depicted in Fig.~\ref{fig:strategy_three_disk_recursion_small_corner}.
A rectangle $\mathcal{A}$ in the top-right corner of $\mathcal{R}$ may remain uncovered by this placement; we check whether we can guarantee successful recursion on this rectangle using the remaining disks.

\pwcstratref{strat:main-3ld}{5} If this does not work, we also consider the routine \textsc{Three Disk Pocket}, placing $r_1,r_2,r_3$ as depicted in Fig.~\ref{fig:strategy_three_disk_recursion_r1_r3_pocket}.
In this routine, we cover a strip of width $S_1 \coloneqq 2\sqrt{r_1^2 - \frac{1}{4}}$ at the left side of the rectangle using $r_1$; if such a placement is impossible, we disregard the routine.
We then cover the uncovered part of the top side of the rectangle using the second-largest disk $r_2$ while maximizing the height $h_2 \coloneqq 2\sqrt{r_2^2 - \frac{1}{4}\left(\lambda-S_1\right)^2}$ of the rectangle it covers.
Finally, we place $r_3$ covering the uncovered part of the right side while maximizing the width $S_3 \coloneqq 2\sqrt{r_3^2 - \frac{1}{4}\left(1-h_2\right)^2}$ of the its covered rectangle.
This routine may leave an uncovered pocket between the largest and the third-largest disk; we check whether we can guarantee successful recursion on the bounding box of that pocket.

\subsubsection{Covering Using the Four Largest Disks}\wcnewsubsec{strat:main-4ld}
Finally, in this section, we describe three routines that rely mostly on the four largest disks to cover $\mathcal{R}$.

\pwcstratref{strat:main-4ld}{1} The first routine, depicted in Fig.~\ref{fig:mainRectangles-four-disks}~(a), places $r_1$ covering its inscribed square of side length $\sqrt{2}r_1$ in the bottom-left corner of $\mathcal{R}$.
On top of $r_1$, we place $r_4$ covering the remainder of $\mathcal{R}$'s left side.
To the right of $r_1$, we place $r_2$ covering the remainder of $\mathcal{R}$'s bottom side.
The routine succeeds if we can cover the entire remaining region in the top-right corner by $r_3$; note that this routine only uses the largest four disks.

\pwcstratref{strat:main-4ld}{2} The second routine, depicted in Fig.~\ref{fig:mainRectangles-four-disks}~(b), partitions $\{r_1,\ldots,r_4\}$ into two groups, each consisting of two disks.
We consider each possible partition.
We place the disks from each group on top of each other such that they cover two vertical strips of height $1$ and maximum possible width at $\mathcal{R}$'s left side.
If this covers $\mathcal{R}$, we are done; otherwise, we check whether we can guarantee successful recursion on the remaining strip $\mathcal{A}$.

\pwcstratref{strat:main-4ld}{3} In the third routine, depicted in Fig.~\ref{fig:mainRectangles-four-disks}~(c), we place $r_1$ and $r_2$ each covering a vertical strip of height $1$ and maximum possible width;
we disregard the routine if such placements are infeasible.
The strip covered by $r_1$ is placed at the left side of $\mathcal{R}$.
The strip covered by $r_2$ is placed to the right of the first strip; between the two strips, we leave a gap of width $g_w$ such that $r_1$ and $r_2$ still intersect each other.
Due to this gap, between $r_1,r_2$ and the top and bottom border of $\mathcal{R}$, there are two symmetric uncovered pockets.
We maximize $g_w$ such that $r_4$ suffices to cover one such pocket; we use $r_3$ to cover the other pocket.

If this leaves some uncovered region to the right of $r_2$, we proceed as follows.
If $r_2$ does not intersect the right boundary of $\mathcal{R}$, we check whether we can guarantee successful recursion on the bounding box of the remaining region.
Otherwise, the uncovered part of $\mathcal{R}$ consists of two disconnected pockets with bounding boxes $\mathcal{B}_1, \mathcal{B}_2$; see Fig.~\ref{fig:mainRectangles-four-disks}.
In that case, we also consider recursion on the bounding box $\mathcal{A}$ of the union of these pockets.
Additionally, we consider applying \textsc{Greedy Splitting} to the remaining disks, bounding the difference between the weight of the two resulting groups $D_{\mathcal{B}_1}, D_{\mathcal{B}_2}$ by $r_4^2$.
Using this bound, and $r_4$ as bound for the largest remaining disk, we check whether we can guarantee successful recursion on each of the remaining regions.

\subsubsection{Concluding the Proof}\label{sec:manual-analysis}
We implemented the success criteria of the routines described above using interval arithmetic and ran the resulting automatic prover.
This process yields a set of critical hypercuboids, for which we have to provide manual analysis to prove that the rectangle $\mathcal{R}$ can be covered in these cases as well.
For all other cases, our automatic analysis guarantees that there is at least one routine whose success criterion holds.
We classified these critical hypercuboids into the following two categories, which correspond to the two types of worst-case configurations depicted in Fig.~\ref{fig:worst-cases-rectangles}.
\begin{enumerate}
	\item[(I)] $1 \leq \lambda \leq 1.0359$, and $r_1^2,r_2^2,r_3^2 \in r^2_{*} + \left[-7 \cdot 10^{-5}, 2.5 \cdot 10^{-5}\right]$, where $r^2_{*} = \frac{\lambda^2}{16} + \frac{5}{32} + \frac{9}{256\lambda^2}$ is such that three disks of weight $r^2_{*}$ exactly suffice to cover $\mathcal{R}$ according to Lemma~\ref{lem:worst-cases-rectangles}.
	\item[(II)] $1.0356 \leq \lambda \leq \overline{\lambda}$, $r_1^2 \in \frac{\lambda^2+1}{4} + \left[-5 \cdot 10^{-5}, 3.5 \cdot 10^{-5}\right]$, and $r_2^2 \in \frac{1}{4} + \left[-5 \cdot 10^{-5}, 6.5 \cdot 10^{-7}\right]$.
\end{enumerate}
In the following, we handle these two cases.
\begin{lemma}\label{lem:mainRectangles-three-disk-pocket}
	For any input corresponding to a critical hypercuboid of type (I), the routine \textsc{Three Disk Pocket} from Section~\ref{sec:three-disk-pocket} allows us to cover $\mathcal{R}$.
\end{lemma}
\begin{lemma}\label{lem:mainRectangles-two-pockets}
	For any input corresponding to a critical hypercuboid of type (II), the routine \textsc{Two Pockets} from Section~\ref{sec:two-pockets} allows us to cover $\mathcal{R}$.
\end{lemma}
\begin{proof}[Proof of Lemma~\ref{lem:mainRectangles-three-disk-pocket}]
	For all critical hypercuboids of type (I), we let our automatic prover compute intervals for the values $S_1, h_2, S_3$ and $\lambda - S_1 - S_3$ occurring when placing the three largest disks according to routine \textsc{Three Disk Pocket}; see Fig.~\ref{fig:strategy_three_disk_recursion_r1_r3_pocket}.
	Combining these intervals, this yields the following bounds.
	For any collection $D$ of disks corresponding to a critical hypercuboid of type (I),
	when we place the three largest disks according to routine \textsc{Three Disk Pocket}, $r_1$ covers a strip of height $1$ and width $S_1 \in \left[0.1248,0.1560\right]$, $r_2$ covers a rectangle of width $\lambda-S_1$ and height $h_2 \in \left[0.4995,0.5002\right]$, and $r_3$ covers a rectangle of height $1-h_2$ and width $S_3 \in \left[0.8746,0.8801\right]$.
	The width of the remaining pocket is at most $\lambda - S_1 - S_3 \leq 0.00056$.
	If $r_1^2 \geq r_2^2 \geq r_3^2 \geq r^2_{*}$, then by Lemma~\ref{lem:worst-cases-rectangles}, the three largest disks, when placed in this manner, cover $\mathcal{R}$.
	Therefore, in the following, we assume $r_3^2 \leq r^2_{*}$.
	In order to prove that the remaining weight $R_4 = \sum_{i=4}^{n} r_i^2 \geq 0$ suffices to cover the remaining pocket of width $\lambda - S_1 - S_3 \geq 0$, we begin by bounding the height of the pocket in terms of its width.
	\begin{lemma}
		The height of the bounding box of the pocket that remains after placing the largest three disks according to \textsc{Three Disk Pocket} is at most $\frac{1}{4}\left(\lambda-S_1-S_3\right)$.
		\label{claim:strategy_three_disk_recursion_height_bound}
	\end{lemma}
	\begin{proof}
		To prove the lemma, we shoot a ray $g$ with slope $\frac{1}{4}$ through the right intersection point $P_1$ of $r_1$ with the bottom side of $\mathcal{R}$; the situation is depicted in Fig.~\ref{fig:strategy_three_disk_recursion_bound_height}.
		\begin{figure}
			\begin{center}
				\resizebox{.35\linewidth}{!}{\includegraphics{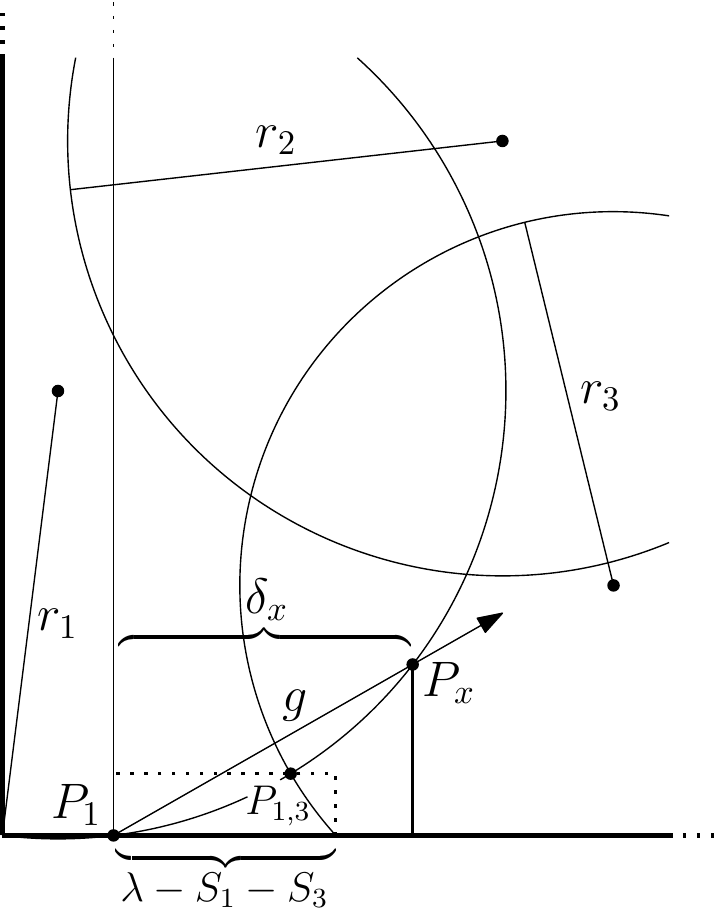}}
			\end{center}
			\caption{
				Shooting a ray $g$ through the intersection point $P_1$ to bound the height of the remaining pocket in terms of its width by proving that the distance $\delta_x$ is greater than $\lambda - S_1 - S_3$.
				To improve readability, the pocket in the figure is much larger than it actually is; to account for this fact, the slope in of $g$ in the figure is greater than $\frac{1}{4}$.
			}
			\label{fig:strategy_three_disk_recursion_bound_height}
		\end{figure}%
		This ray intersects the disk $r_1$ a second time at a point $P_x$ which is at distance $\delta_x = \frac{2\sqrt{68r_1^2 - 16 - 8S_1 - S_1^2}+2-8S_1}{17}$ to the right of $P_1$.
		We know that the intersection point $P_{1,3}$ of the largest and third-largest disk lies below $g$ as long as $\lambda - S_1 - S_3 \leq \delta_x$.
		By $r_1^2 = \frac{S_1^2+1}{4}$, we have $\delta_x = \frac{2\sqrt{16S_1^2 - 8S_1 + 1} + 2 - 8S_1}{17}$.
		Using $S_1 \in \left[0.1248,0.1560\right]$, we see that $\delta_x \geq 0.046 > 0.00056 \geq \lambda - S_1 - S_3$.
		Therefore, the height of the bounding box of the remaining pocket is at most $\frac{1}{4}$ its width.
	\end{proof}
	By induction using Theorem~\ref{thm:mainRectangles}, this means that to cover the bounding box of the remaining pocket, we need at most weight \[ R_4 \geq \frac{1}{4}\left(\lambda-S_1-S_3\right)^2 + \frac{1}{2} \cdot \frac{1}{16} \cdot \left(\lambda-S_1-S_3\right)^2 = \frac{9}{32}\left(\lambda-S_1-S_3\right)^2.\]
	In other words, it suffices to prove that $\Delta \coloneqq R_4 - \frac{9}{32}\max^2\left(0,\lambda-S_1-S_3\right) \geq 0.$

	We distinguish the three cases (1) $r_2^2 \leq r_1^2 \leq r_{*}^2$, (2) $r_1^2 \geq r_2^2 \geq r_{*}^2$, and (3) $r_2^2 \leq r^2_{*}, r_1^2 \geq r^2_{*}$.
	For case (1), we let $\varepsilon_1, \varepsilon_2, \varepsilon_3 \geq 0$, $r_1^2 = r^2_{*} - \varepsilon_1$, $r_2^2 = r^2_{*} - \varepsilon_1 - \varepsilon_2$ and $r_3^2 = r^2_{*} - \varepsilon_1 - \varepsilon_2 - \varepsilon_3$.
	In this case, the weight that remains for recursion is $R_4 = W^*(\lambda) - r_1^2 - r_2^2 - r_3^2 \geq 3\varepsilon_1 + 2\varepsilon_2 + \varepsilon_3$.
	By taking the partial derivatives
	\begin{align*}
		\frac{\partial\Delta}{\partial\varepsilon_3} &= 1 - \frac{9}{8} \cdot \max\left(0,\lambda-S_1-S_3\right)\frac{1}{S_3} \geq 0,\\
		\frac{\partial\Delta}{\partial\varepsilon_2} &= 2 - \frac{9}{8} \cdot \max\left(0,\lambda-S_1-S_3\right)\frac{1}{h_2S_3} \geq 0\textrm{, and}\\
		\frac{\partial\Delta}{\partial\varepsilon_1} &= 3 - \frac{9}{8} \cdot \max\left(0,\lambda-S_1-S_3\right)\frac{S_1+S_3+\lambda\left(\frac{1}{h_2}-1\right)}{S_1S_3} \geq 0
	\end{align*}
	of $\Delta$ w.r.t.\ $\varepsilon_1$, $\varepsilon_2$ and $\varepsilon_3$ and bounding their values using the bounds on $\lambda, S_1, h_2$ and $S_3$, we see that $\Delta$ is monotonically increasing in all three variables on the entire range.
	To minimize $\Delta$, we set $\varepsilon_1 = \varepsilon_2 = \varepsilon_3 = 0$; in that case we have $\lambda - S_1 - S_3 = 0$ due to Lemma~\ref{lem:worst-cases-rectangles} and thus $\Delta \geq 0$.

	For case (2), we let $\varepsilon_1, \varepsilon_2, \varepsilon_3 \geq 0$, $r_1^2 = r_{*}^2 + \varepsilon_1, r_2^2 = r_{*}^2 + \varepsilon_2$ and $r_3^2 = r_{*}^2 - \varepsilon_1 - \varepsilon_2 - \varepsilon_3$.
	In this case, we have $R_4 \geq \varepsilon_3$.
	We can again prove that $\Delta$ is monotonically increasing in $\varepsilon_1, \varepsilon_2$ and $\varepsilon_3$ by analyzing the partial derivatives of $\Delta$; note that in this case, the largest and second largest disks have radius at least $r_{*}^2$ and thus we have $h_2 \geq \frac{1}{2}$ as shown in the proof of Lemma~\ref{lem:worst-cases-rectangles}.
	We then get
	\begin{align*}
		\frac{\partial\Delta}{\partial\varepsilon_3} &= 1 - \frac{9}{8}\max\left(0,\lambda-S_1-S_3\right)\frac{1}{S_3} \geq 0,\\
		\frac{\partial\Delta}{\partial\varepsilon_2} &= \frac{9}{8}\max\left(0,\lambda-S_1-S_3\right) \cdot \frac{1}{S_3} \cdot \left(2-\frac{1}{h_2}\right) \geq 0\textrm{, because }h_2\geq\frac{1}{2}\textrm{, and}\\
		\frac{\partial\Delta}{\partial\varepsilon_1} &= \frac{9}{8}\max\left(0,\lambda-S_1-S_3\right) \cdot \frac{S_1 + h_2S_3 - \left(1-h_2\right)\lambda}{S_1h_2S_3} \geq 0\textrm{, because }\\
		&S_1 + h_2S_3 - \left(1-h_2\right)\lambda \geq S_1 + \frac{1}{2}S_3 - \frac{1}{2}\lambda \geq 0 \Leftrightarrow 2S_1 + S_3 \geq \lambda.
	\end{align*}

	For case (3), we let $\varepsilon_1, R_4 \geq 0$ and $c \in \left[0,\frac{1}{2}\right]$ and let $r_1^2 = r_{*}^2 + \varepsilon_1, r_2^2 = r^2_{*} - c\left(\varepsilon_1 + R_4\right)$ and $r_3^2 = r_{*}^2 - (1-c)\cdot\left(\varepsilon_1+R_4\right)$.
	Again, we can prove that $\Delta$ is monotonically increasing in $\varepsilon_1$ and $R_4$ by analyzing its partial derivatives
	\begin{align*}
		\frac{\partial\Delta}{\partial R_4} &= 1 - \frac{9}{8}\max\left(0,\lambda-S_1-S_3\right)\cdot\frac{1}{S_3}\cdot\left(1-2c+\frac{c}{h_2}\right) > 0\textrm{, and}\\
		\frac{\partial\Delta}{\partial \varepsilon_1} &= \frac{9}{8} \cdot \max\left(0,\lambda-S_1-S_3\right) \cdot \left(\frac{1}{S_1} + \frac{1}{S_3}\cdot\left(-1+c+(1-h_2)\frac{\lambda-cS_1-S_1}{S_1h_2}\right)\right) \geq 0.
	\end{align*}
	In all three cases, we obtain that $\Delta$ is minimized exactly for the case $r_1^2 = r_2^2 = r_3^2 = r_{*}^2$ depicted in Fig.~\ref{fig:strategy_three_disk_recursion_r1_r3_pocket}, where $\lambda - S_1 - S_3 = 0$ and thus $\Delta \geq 0$, with $\Delta = 0$ for skew $\lambda \leq \lambda_2$ below the break-even point.
	This concludes the proof for critical hypercuboids of type (I).
\end{proof}
\begin{proof}[Proof of Lemma~\ref{lem:mainRectangles-two-pockets}]
	For a collection $D$ of disks corresponding to a critical hypercuboid of type (II), when we place the largest disk according to routine \textsc{Two Pockets} as depicted in Fig.~\ref{fig:mainRectangles-largest-disk}~(c), $r_1$ intersects the right border of $\mathcal{R}$, covering a piece of length $h_1 \geq 0.9888$ of that side, possibly leaving two symmetric pockets of width $\lambda-S_1$ and height $\frac{1-h_1}{2}$ uncovered.
	Moreover, $r_1$ covers a strip of height $1$ and width $S_1$.
	If $r_2^2 \geq \frac{1}{4}\left(1+\left(\lambda-S_1\right)^2\right)$, we can cover both pockets using $r_2$ and \textsc{Two Pockets} succeeds.
	Therefore we only need to consider the case that $r_1^2 = \frac{\lambda^2+1}{4} - \varepsilon_1$ and $r_2^2 = \frac{1}{4}\left(1+\left(\lambda-S_1\right)^2\right) - \varepsilon_2$ for some $\varepsilon_1,\varepsilon_2 \geq 0$.
	We thus have \[S_1 = 2\sqrt{r_1^2 - \frac{1}{4}} = 2\sqrt{\frac{\lambda^2}{4} - \varepsilon_1}\textrm{ and }h_1 = 2\sqrt{r_1^2 - \left(\lambda - \frac{S_1}{2}\right)^2} = 2\sqrt{\frac{1}{4} - \lambda^2 + \lambda S_1}.\]
	
	For inputs corresponding to type (II), $r_2$ is large enough to cover one of the pockets on its own.
	The other pocket is partially covered by $r_2$ as well; we ignore this and cover the whole remaining pocket recursively.
	We begin by bounding the width of the remaining pocket in terms of its height; to be precise, we prove that its width is at most $\frac{1}{\lambda}$ times its height.
	In other words, we want to prove $\Delta_s := \frac{1-h_1}{2} - \lambda\left(\lambda-S_1\right) \geq 0$.
	By considering the partial derivative 
	$$\frac{\partial S_1}{\partial \varepsilon_1} = \frac{-2}{S_1}, \frac{\partial h_1}{\partial \varepsilon_1} = \frac{-4\lambda}{S_1h_1}, \frac{\partial \Delta_s}{\partial \varepsilon_1} = \frac{2\lambda\left(1-h_1\right)}{S_1h_1} \geq 0,$$
	we see that $\Delta_s$ is minimized for $\varepsilon_1 = 0$.
	In that case, we have $S_1 = \lambda, h_1 = 1$ and $\Delta_s = 0$ independent of $\lambda$; therefore, $\Delta_s \geq 0$ holds.

	To make use of this bound, we increase the width of the pocket to $W_p := \frac{1-h_1}{2\lambda}$ and cover the slightly larger $\left(W_p \times \frac{1-h_1}{2}\right)$-rectangle instead of the bounding box of the original pocket.
	Because of $\lambda \leq \overline{\lambda}$, we can cover this rectangle with disks of weight $\frac{195}{256} \cdot W_p\frac{1-h_1}{2} = \frac{195\left(1-h_1\right)^2}{1024\lambda}$ according to Theorem~\ref{thm:mainRectangles}.
	Therefore it suffices to prove that the remaining weight
	\begin{align*}%
		R_3 &= W^*(\lambda) - r_1^2 - r_2^2 \geq \frac{\lambda^2+2}{4} - \frac{\lambda^2 + 1}{4} + \varepsilon_1 - \frac{1}{4}\left(1+\left(\lambda-S_1\right)^2\right) + \varepsilon_2\\
		&= \varepsilon_1 + \varepsilon_2 - \frac{1}{4}\left(\lambda-S_1\right)^2 = 2\varepsilon_1 + \varepsilon_2 - \frac{1}{2}\lambda\left(\lambda-S_1\right)
	\end{align*}
	satisfies $\Delta := R_3 - \frac{195\left(1-h_1\right)^2}{1024\lambda} \geq 0$.
	$\Delta$ is monotonically increasing in $\varepsilon_2$, so we set $\varepsilon_2 = 0$ to minimize $\Delta$.
	Regarding $\varepsilon_1$, we have
		\[\frac{\partial R_3}{\partial \varepsilon_1} = 2 - \underbrace{\frac{\lambda}{S_1}}_{\leq 1.005}\textrm{ and }\frac{\partial \Delta}{\partial \varepsilon_1} = 2 - \frac{\lambda}{S_1} - \underbrace{\frac{195\left(1-h_1\right)}{128S_1h_1}}_{\leq 0.017} > 0,\]
	so $\Delta$ is monotonically increasing in $\varepsilon_1$ as well; we set $\varepsilon_1 = 0$ to minimize $\Delta$.
	For $\varepsilon_1 = 0$, we have $S_1 = \lambda$ and $h_1 = 1$, and $R_3 = \Delta = 0$; therefore, we have $\Delta \geq 0$ for all $\varepsilon_1,\varepsilon_2$ under consideration.
	This implies that, for any input corresponding to a critical hypercuboid of type (II), we can either cover $\mathcal{R}$ using only $r_1,r_2$ or the remaining weight $R_3$ is sufficient to cover one of the pockets left uncovered by $r_1$.
\end{proof}%
}

\def\prooflemmaworstcasesrectangles{%
\subsection{Proof of Lemma~\ref{lem:worst-cases-rectangles}}
\label{sec:proof-lemma-worst-cases-rectangles}
In this section, we give the proof of Lemma~\ref{lem:worst-cases-rectangles}; see Fig.~\ref{fig:worst-cases-rectangles}.
\restatethm{\lemmaworstcasesrectangles*}{lem:worst-cases-rectangles}
\begin{proof}
	(1) is clear because $r_1^2$ is the weight of $\mathcal{R}$'s circumcircle.
	Regarding~(2), any covering of $\mathcal{R}$ has to cover all four corners of $\mathcal{R}$.
	The larger disk $r_1$ can only cover at most two corners; the smaller disk $r_2$ can only cover two corners at distance $1$ from each other.
	Therefore $r_1$ must be placed covering two corners at distance $1$ from each other; w.l.o.g., let this be the corners on $\mathcal{R}$'s left side.
	After placing $r_1$ in this way, the area that remains to be covered includes the right corners and two points $p_1, p_2$ on the top and the bottom side at some distance $\delta_1,\delta_2 > 0$ from $\mathcal{R}$'s right side.
	The smaller disk $r_2$ cannot cover these four points.

	Regarding~(3), we place the first disk covering a strip of width $S_1 = \sqrt{4r^2-1} = \sqrt{\frac{\lambda^2}{4} - \frac{3}{8} + \frac{9}{64\lambda^2}}$ as depicted in Fig.~\ref{fig:worst-cases-rectangles}.
	We place the second disk covering a rectangular strip of height $h_2$ of the remaining rectangle.
	The height covered in this way is \begin{align*}
		h_2 &= \sqrt{4r^2 - \left(\lambda-S_1\right)^2} = \sqrt{\frac{\lambda^2}{4} + \frac{5}{8} + \frac{9}{64\lambda^2} - \lambda^2 + 2\lambda S_1 - S_1^2} = \sqrt{1 - \lambda^2 + 2\lambda S_1}\\
		&= \sqrt{1 - \lambda^2 + 2\sqrt{\frac{\lambda^4}{4}-\frac{3\lambda^2}{8} + \frac{9}{64}}} = \sqrt{1-\lambda^2+\sqrt{\lambda^4-\frac{3}{2}\lambda^2+\frac{9}{16}}}\\
		&= \sqrt{1-\lambda^2+\lambda^2-\frac{3}{4}} = \frac{1}{2},
	\end{align*}
	so the second and the third disk suffice to fully cover the remainder of $\mathcal{R}$.

	Regarding~(4), we have to cover all four corners of $\mathcal{R}$ using three identical disks.
	Therefore, w.l.o.g., disk $r_1$ has to be placed covering two corners; these corners can either be at distance $1$ or at distance $\lambda$ from each other.

	In the first case, depicted in Fig.~\ref{fig:worst-cases-rectangles}, assume w.l.o.g. that $r_1$ covers $\mathcal{R}$'s left corners.
	We argue that the three disks cannot cover the entire boundary of $\mathcal{R}$.
	Therefore, w.l.o.g., we may assume that $r_1$ is placed such that it touches at least one of $\mathcal{R}'s$ corners; otherwise, we could push $r_1$ to the left until it does.
	This cannot lead to a previously covered point of $\mathcal{R}$'s boundary becoming uncovered.
	In the following, we argue that we can assume the center $c_1$ to be at $c_1^*$, at height $\frac{1}{2}$ above $\mathcal{R}$'s bottom side and at width $y^* = \sqrt{r_-^2 - 1/4}$ right of $\mathcal{R}$'s left side; see Fig.~\ref{fig:proof-tightness-3d}.
	\begin{figure}[ht!]
		\begin{center}
			\resizebox{.75\textwidth}{!}{\includegraphics{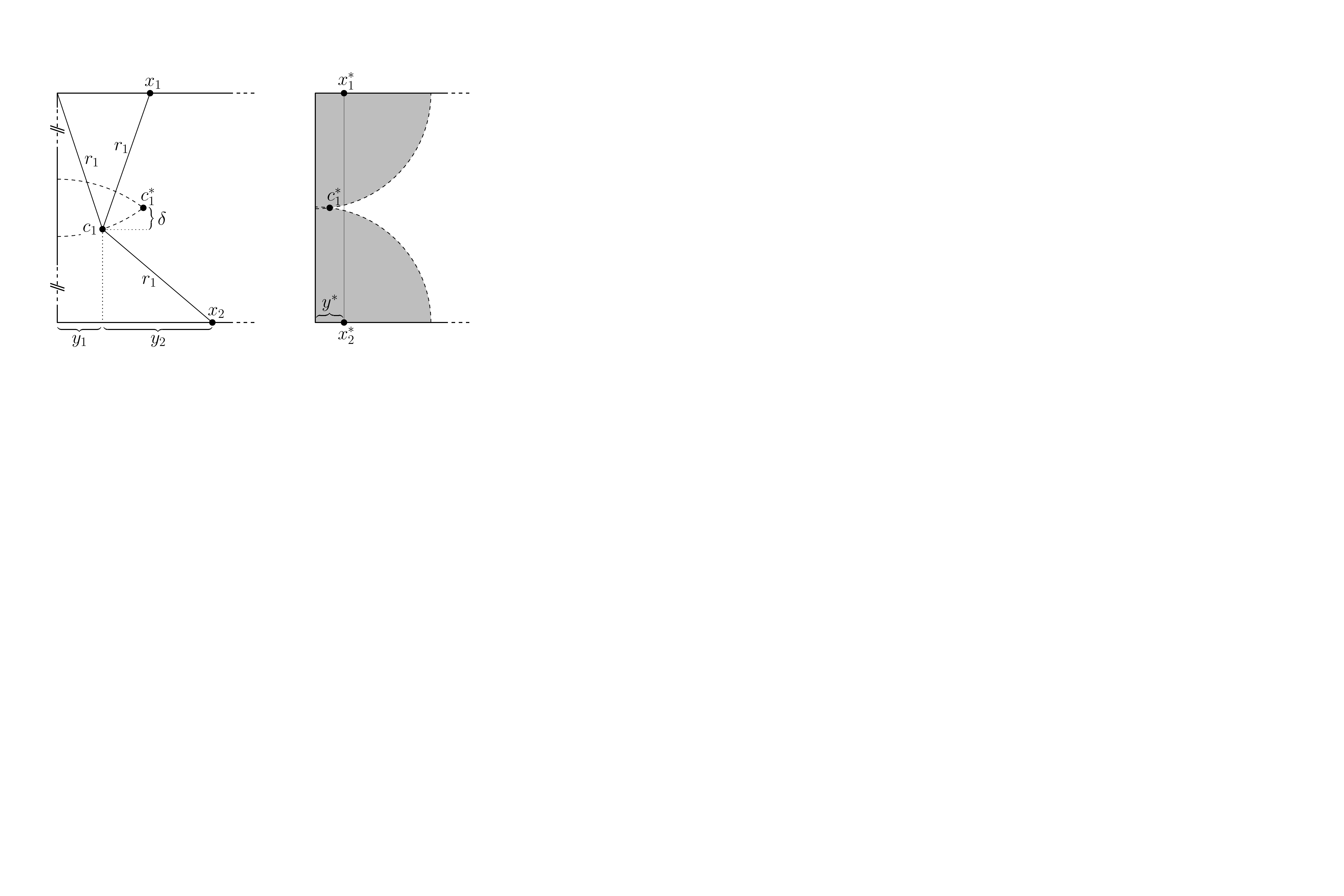}}
		\end{center}
		\caption{
			The first disk $r_1$ must cover $\mathcal{R}$'s left corners.
			Its center $c_1$ must therefore must be within distance $r_-$ of both corners (left, dashed outline).
			We show that placing its center at $c_1^*$, which is at height exactly $\frac{1}{2}$, is optimal w.r.t.\ the part of $\mathcal{R}$'s boundary that remains to be covered after placing $r_1$.
			In other words, we show that moving $c_1$ down as depicted on the right by some $\delta > 0$ moves \emph{both} $x_1$ and $x_2$ further to the left.
		}
		\label{fig:proof-tightness-3d}
	\end{figure}
	Towards this goal, we consider moving $c_1$ down by some $\delta > 0$, as depicted in Fig.~\ref{fig:proof-tightness-3d}; the case of moving $c_1$ up is symmetric.
	It is straightforward to see that this causes the upper intersection point $x_1$ of $\mathcal{R}$ and $r_1$ to move to the left compared to the upper intersection point $x_1^*$ obtained by placing $r_1$ at $c_1^*$.
	Regarding the lower intersection point $x_2$ of $\mathcal{R}$ and $r_1$, its distance to $\mathcal{R}$'s left border is $y \coloneqq y_1 + y_2 = \sqrt{r_-^2 - (1/2 + \delta)^2} + \sqrt{r_-^2 - (1/2 - \delta)^2}$.
	We have \[\frac{dy}{d\delta} = \frac{-\frac{1}{2}-\delta}{y_1} + \frac{\frac{1}{2}-\delta}{y_2} = \underbrace{\left(\frac{1}{2y_2} - \frac{1}{2y_1}\right)}_{<0} - \frac{\delta}{y_1} - \frac{\delta}{y_2} < 0,\]
	which implies that moving $c_1$ down also moves $x_2$ to the left; thus, placing $c_1$ at $c_1^*$ is (inclusion-wise) strictly better than any other placement of $r_1$ w.r.t.\ the set of boundary points that we cover.

	After placing $r_1$ in this manner, the remaining two disks have to cover $x_1,x_2$ and the right corners of $\mathcal{R}$; each disk can cover at most two of these points.
	Moreover, if one disk is placed covering $x_1$ and $x_2$ or both right corners, the remaining region is too large to be covered by the third disk.
	Therefore, the first remaining disk $r_2$ has to cover $x_1$ and $\mathcal{R}$'s upper right corner and the last disk $r_3$ has to cover $x_2$ and $\mathcal{R}$'s lower right corner and the intersection point $x_3$ of $r_2$ and $\mathcal{R}$'s right side.
	W.l.o.g., we may assume that $r_2$ is placed such that it touches $x_1$ and the upper right corner; otherwise, $x_3$ moves upwards.
	As discussed for Statement~(3), if $r_1 = r_2 = r$, when placed in this way, $r_2$ covers a sub-rectangle of height $\frac{1}{2}$; because all our disks are smaller than $r$, the covered height is $h_2 < \frac{1}{2}$.
	This implies that $r_3$ cannot cover both $x_2$ and $x_3$; the distance between these two points is greater than $2r > 2r_-$.

	An analogous argument works for the case that $r_1$ covers two corners at distance $\lambda$ from each other.
\end{proof}%
}

\def\prooflemmasizeboundlarge{%
	\subsection{Proof of Lemma~\ref{lem:size-bound-large}}
	\label{sec:proof-lemma-size-bound-large}
	In this section, we give a proof of Lemma~\ref{lem:size-bound-large}.
	\restatethm{\lemmasizeboundlarge*}{lem:size-bound-large}
	\begin{proof}
		\revsocg{
			In the following, let $E \coloneqq E(\sigma)$; we assume w.l.o.g.\ that $W(D) = E\lambda$.
			First, we observe that $\sigma \geq \hat{\sigma}$ implies $E \geq \frac{195}{256} = E^*(\bar{\lambda})$.
			Because $E^*(\lambda)$ for $\lambda \geq \bar{\lambda}$ is continuous and strictly monotonically increasing,
			there is a unique $\Lambda(E) \geq \bar{\lambda}$ such that $E^*(\Lambda(E)) = E$, given by $\Lambda(E) := 2E + \sqrt{4E^2-2}$.
			Similarly, we observe that $\sigma(E) = E \cdot \left(\Lambda(E) - \frac{2}{\Lambda(E)}\right)$ is the inverse function of $E(\sigma)$.
			If $\lambda \leq \Lambda(E)$, we have $E \geq E^*(\lambda)$ and the result immediately follows from Theorem~\ref{thm:mainRectangles}.

			Otherwise, we apply \textsc{Greedy Splitting} to $D$.
			This yields a partition into two groups $D_1,D_2$; w.l.o.g., let $D_1$ be the heavier one.
			We split $\mathcal{R}$ into two rectangles $\mathcal{R}_1,\mathcal{R}_2$ such that $\frac{W(D_1)}{W(D_2)} = \frac{\abs*{\mathcal{R}_1}}{\abs*{\mathcal{R}_2}}$
			by dividing the longer side (w.l.o.g., the width) of $\mathcal{R}$ in that ratio.
			After the split, we have $E = \frac{W(D)}{\abs*{\mathcal{R}}} = \frac{W(D_1)}{\abs*{\mathcal{R}_1}} = \frac{W(D_2)}{\abs*{\mathcal{R}_2}}$ and $\abs*{\mathcal{R}_2} = \frac{W(D_2)}{E}$.

			If the resulting width of any $\mathcal{R}_i$ is greater than $\Lambda(E)$, we use $D_i$ to inductively apply Lemma~\ref{lem:size-bound-large} to it.
			Otherwise, we apply Theorem~\ref{thm:mainRectangles}; in order to do so, we must show that the skew of the narrower rectangle $\mathcal{R}_2$ is at most $\Lambda(E)$,
			which means proving that its width is at least $\frac{1}{\Lambda(E)}$.
			Because of $W(D_1)-W(D_2) \leq r_1^2 \leq \sigma$, we have $W(D_2) \geq \frac{W(D)-\sigma}{2} = \frac{E\lambda-\sigma(E)}{2}$.
			This implies that the area, and thus the width, of $\mathcal{R}_2$ is $\frac{W(D_2)}{E} \geq \frac{\Lambda(E)-\sigma(E)/E}{2} = \frac{1}{\Lambda(E)}$.
		}
	\end{proof}%
}

\def\proofdetailsweightboundedcovering{%
	\subsection{Weight-Bounded Covering --- Proof of Lemma~\ref{lem:rectanglesSB}}
\label{sec:proof-rectanglesSB}
In this section, we prove Lemma~\ref{lem:rectanglesSB}.
\restatethm{\lemmarectanglessb*}{lem:rectanglesSB}
For the remainder of this proof, let $E = \rectanglesSBEff$ be the covering coefficient Lemma~\ref{lem:rectanglesSB} guarantees.
Assume that all disks in $D$ have radius at most $\rectanglesSBRB$ and total weight $W(D) = E\lambda$.
In this case, our \revsocg{success criteria} consider $\lambda$ and the up to seven largest disk weights $r_1^2,\ldots,r_7^2$.
We make the following observation.
\begin{observation}
\label{obs:at-least-5}
	Due to the weight bound of Lemma~\ref{lem:rectanglesSB}, at least \(\left\lceil\frac{E}{\rectanglesSBRB^2}\right\rceil= 5\) disks are always present.
\end{observation}
For disks $r_6$ and $r_7$, \revsocg{recall that we} also consider the cases of $r_6 = r_7 = 0, R_8 \coloneqq \sum_{i=8}^{n}r_i^2 = 0$ or $r_7 = 0, R_8 = 0$;
\revsocg{in this way, we handle the induction base and step simultaneously.}
In the following, we describe the routines used by our algorithm.
If they are not straightforward, we also describe \revsocg{the success criteria} by which we ensure that the routine works given only the seven largest disks and $\lambda$.

\revsocg{In order to apply Lemma~\ref{lem:rectanglesSB} to a rectangle $\mathcal{A}$,
we have to scale the rectangle and the disks such that $\mathcal{A}$'s shorter side has length $1$.
Therefore, the radius bound required to apply the lemma to $\mathcal{A}$ depends on its shorter side.

\begin{definition}
	Let $\mathcal{A}$ be any rectangle of dimensions $\beta_{\mathcal{A}} \times \gamma_{\mathcal{A}}$.
	By \[\rho({\mathcal{A}}) = \rectanglesSBRB \cdot \min\left(\beta_{\mathcal{A}},\gamma_{\mathcal{A}}\right)\] we denote the radius bound required to apply Lemma~\ref{lem:rectanglesSB} to $\mathcal{A}$.
	If $r_i \leq \rho(\mathcal{A})$ holds for a disk $r_i$, we say that $r_i$ \emph{satisfies the radius bound w.r.t.\ $\mathcal{A}$} or $r_1$ \emph{fits} $\mathcal{A}$.
	Similarly, for a collection $D_{\mathcal{A}}$ of disks with largest disk $r_1$, we say that $D_{\mathcal{A}}$ \emph{fits} $\mathcal{A}$ if $r_1$ fits $\mathcal{A}$.
\end{definition}
}

\subsubsection{Balanced and Unbalanced Recursive Splitting}\sbnewsubsec{strat:rectanglesSB-eurs}
\psbstratref{strat:rectanglesSB-eurs}{1} The first routine is based on splitting $D$ into two groups of approximately the same total weight and recursing.
In order to do this, we first compute the partition of the seven largest disks into two groups $D'_1,D'_2$ minimizing the difference $\Delta \coloneqq \abs*{W(D'_1) - W(D'_2)}$.
Starting with this subdivision, we apply \textsc{Greedy Splitting} to distribute the remaining disks to the two groups, resulting in a partition of the disks into collections $D_1$ and $D_2$.
We subdivide $\mathcal{R}$ into two rectangles $\mathcal{R}_1,\mathcal{R}_2$ according to the weights $W(D_i)$ such that
$E = \frac{W(D)}{\abs*{\mathcal{R}}} = \frac{W(D_1)}{\abs*{\mathcal{R}_1}} = \frac{W(D_2)}{\abs*{\mathcal{R}_2}}$ \revsocg{by splitting $\mathcal{R}$'s width in that ratio
and recursively apply Lemma~\ref{lem:rectanglesSB} on $\mathcal{R}_i$ using disks $D_i$}.

\revsocg{Our success criterion is as follows.
Because the rectangle is split according to the actual weights $W(D_i)$ such that both sides require a covering coefficient of exactly $E$, we do not waste any weight, i.e., the disk weight always suffices to recurse.
We only have to ensure that the largest disk in each $D_i$ satisfies the radius bound of Lemma~\ref{lem:rectanglesSB} w.r.t.\ $\mathcal{R}_i$.
In particular, because we cannot assume $r_1$ to end up in the larger group, we want to ensure that $r_1$ fits the radius bound w.r.t.\ the smallest possible $\mathcal{R}_i$.
We obtain an upper bound $\gamma \geq |W(D_1)-W(D_2)|$ as follows.
If $\Delta > R_8$, \textsc{Greedy Splitting} adds all remaining disks to the smaller of $D'_1, D'_2$; we thus know the exact subdivision $D_1,D_2$ and can compute $\gamma$ accordingly.
Otherwise, \textsc{Greedy Splitting} either still adds all remaining disks to the smaller of $D'_1, D'_2$, turning it into the larger group, or adds disks to both $D'_1,D'_2$.
In either case, the smallest disk in the larger group is at most $r_7^2$; therefore, we may use $\gamma = r_7^2$.
We then use our bound $\gamma$ to compute a lower bound $\frac{\lambda - \gamma/E}{2}$ on the width of $\mathcal{R}_i$ and check whether $r_1$ satisfies the resulting radius bound.
Note that we can analogously formulate the success criterion using fewer than $7$ disks.
}

\psbstratref{strat:rectanglesSB-eurs}{2} \revsocg{The second routine is also based on splitting $D$ into two groups; however, in this case, we do not try to make the subdivision as balanced as possible.
Instead, we compute the minimum $\ell$ such that $r_1$ satisfies the radius bound of Lemma~\ref{lem:rectanglesSB} w.r.t.\ a rectangle with shorter side $\ell$.
Starting with $r_1$, we collect disks in $D_1$ until $W(D_1) \geq E\ell$; all remaining disks are placed in $D_2$, which must not be empty.
We again split $\mathcal{R}$ into $\mathcal{R}_1,\mathcal{R}_2$ according to the weights $W(D_1),W(D_2)$, and recurse using Lemma~\ref{lem:rectanglesSB}.

Our success criterion is as follows.
Again, we do not waste any disk weight; thus, recursion cannot fail due to weight.
Moreover, $r_1$ satisfies the radius bound w.r.t.\ $\mathcal{R}_1$ by construction.
However, we have to ensure that $D_2$ is nonempty and that its largest disk $r_k$ satisfies the radius bound w.r.t.\ $\mathcal{R}_2$.
We begin by computing $\ell = \frac{r_1}{\rectanglesSBRB}$.
We need to find (i) a bound $W_{\ell} \geq W(D_1)$ on the weight of $D_1$ and (ii) a bound on $r_k$.

We add disks from $\{ r_1^2, \ldots, r_7^2 \}$ to $D_1$ until we either (a) run out of disks or (b) exceed a total weight of $E\ell$.
In case (a), the routine keeps adding disks until $E\ell$ is exceeded, at which point $W(D_1) \leq E\ell + r_7^2 \eqqcolon W_{\ell}$, and $r_k \leq r_7$.
In case (b), we compute $W_{\ell}$ from the disks added to $W(D_1)$, and use $r_{i+1}$ as bound on $r_k$, where $r_i$ is the last disk possibly added to $D_1$.
In either case, we can exclude $D_2 = \emptyset$ if $W_{\ell} < E\lambda$.
Moreover, we use $W_{\ell}$ to compute a lower bound on $W(D_2)$ and the width $\beta_{\mathcal{R}_2}$ of $\mathcal{R}_2$, and check whether our bound on $r_k$ satisfies the radius bound w.r.t.\ $\beta_{\mathcal{R}_2}$.
}

\subsubsection{Building a Strip}\sbnewsubsec{strat:rectanglesSB-bas}
\psbstratref{strat:rectanglesSB-bas}{1} \revsocg{The next routine is as follows.
We choose some disks $\emptyset \neq T \subseteq \{r_1,\ldots,r_7\}$ and cover a rectangular strip $\mathcal{S}$ of area $\abs*{\mathcal{S}} = \frac{W(T)}{E}$.
In the following, we assume that $S$ is vertical, i.e., of height $1$ and width $\frac{W(T)}{E}$; the case of horizontal strips is analogous; see Fig.~\ref{fig:rectangleSB-vertical-strip} for examples.
After covering $\mathcal{S}$, we inductively apply Lemma~\ref{lem:rectanglesSB} on the remaining rectangle $\mathcal{R}_2 \coloneqq \mathcal{R}\setminus \mathcal{S}$.
To cover $\mathcal{S}$, we consider all possible subdivisions of $T$ into up to $\abs{T}$ rows.
Each row is then built from its disks by placing the disks covering a rectangle of width corresponding to the strip's width and maximum possible height.
In addition to these coverings of $\mathcal{S}$ based on multiple rows, we also check the configuration depicted in Fig.~\ref{fig:rectangleSB-vertical-strip}~(d).

Our success criterion is as follows.
Because the covering coefficient achieved on $\mathcal{S}$ is $E$, recursing on $\mathcal{R}_2$ can not fail due to missing disk weight.
However, we have to check whether $\mathcal{S}$ can be covered in the manners mentioned above,
which is straightforward because it only involves a constant number of known disks placed according to a constant number of possible configurations.
Additionally, we check that the largest disk not in $T$ --- or $r_7$, if $T = \{r_1,\ldots,r_7\}$ --- satisfies the radius bound w.r.t. $\mathcal{R}_2$.
}

\subsubsection{Wall Building}\sbnewsubsec{strat:rectanglesSB-wb}
\psbstratref{strat:rectanglesSB-wb}{1} The next routine is based on the idea of covering a rectangular strip of fixed length $\ell$ and variable width $b$ by stacking disks on top of each other, using the following lemma; see Fig.~\ref{fig:rectangleSB-wall-building}.
Intuitively speaking, this works under the following conditions.
(1) The largest disk is not too large when compared to the length of the strip.
(2) The weight of the individual disks does not decrease too much before (3) a certain total weight is exceeded.
\begin{restatable}{lemma}{lemmawallbuilding}
	\label{lem:rectangleSB-wall-building}
	Let $E > \frac{1}{2}$ be some fixed covering coefficient that we want to realize.
	Let $\ell > 0$ be some fixed strip length.
	Let $q_1 \geq q_2 \geq \ldots \geq q_m > 0$ be a sequence of disk radii such that \begin{align*}
		\textup{(1) }&q_1 \leq \frac{\ell}{\sqrt{2}} \cdot \left(1 - \frac{1}{\sqrt{1+\sqrt{1-\frac{1}{4E^2}}}}\right)\textup{, (2) }q_m \geq q_1 \cdot \sqrt{4E^2-2E\sqrt{4E^2-1}}\textup{, and}\\
		\textup{(3) }&\sum\limits_{i=1}^{m} q_i^2 \geq \sqrt{2} \cdot q_1 \cdot E \cdot L(\ell)\textup{, where } L(\ell) \coloneqq \frac{\ell}{\sqrt{1+\sqrt{1-\frac{1}{4E^2}}}}.
	\end{align*}
	Then there is some $\sqrt{2}q_1 \geq b > 0$ such that we can cover a rectangular strip of dimensions $\ell \times b$ with disks from $q_1,\ldots,q_m$, using no more than $E \cdot \ell \cdot b$ weight in total.
\end{restatable}
\begin{proof}
	We defer the proof of Lemma~\ref{lem:rectangleSB-wall-building} to Section~\ref{sec:proof-lemma-wall-building}.
\end{proof}

\revsocg{
	Starting with disk $r_7$, in non-increasing order of weight, we search for a sequence $q_1,\ldots,q_m$ of consecutive disks satisfying the conditions of 
	Lemma~\ref{lem:rectangleSB-wall-building} for covering coefficient $E$ and a strip of length $\ell = 1$; see Fig.~\ref{fig:rectangleSB-wall-building}.
	We then determine the width $b$ according to the lemma, subdivide $\mathcal{R}$ into two rectangles $\mathcal{R}_1$ of dimensions $b \times 1$ and $\mathcal{R}_2$ of dimensions $\left(\lambda-b\right) \times 1$,
	cover $\mathcal{R}_1$ using Lemma~\ref{lem:rectangleSB-wall-building} and recurse on $\mathcal{R}_2$ using the remaining disks if $r_1$ satisfies the radius bound w.r.t. $\mathcal{R}_2$.

	If we do not find such a sequence, or if $r_1$ does not satisfy the radius bound w.r.t. the resulting $\mathcal{R}_2$, we collect disks in $D_1$, starting from the \emph{smallest} disk $r_n$.
	After adding a disk $r_i$ we compute the width $\beta_i = \frac{r_i}{\rectanglesSBRB}$ of the smallest rectangle $\mathcal{R}_i$ of height $1$ such that $r_i$ satisfies the radius bound of Lemma~\ref{lem:rectanglesSB} w.r.t.\ $\mathcal{R}_i$.
	If $W(D_1) \geq E\beta_i$ holds after adding disk $r_i$ to $D_1$, we subdivide $\mathcal{R}$ into two rectangles of widths $\frac{W(D_1)}{E}$ and $\lambda - \frac{W(D_1)}{E}$ and recurse using $D_1$ and $D \setminus D_1$.

	Because the routine heavily depends on disks $r_7,\ldots,r_n$, it is not straightforward to find a success criterion for it; we use the criterion given by the following lemma.
	\begin{lemma}
		\label{lem:wall-building-suff-crit}
		Let $s \coloneqq \sqrt{4E^2-2E\sqrt{4E^2-1}} \approx 0.7974$ and let $L(\ell)$ be defined as in Lemma~\ref{lem:rectangleSB-wall-building}.
		For $k \in \mathbb{N}$, let  $t_k \coloneqq r_7 \cdot s^{k+1}$ and \[a_k \coloneqq E\lambda - \sum\limits_{i=1}^{6}r_i^2 - \sqrt{2}r_7 \cdot E \cdot L(1) \cdot \sum\limits_{i=0}^{k}s^i\text{.}\]
		Routine~\sbstratref{strat:rectanglesSB-wb}{1} is guaranteed to be successful if the following conditions hold.
		\[ \text{(1)\ } r_1 \leq \rectanglesSBRB \cdot \left(\lambda-\sqrt{2}r_7\right)\text{,}\]
		\[ \text{(2)\ } r_7 \leq \frac{1}{\sqrt{2}} \cdot \left(1 - \frac{1}{\sqrt{1+\sqrt{1-\frac{1}{4E^2}}}}\right) \approx 0.1433 \text{, and }\]
		\[ \text{(3)\ there is $0 \leq k \in \mathbb{N}$ such that }a_k \geq \frac{t_k\cdot E}{\rectanglesSBRB}\text{ and } r_1 \leq \rectanglesSBRB\left(\lambda - \frac{t_k}{\rectanglesSBRB} - \frac{t_k^2}{E}\right).\]
	\end{lemma}
	\begin{proof}
		The largest width $b$ that could result from Lemma~\ref{lem:rectangleSB-wall-building} using a sequence $q_1,\ldots,q_m$ with $q_1 \leq r_7$ is $\sqrt{2}r_7$.
		By Condition~(1), $r_1$ satisfies the radius bound w.r.t. a rectangle of width $\lambda - \sqrt{2}r_7$.
		Therefore, if the routine finds a sequence $q_1,\ldots,q_m$ satisfying the preconditions of Lemma~\ref{lem:rectangleSB-wall-building}, it succeeds and we are done.

		By Condition~(2), $r_7$ and all smaller disks satisfy Condition~(1) of Lemma~\ref{lem:rectangleSB-wall-building}.
		Therefore, if the routine cannot find a sequence $q_1,\ldots,q_m$, that must be due to Conditions~(2)~and~(3) of Lemma~\ref{lem:rectangleSB-wall-building}.
		Consider any disk $r_i \leq r_7$ and the smallest disk $r_j \leq r_i$ such that $r_j \geq s \cdot r_i$.
		The total weight $\sum_{u=i}^{j}r_u^2$ of disks between $r_i$ and $r_j$ must be less than $\sqrt{2}r_i \cdot E \cdot L(1)$,
		as otherwise $r_i,\ldots,r_j$ would satisfy the conditions of Lemma~\ref{lem:rectangleSB-wall-building}.
		Applying this to $r_7$ implies that the total weight of all disks $r_i \geq s \cdot r_7 = t_0$ is at most $\sum_{u=1}^{6} r_u^2 + \sqrt{2}r_7 \cdot E \cdot L(1) \cdot s^0$;
		therefore, the weight of disks $r_i < t_0$ is at least $E\lambda - \sum_{u=1}^{6} r_u^2 - \sqrt{2}r_7 \cdot E \cdot L(1) \cdot s^0 = a_0$. 
		More generally, repeatedly applying this to $r_7$ yields that $a_k$ is a lower bound on the total weight of disks of radius $r_i < t_k$.
		
		Due to Condition~(3), there is a $k$ for which $a_k \geq \frac{t_k \cdot E}{\rectanglesSBRB}$.
		Thus, there is a set $D'_k = \{r_v,r_{v+1},\ldots,r_n\}, r_v < t_k$ such that $W(D'_k) \geq \frac{t_k E}{\rectanglesSBRB}$.
		Let $D_k$ be the smallest such set; by the bound $a_k$, we have $D_k \neq \emptyset$.
		Moreover, due to being smallest possible, $W(D_k) < \frac{t_k \cdot E}{\rectanglesSBRB} + t_k^2$.
		By $W(D_k) \geq a_k \geq \frac{t_k E}{\rectanglesSBRB}$, the disks from $D_k$ can cover a rectangle $\mathcal{R}_1$ of width $\frac{W(D_k)}{\rectanglesSBRB \cdot E} \geq \frac{t_k}{\rectanglesSBRB}$ by recursion, also satisfying the radius bound.
		Therefore, if $D_k \neq D$ and $r_1$ satisfies the radius bound of Lemma~\ref{lem:rectanglesSB} w.r.t. the remaining rectangle $\mathcal{R}_2$, we are done.
		If $D_k = D$, then $r_1 > 0$ cannot satisfy the radius bound w.r.t.\ the empty remaining rectangle $\mathcal{R}_2$.
		Furthermore, due to $W(D_k) < \frac{t_k \cdot E}{\rectanglesSBRB} + t_k^2$, the width of $\mathcal{R}_2$ is at least $\lambda - \frac{t_k}{\rectanglesSBRB} - \frac{t_k^2}{E}$;
		the second part of Condition~(3) guarantees that $r_1$ satisfies the radius bound w.r.t. $\mathcal{R}_2$.
	\end{proof}
	All quantities occurring in this lemma are known except for $k$; in our implementation, we check the preconditions for $k=0,\ldots,32$ and ignore the routine if none of these values work.
	We can also adapt this lemma and the routine to use fewer than the 7 largest disks using analogous arguments.
}

\subsubsection{Placing \texorpdfstring{$r_1$}{r1} in a Corner}\sbnewsubsec{strat:rectanglesSB-r1c}
\psbstratref{strat:rectanglesSB-r1c}{1} The routine described in this section leverages Lemma~\ref{lem:rectangleSB-wall-building} to handle the case of a single large disk $r_1$.
The idea is to place $r_1$ in the lower left-hand corner, filling up the space above $r_1$ using Lemma~\ref{lem:rectangleSB-wall-building} and using recursion to handle the remaining region to the right of $r_1$; see Fig.~\ref{fig:rectanglesSB-r1-wall-building}.

To be more precise, we begin by placing $r_1$ in the lower left-hand corner of $\mathcal{R}$ such that $r_1$ covers a square of side length $\ell \coloneqq \sqrt{2}r_1$.
We subdivide the remaining region into a rectangle $\mathcal{A}$ \revsocg{of dimensions $\ell \times \left(1-\ell\right)$} above $r_1$ and a rectangle $\mathcal{B}$ \revsocg{of dimensions $\left(\lambda-\ell\right) \times 1$} to the right of $r_1$; see Fig.~\ref{fig:rectanglesSB-r1-wall-building}.%
\revsocg{
	We then use the remaining disks to cover $\mathcal{A}$ and $\mathcal{B}$.
	For $\mathcal{A}$, we consider the following options.
	\begin{enumerate}
		\item[$\mathcal{A}.1$] Starting with the smallest disk $r_n$, we create a collection $D_{\mathcal{A}}$ of disks that we use to apply Lemma~\ref{lem:rectanglesSB} to $\mathcal{A}$ recursively.
		\item[$\mathcal{A}.2$] We cover $\mathcal{A}$ by horizontal rows built using Lemma~\ref{lem:rectangleSB-wall-building}; see Fig.~\ref{fig:rectanglesSB-r1-wall-building}.
			Beginning with $r_7$ and continuing in decreasing order of radius, we build horizontal rows of fixed width $\ell = \sqrt{2}r_1$ and variable height $b$ using Lemma~\ref{lem:rectangleSB-wall-building}.
			Each row is built by adding disks until either (a) Condition~(2)~of~Lemma~\ref{lem:rectangleSB-wall-building} is violated or (b) Condition~(3) is met.
			In case (a), the disks added to this \emph{incomplete} row so far are excluded for covering $\mathcal{A}$ and added to the collection of disks used to cover $\mathcal{B}$.
			In case (b), we complete the row according to Lemma~\ref{lem:rectangleSB-wall-building} and place it on top of $r_1$ and the previously built rows.
	\end{enumerate}
	In either case, the remaining disks are used to cover $\mathcal{B}$, for which we consider the following options.
	\begin{enumerate}
		\item[$\mathcal{B}.1$] We recurse on $\mathcal{B}$ using Lemma~\ref{lem:rectangleSB-wall-building} and all remaining disks.
		\item[$\mathcal{B}.2$] We place $r_2$ and $r_3$ covering a strip of width $\lambda-\ell$ and maximum possible height at the 
			bottom of $\mathcal{B}$ and recurse on the remaining rectangle using Lemma~\ref{lem:rectangleSB-wall-building} and the remaining disks.
		\item[$\mathcal{B}.3$] Analogous to $\mathcal{B}.2$, but using $r_2$ -- $r_4$ for the strip at the bottom of $\mathcal{B}$ instead of just $r_2$ and $r_3$.
	\end{enumerate}
	
	Like Routine~\sbstratref{strat:rectanglesSB-wb}{1}, this routine heavily relies on the small disks $r_8,\ldots,r_n$ and it is not straightforward to give a success criterion based on $\lambda$ and $r_1,\ldots,r_7$.
	We use the criterion given by the following lemma; its preconditions can be checked knowing only $\lambda$ and $r_1,\ldots,r_7$.
	Note that, in particular, the dimensions of $\mathcal{A}$ and $\mathcal{B}$ can be computed based on $\lambda$ and $r_1$.
	\begin{lemma}
		Let $s \coloneqq \sqrt{4E^2-2E\sqrt{4E^2-1}}$ as in Lemma~\ref{lem:wall-building-suff-crit} and $L(\ell)$ as in Lemma~\ref{lem:rectangleSB-wall-building}.
		Let $k$ be the smallest non-negative integer such that $r_7 \cdot s^k \leq \rho(\mathcal{A})$.
		Let $w_g \coloneqq (2E-1)r_1^2$,
		\[w_k \coloneqq E\cdot (\lambda-\abs*{\mathcal{A}}) - \sum\limits_{i=1}^{6}r_i^2 - L\left(\sqrt{2}r_1\right) \cdot \sqrt{2}r_7 \cdot \sum\limits_{i=0}^k s^i\text{, and}\]
		\[w_{\infty} \coloneqq E\lambda - \sum\limits_{i=1}^{7}r_i^2 - \frac{E \cdot L\left(\sqrt{2}r_1\right) \cdot \sqrt{2}r_7}{1-s}.\]
		Let $h_{23} \geq 0$ be the height of the tallest rectangular strip $\mathcal{S}_{23}$ at the bottom of $\mathcal{B}$ that can be fully covered by $r_2,r_3$ (case $\mathcal{B}.2$),
		and, analogously, let $h_{234}\geq 0$ be the height of the tallest rectangular strip $\mathcal{S}_{234}$ coverable by $r_2,r_3,r_4$ (case $\mathcal{B}.3$).
		Let
		\[ w_2 \coloneqq \begin{cases}0\text{,} & \text{if } r_2\text{ fits }\mathcal{B}\text{,}\\
		-\infty\text{,} & \text{otherwise,}\end{cases}\ \,
			w_{23} \coloneqq \begin{cases}E\cdot\left(\lambda-\sqrt{2}r_1\right)\cdot h_{23}-r_2^2-r_3^2\text{,} & \text{if }r_4\text{ fits }\mathcal{B}\setminus \mathcal{S}_{23}\text{,}\\
			-\infty\text{,} & \text{otherwise, and}\end{cases}
		\]
		\[
			w_{234} \coloneqq \begin{cases}E\cdot\left(\lambda-\sqrt{2}r_1\right)\cdot h_{234}-r_2^2-r_3^2-r_4^2\text{,} & \text{if }r_5\text{ fits }\mathcal{B}\setminus \mathcal{S}_{234}\text{,}\\
			-\infty\text{,} &  \text{otherwise.}\end{cases}
		\]

		\noindent Routine~\sbstratref{strat:rectanglesSB-r1c}{1} is guaranteed to succeed if the following conditions hold.
		\[(1)\ w_g + \max \left(w_2,w_{23},w_{234} \right) \geq \max\left\{2E\cdot r_1 \cdot r_7,\, \min\left\{r_7^2, \rho(\mathcal{A})^2\right\}\right\}\text{,} \]
		\[(2)\ r_7 \leq r_1 \cdot \left(1 - \frac{1}{\sqrt{1+\sqrt{1-\frac{1}{4E^2}}}}\right) \approx 0.20263 \cdot r_1\text{,} \]
		\[(3)\ \max\left\{w_k,w_{\infty}\right\} \geq E\abs*{\mathcal{A}}\text{, and}\]
		\[(4)\ E\lambda - \sum\limits_{i=1}^{7}r_i^2 \geq E\abs*{\mathcal{A}}\text{.}\]
	\end{lemma}
	\begin{proof}
		Firstly, by Condition~(1), for at least one of the three options $\mathcal{B}.1,\mathcal{B}.2,\mathcal{B}.3$, the largest disk that we use for $\mathcal{B}$ must fit, as otherwise, $\max \{w_2,w_{23},w_{234}\} = -\infty$.
		Moreover, $w_g = (2E-1)r_1^2 = E\lambda - r_1^2 - E\abs*{\mathcal{A}} - E\abs*{\mathcal{B}}$ is the spare weight that we have gained,
		compared to a covering with coefficient $E$, by placing $r_1$ covering a square of side length $\sqrt{2}r_1$, which yields a coefficient of $\frac{1}{2} < E$.
		Similarly, $w_{23}$ is the spare weight that we gain (or lose, if $w_{23} < 0$), by placing $r_2$ and $r_3$ covering the strip $\mathcal{S}_{23}$ at the bottom of $\mathcal{B}$,
		and analogously for $w_2$ and $w_{234}$.
		Intuitively speaking, we use this accumulated spare weight to pay for the waste that we may incur while covering $\mathcal{A}$ and $\mathcal{B}$.
		We distinguish two cases based on the total weight $W_{a} \coloneqq \sum_{r_i \leq \rho(\abs*{\mathcal{A}})}r_i^2$ of disks below $\rho(\abs*{\mathcal{A}})$.

		First, assume $W_{a} \geq E\abs*{\mathcal{A}}$.
		Then, option $\mathcal{A}.1$ can be used for covering $\mathcal{A}$.
		Because we stop adding disks to $D_{\mathcal{A}}$ once the total weight exceeds $E\abs*{\mathcal{A}}$, due to the radius bound and Condition~(4),
		the last disk added to $D_{\mathcal{A}}$ can weigh no more than $\min\left(r_7^2,\rho(\mathcal{A})^2\right)$ and thus $W(D_{\mathcal{A}}) \leq E\abs*{\mathcal{A}} + \min\left(r_7^2,\rho(\mathcal{A})^2\right)$.
		Therefore, due to Condition~(1), we have enough spare weight to handle $\mathcal{B}$ using one of the three options $\mathcal{B}.1\text{--}\mathcal{B}.3$.

		Now, we assume $W_{a} < E\abs*{\mathcal{A}}$; in this case, we use option $\mathcal{A}.2$ to cover $\mathcal{A}$.
		In order to guarantee success in this case, we have to show that
		(a)~building new rows cannot fail due to Condition~(1) of Lemma~\ref{lem:rectangleSB-wall-building},
		(b)~the disk weight used to cover $\mathcal{A}$ is at most $E\abs*{\mathcal{A}} + w_g$, and
		(c)~we do not run out of disks while covering $\mathcal{A}$ due to too much disk weight in \emph{incomplete} rows.

		Condition~(2) guarantees that (a) holds, because if we can guarantee that $r_7$ does not violate Condition~(1) of Lemma~\ref{lem:rectangleSB-wall-building},
		the smaller disks $r_8,\ldots,r_n$ cannot violate the condition either.

		Regarding~(b), we waste at most one complete strip of length $\sqrt{2}r_1$ and height $b \leq \sqrt{2}r_7$ that is covered with coefficient $E$; see Fig.~\ref{fig:rectangleSB-wall-building}.
		Therefore, we waste at most weight $E \cdot \sqrt{2}r_1 \cdot \sqrt{2}r_7 = 2E\cdot r_1 \cdot r_7$, which is at most $w_g$ due to Condition~(1).
		
		Regarding~(c), we consider the sequence $I_1,\ldots,I_m$ of incomplete rows encountered by $\mathcal{A}.2$.
		Let $r_i$ be the largest disk of $I_u$ and let $r_j < r_i$ be the largest disk of $I_{u+1}$ for some $1 \leq u < m$.
		Because $I_u$ is incomplete, somewhere between $r_i$ and $r_j$, Condition~(2) of Lemma~\ref{lem:rectangleSB-wall-building} must have been violated.
		Therefore, we have $r_j < s \cdot r_i$.
		Moreover, an incomplete row in which $r_i$ is the largest disk can have at most weight $E \cdot \sqrt{2}r_i \cdot L\left(\sqrt{2}r_1\right)$.
		In the following, we give two upper bounds on the total weight $W_i$ of disks that may end up in incomplete rows.

		Recall that $k$ is the smallest non-negative integer for which $r_7 \cdot s^k \leq \rho(\mathcal{A})$, and that we have less than $E\abs*{\mathcal{A}}$ weight in disks below $\rho(\mathcal{A})$.	
		Therefore, by assuming that all weight in disks below $\rho(\mathcal{A})$ is in incomplete rows, we can bound the weight in incomplete rows by
		\begin{align*}
			B_k &\coloneqq \underbrace{E\abs*{\mathcal{A}}}_{\leq\,\rho(\mathcal{A})} + \underbrace{E \cdot \sqrt{2}r_7 \cdot L\left(\sqrt{2}r_1\right) \cdot s^0}_{\text{\nth{1} incomplete row}} + 
			\cdots + \underbrace{E \cdot \sqrt{2}r_7 \cdot L\left(\sqrt{2}r_1\right) \cdot s^k}_{\text{$(k+1)$st incomplete row}}\\
			&= E\abs*{\mathcal{A}} + E \cdot \sqrt{2}r_7 \cdot L\left(\sqrt{2}r_1\right) \cdot \sum\limits_{i=0}^{k} s^i.
		\end{align*}
		Moreover, instead of subsuming all incomplete rows below $\rho(\mathcal{A})$, we can also bound the weight in incomplete rows by
		\begin{align*}
			B_{\infty} &\coloneqq \underbrace{E \cdot \sqrt{2}r_7 \cdot L\left(\sqrt{2}r_1\right) \cdot s^0}_{\text{\nth{1} incomplete row}} + \underbrace{E \cdot \sqrt{2}r_7 \cdot L\left(\sqrt{2}r_1\right) \cdot s^1}_{\text{\nth{2} incomplete row}} + \cdots\\
			&= E \cdot \sqrt{2}r_7 \cdot L\left(\sqrt{2}r_1\right) \cdot \sum\limits_{i=0}^{\infty} s^i = \frac{E \cdot \sqrt{2}r_7 \cdot L\left(\sqrt{2}r_1\right)}{1-s}\text{.}
		\end{align*}
		Therefore, at least weight $E\lambda - \sum\limits_{i=1}^{7}r_i^2 - \min\{B_k,B_{\infty}\} = \max\{w_k,w_{\infty}\}$ is available for covering $\mathcal{A}$ using options $\mathcal{A}.2$.
		By Condition~(3) and Lemma~\ref{lem:rectangleSB-wall-building}, this suffices to cover $\mathcal{A}$.
	\end{proof}
	As for Routine~\sbstratref{strat:rectanglesSB-wb}{1}, we can give success criteria using fewer than $7$ largest disks using analogous arguments.
}

\subsubsection{Placing \texorpdfstring{$r_1$}{r1} and \texorpdfstring{$r_2$}{r2} in Opposite Corners}\sbnewsubsec{strat:rectanglesSB-r12oc}
\psbstratref{strat:rectanglesSB-r12oc}{1} The next routine is based on placing $r_1$ and $r_2$ covering squares in diagonally opposite corners of $\mathcal{R}$.
If the total height $\sqrt{2}\left(r_1+r_2\right)$ covered by $r_1$ and $r_2$ exceeds $1$, we do not consider this routine;
therefore, the situation is as depicted in Fig.~\ref{fig:rectanglesSB-r1_r2_opposite_corners}.
After placing $r_1$ and $r_2$, for each disk \revsocg{$r_3,\ldots,r_7$}, we check whether we can place it on the remainder $\mathcal{C}'$ of $\mathcal{C}$
such that it cuts off a part of the longer side of $\mathcal{C}'$; see Fig.~\ref{fig:rectanglesSB-r1_r2_opposite_corners}~(right).
We continue this until $\mathcal{C}'$ disappears completely or the rectangle $\mathcal{R}_i$ covered by disk $r_i$ would not satisfy $E\abs*{\mathcal{R}_i} \geq r_i^2$\revsocg{, i.e., until the placement would become too inefficient.}
\revsocg{W.l.o.g., let the short side of region $\mathcal{B}$ be no longer than the short side} of region $\mathcal{A}$; the other case is handled analogously.

\revsocg{It is straightforward to decide, based on $\lambda,r_1,\ldots,r_7$, whether $\mathcal{C}'$ disappears.
Moreover, we can decide which disk from $r_3,\ldots,r_7$ is the largest disk not placed on $\mathcal{C}'$.
If $\mathcal{C}'$ disappears, we proceed as follows.
Our success criterion checks that this disk fits into $\mathcal{A}$ and that $r_7$ fits into $\mathcal{B}$.
We build a collection $D_{\mathcal{A}}$ of disks for recursively covering $\mathcal{A}$, beginning with the largest remaining disk, until $W(D_{\mathcal{A}}) \geq E\abs*{\mathcal{A}}.$
Because the largest remaining disk fits $\mathcal{A}$, by Observation~\ref{obs:at-least-5}, this set contains at least $5$ disks;
thus, we can bound $W\left(D_{\mathcal{A}}\right) \leq E\abs*{\mathcal{A}} + r_7^2$.
Our success criterion then checks whether the remaining weight suffices to recurse on $\mathcal{B}$.}

If $\mathcal{C}'$ does not disappear, we proceed as follows.
We tentatively build $D_{\mathcal{A}}$ by adding disks in decreasing order of radius until \revsocg{$E\abs*{\mathcal{A}} + r_7^2 \geq W\left(D_{\mathcal{A}}\right)\geq E\abs*{\mathcal{A}}$};
once $D_{\mathcal{A}}$ is complete, we continue building $D_{\mathcal{B}}$ in the same manner.
We place the remaining disks in $D_{\mathcal{C}'}$.
\revsocg{Our success criterion then checks, using $r_7$ as bound on the largest radii in $D_{\mathcal{B}},D_{\mathcal{C}'}$, whether we can recurse on $\mathcal{A},\mathcal{B},\mathcal{C}'$ using $D_{\mathcal{A}}, D_{\mathcal{B}}, D_{\mathcal{C}'}$.}
Otherwise, we discard $D_{\mathcal{A}}, D_{\mathcal{B}}$ and $D_{\mathcal{C}'}$ and continue as follows.
Instead of using Lemma~\ref{lem:rectanglesSB}, we try to use Theorem~\ref{thm:mainRectangles} to recurse on $\mathcal{C}'$.
Towards this goal, we compute the skew $\lambda_{\mathcal{C}'}$ of $\mathcal{C}'$ and build $D_{\mathcal{C}'}$ by adding remaining disks in decreasing order of radius until $W\left(D_{\mathcal{C}'}\right) \geq E^*\left(\lambda_{\mathcal{C}'}\right) \cdot \abs*{\mathcal{C}'}$.
Afterwards, we build $D_{\mathcal{A}}$ of weight at least $E\abs*{\mathcal{A}}$, placing all remaining disks in $D_{\mathcal{B}}$.
\revsocg{By keeping track of the disks definitely placed in $D_{\mathcal{C}'}$ in this manner, we can upper-bound the size of the largest disk and the weight for each of $D_{\mathcal{A}}, D_{\mathcal{B}}, D_{\mathcal{C}'}$.
Our success criterion checks whether these bounds guarantee that we can recurse on $\mathcal{A}$ and $\mathcal{B}$ using $D_{\mathcal{A}}$ and $D_{\mathcal{B}}$.}

\subsubsection{Using the Three Largest Disks}\sbnewsubsec{strat:rectanglesSB-3ld}
In this section, we describe two routines that are based on considering the largest three disks; see Fig.~\ref{fig:rectanglesSB-l-shaped-recursion}.
In either routine, we begin by covering a vertical rectangular strip $\mathcal{S}_2$ of height $1$ and maximal width at the left of $\mathcal{R}$.
\revsocg{Our success criterion checks whether $E\abs*{\mathcal{S}_2} \geq r_1^2+r_2^2$, i.e., whether covering this strip is efficient enough.}
Afterwards, we consider two different placements of $r_3$ to cover a part of the remaining region.

\psbstratref{strat:rectanglesSB-3ld}{1} The first option is to place $r_3$ covering its inscribed square at the lower left corner of the remaining region.
As depicted in Fig.~\ref{fig:rectanglesSB-l-shaped-recursion}, we can subdivide the remaining area into two rectangles $\mathcal{A}$ and $\mathcal{B}$
either horizontally or vertically; we try both options and handle them analogously.
W.l.o.g., let the short side of $\mathcal{A}$ be at least as long as the short side of $\mathcal{B}$; the other case is symmetric.

\revsocg{Our success criterion checks whether $r_4$ fits $\mathcal{A}$.}
\revsocg{In that case,} we build $D_{\mathcal{A}}$ by adding disks in decreasing order of radius until $W\left(D_{\mathcal{A}}\right) \geq E\abs*{\mathcal{A}}$;
we place all other disks in $D_{\mathcal{B}}$.
\revsocg{Because $r_4$ fits $\mathcal{A}$, by Observation~\ref{obs:at-least-5}, at least $5$ disks are added to $D_{\mathcal{A}}$.
Therefore, we can bound the weight $W(D_{\mathcal{A}}) \leq E\abs*{\mathcal{A}} + r_7^2$ and the largest disk in $\mathcal{B}$ is at most $r_7$.
Our success criterion checks whether these bounds guarantee that we can recurse on $\mathcal{A},\mathcal{B}$ in this manner using Lemma~\ref{lem:rectanglesSB}.}

\revsocg{Otherwise, or if $r_4$ does not fit $\mathcal{A}$, we compute the amount of weight $w_{\mathcal{A}},w_{\mathcal{B}}$
necessary to recurse on $\mathcal{A},\mathcal{B}$ using Theorem~\ref{thm:mainRectangles}.}
We build a collection $D_1$ by adding disks in order of decreasing radius until $W(D_1) \geq \min\{w_{\mathcal{A}},w_{\mathcal{B}}\}$.
\revsocg{By keeping track of the largest disk that is definitely placed in $D_1$, using $r_7$ if we run out of known large disks,
we can bound the size of the largest remaining disk and the amount of waste by which $W(D_1)$ exceeds $\min\{w_{\mathcal{A}},w_{\mathcal{B}}\}$.
Our success criterion checks that the largest remaining disk definitely fits the remaining region and that there definitely is enough remaining weight to recurse on that region.}

\psbstratref{strat:rectanglesSB-3ld}{2}
We also consider placing $r_3$ such that it covers a horizontal strip of maximum height at the bottom of the remaining area $\mathcal{A}$;
see Fig.~\ref{fig:rectanglesSB-l-shaped-recursion}~(b).
\revsocg{Our success criterion for this routine checks that such a placement is feasible.}
In that case, we \revsocg{check whether we can recurse on the rectangle using the remaining disks.}
If that does not work, we consider placing disks $r_4,r_5$ and $r_6$ as follows\revsocg{, checking whether we can apply recursion after each placement.}
Each disk is placed such that it cuts off a rectangular piece of $\mathcal{A}$, reducing the length of the longer side of $\mathcal{A}$ as much as possible.
\revsocg{Our success criterion excludes this routine if we cannot place a disk in this way.}

\subsubsection{Using the Four Largest Disks}\sbnewsubsec{strat:rectanglesSB-4ld}
In this section, we describe several routines that are based on computing a placement for the four largest disks.
For an overview, see Fig.~\ref{fig:rectanglesSB-r1_in_corner_right_34_and_2}.

\psbstratref{strat:rectanglesSB-4ld}{1} The first routine, depicted in Fig.~\ref{fig:rectanglesSB-r1_in_corner_right_34_and_2}~(a), places $r_1$ covering its inscribed square in the bottom-left corner of $\mathcal{R}$.
The remaining area is subdivided into two regions $\mathcal{A}$ above $r_1$ and $\mathcal{B}$ right of $r_1$.
W.l.o.g., let the shorter side of $\mathcal{B}$ be at least as long as the shorter side of $\mathcal{A}$.
We split the remaining disks into two groups $D_{\mathcal{A}},D_{\mathcal{B}}$ by adding disks to $D_{\mathcal{B}}$ in decreasing order of radius until $W\left(D_{\mathcal{B}}\right) \geq E\abs*{\mathcal{B}}$ and putting the remaining disks into $D_{\mathcal{A}}$.

If we cannot recurse on the remaining rectangles immediately due to the radius bound or $W\left(D_{\mathcal{A}}\right) < E\abs{\mathcal{A}}$, we continue as follows.
We place $r_3,r_4$ covering a horizontal strip of maximal height at the bottom of $\mathcal{B}$.
We place $r_2$ covering another rectangular strip of $\mathcal{B}$ on top of that.
If either of these placements is impossible because the disks are too small, the routine fails.
Otherwise, we retry building $D_{\mathcal{A}},D_{\mathcal{B}}$ and recursing.

\revsocg{It is straightforward to check the feasibility of the explicit disk placements in our success criteria.
We keep track of the largest disks that we place in $D_{\mathcal{A}}$ and $D_{\mathcal{B}}$.
Moreover, we also keep track of the smallest disk $r_{\mathcal{A}}$ among $r_1,\ldots,r_7$ that we place in $D_{\mathcal{A}}$ or $D_{\mathcal{B}}$ to bound $W(D_{\mathcal{A}}) \leq E\abs*{\mathcal{A}} + r_{\mathcal{A}}^2$ (and analogously for $W(B)$).
Our success criterion then uses these bounds to check if we can guarantee the success of the recursion using Lemma~\ref{lem:rectanglesSB} or Theorem~\ref{thm:mainRectangles}.}

\psbstratref{strat:rectanglesSB-4ld}{2} The next routine, depicted in Fig.~\ref{fig:rectanglesSB-r1_in_corner_right_34_and_2}~(b)~and~(c), consists of first computing a covering $\mathcal{C}_{2\times 2}$ of a width-maximal vertical strip of $\mathcal{R}$ using two columns, each containing two disks from $r_1,\ldots,r_4$.
Among all partitions of $r_1,\ldots,r_4$ into two groups of size two, we pick the one for which the covered width of $\mathcal{R}$ is maximized.
\revsocg{Our success criterion discards this routine if we cannot cover at least width $\frac{r_1^2+\cdots+r_4^2}{E}$ in this way.}

Otherwise, we consider several ways to cover the remaining rectangular strip $\mathcal{A}$.
The first way consists of simply recursing on $\mathcal{A}$ and the remaining disks.
\revsocg{If we cannot guarantee this to work by Theorem~\ref{thm:mainRectangles} or Lemma~\ref{lem:rectanglesSB}~or~\ref{lem:size-bound-large}}, we consider placing $r_5$ covering a rectangular strip at the bottom of $\mathcal{A}$.
If this is impossible, we discard the routine; otherwise, we remove the covered rectangle from $\mathcal{A}$ and reconsider recursing.
If we again cannot guarantee success, we consider placing $r_6$ covering a horizontal strip on top of $r_5$ at the bottom of $\mathcal{A}$ and reconsider recursion again.
Moreover, we also consider the placement of disks depicted in Fig.~\ref{fig:rectanglesSB-r1_in_corner_right_34_and_2}~(c).

Instead of covering a horizontal strip at the bottom of $\mathcal{A}$, we also consider using $r_6$ to cover a vertical strip at the right of $\mathcal{A}$.
If $r_6$ cannot cover the entire remainder of the right side of $\mathcal{R}$, we ignore the routine.
If $r_6$ can cover $\mathcal{A}$ completely, the routine succeeds.
Otherwise, we also disregard the routine if the left intersection point of $r_5$ and $r_6$ does not lie within the upper-right disk of $\mathcal{C}_{2\times 2}$.

The only part of $\mathcal{R}$ that remains to be covered is a region above the intersection point of the upper-right disk of $\mathcal{C}_{2\times 2}$ and $r_6$.
Our success criterion checks whether we can \revsocg{use Lemma~\ref{lem:rectanglesSB}~or~\ref{lem:size-bound-large} or Theorem~\ref{thm:mainRectangles} to guarantee successful recursion on} the bounding box of that region using the remaining disks.

\psbstratref{strat:rectanglesSB-4ld}{3} Finally, we consider the routine depicted in Fig.~\ref{fig:rectanglesSB-r1_in_corner_right_34_and_2}~(d), where we place $r_1,r_2$ such that they cover a rectangular strip of height $1$ and maximum width $\omega_{12}$ at the left side of $\mathcal{R}$ and $r_3,r_4$ such that they cover a rectangular strip of width $\lambda - \omega_{12}$ and maximum height.
The routine succeeds if we can guarantee successful recursion on the remaining region $\mathcal{A}$.

\subsubsection{Covering with Five Disks}\sbnewsubsec{strat:rectanglesSB-5ld}
In this section, we describe routines for covering $\mathcal{R}$ that rely on using the five largest disks and recursion on the remaining region $\mathcal{A}$.

\psbstratref{strat:rectanglesSB-5ld}{1} The first routine, depicted in Fig.~\ref{fig:rectanglesSB-five-disks}~(a)~and~(b), uses either $r_1,r_2,r_3$ or $r_1,r_2,r_5$ to cover a horizontal strip of width $\lambda$ and maximal height at the bottom of $\mathcal{R}$.
Afterwards, we use the two remaining disks to cover a vertical strip of maximum width at the left side of $\mathcal{R}$, and recurse on the remaining area $\mathcal{A}$. 

\psbstratref{strat:rectanglesSB-5ld}{2} The second routine, depicted in Fig.~\ref{fig:rectanglesSB-five-disks}~(c), begins by placing $r_1$ and $r_2$ covering a horizontal strip of width $\lambda$ and maximal height at the bottom of $\mathcal{R}$.
We use $r_4$ and $r_5$ to cover the remainder of the right and left sides of $\mathcal{R}$ and place $r_3$ such that it covers the remainder of the top side of $\mathcal{R}$;
if either of these placements are impossible, the routine is discarded.
If the five largest disks cover the entire rectangle $\mathcal{R}$ when placed in this manner, we are done.
Otherwise, we compute the bounding box $\mathcal{A}$ of the region that remains to be covered.
\revsocg{Our success criterion checks whether we can guarantee successful recursion on $\mathcal{A}$.}

\subsubsection{Covering with Six Disks}\sbnewsubsec{strat:rectanglesSB-6ld}
In this section, we describe three routines based on covering $\mathcal{R}$ using the six largest disks.

\psbstratref{strat:rectanglesSB-6ld}{1} The first routine uses only the six largest disks as depicted in Fig.~\ref{fig:rectanglesSB-six-disks}~(a); after covering a strip of width $\lambda$ and maximal height at the bottom of $\mathcal{R}$, we place the disks $r_4$ and $r_3$ covering the remainder of $\mathcal{R}$'s left and right boundary.
\revsocg{Our success criterion checks whether} $r_3$ and $r_4$ intersect; in that case, two uncovered pockets remain.
We check whether we can cover the smaller pocket using $r_6$ and the larger one using $r_5$.

\psbstratref{strat:rectanglesSB-6ld}{2} The second routine begins by covering a strip of width $\lambda$ and maximal height at the bottom of $\mathcal{R}$ using disks $r_1,r_2$ and recursion on a rectangular region $\mathcal{A}$.
The maximal height that can be covered in this way can be obtained by solving two systems of quadratic equations, one for each case $\lambda_{\mathcal{A}} < \lambda_2, \lambda_{\mathcal{A}} \geq \lambda_2$, where $\lambda_{\mathcal{A}}$ is the skew of $\mathcal{A}$.
Again, it is straightforward to check for any given height $h$ whether it is achievable; therefore, in our automatic prover, we simply use the bisection method to find a lower bound for the height that is definitely achievable and an upper bound on the height that may possibly be achieved.

After placing $r_1$ and $r_2$, we again try to place $r_4$ and $r_3$ covering the remaining part of $\mathcal{R}$'s left and right border.
Afterwards, we consider placing $r_5$ covering the remaining part of $\mathcal{R}$'s top border and check whether $r_6$ can be used to cover the remaining region.

\psbstratref{strat:rectanglesSB-6ld}{3} Finally, we also use the routine depicted in Fig.~\ref{fig:rectanglesSB-six-disks}~(c), where we try to cover two vertical strips of height $1$ and maximal width using disks $r_1,r_2$ and $r_4,r_5,r_6$.
We try to place $r_3$ covering a rectangle of maximal height of the remaining strip and \revsocg{check whether we can guarantee successful recursion} on the remaining rectangular region $\mathcal{A}$.

\subsubsection{Covering with Seven Disks}\sbnewsubsec{strat:rectanglesSB-7ld}
In this section, we describe several routines for covering $\mathcal{R}$ that are based on using the seven largest disks.

\psbstratref{strat:rectanglesSB-7ld}{1} We begin by considering to cover a strip of height $1$ and maximum possible width using the first six disks as depicted in Fig.~\ref{fig:rectanglesSB-six-disks-max-width}.
If this leads to a full cover of $\mathcal{R}$ or if we can guarantee successful recursion on the remaining region $\mathcal{A}$, we are done.
Otherwise, we consider placing $r_7$ covering a horizontal rectangular strip at the bottom of $\mathcal{A}$ as depicted in Fig.~\ref{fig:rectanglesSB-six-disks-max-width} and \revsocg{check whether we can guarantee successful recursion} on the remaining rectangle.

\psbstratref{strat:rectanglesSB-7ld}{2} Next, we describe the routine depicted in Fig.~\ref{fig:rectanglesSB-six_disk_222_with_recursion}.
\revsocg{It works by covering two vertical strips of maximum width using the partition of $r_1,\ldots,r_4$ into two groups of two disks that maximizes the covered width.}
We place $r_5$ and $r_6$ on the remaining strip, covering rectangles of maximum possible height.
If this covers the entire remaining area, we are done; \revsocg{otherwise, our success criterion discards the routine if} $r_6$ does not intersect the top side of $\mathcal{R}$.
Two pockets remain uncovered; we consider their bounding boxes $\mathcal{A}$ and $\mathcal{B}$.
If $r_7$ suffices to cover one of these pockets, we place it covering $\mathcal{B}$ and \revsocg{check whether we can guarantee successful recursion} on $\mathcal{A}$.
Otherwise, we apply \textsc{Greedy Splitting} to the remaining disks\revsocg{; this partitions the remaining disks into two collections $D_{\mathcal{A}}$ and $D_{\mathcal{B}}$ with $\abs*{W(D_{\mathcal{A}}) - W(D_{\mathcal{B}})} \leq r_7^2$.
Using this to bound the cost of the split, we check whether we can guarantee successful recursion on $\mathcal{A}$ and $\mathcal{B}$.}

\psbstratref{strat:rectanglesSB-7ld}{3} We continue describing the routines depicted in Fig.~\ref{fig:rectanglesSB-seven-disk-strategies}.
In the first routine, depicted in Fig.~\ref{fig:rectanglesSB-seven-disk-strategies}~(a)--(a''), we cover a strip of maximum width using disks $r_3,r_4,r_5$ and try placing $r_1,r_2$ covering the remainder of the top and bottom border.
If the disks are large enough (case~(a)), this covers $\mathcal{R}$ except for a remaining region $\mathcal{A}$, for which we check whether we can guarantee that recursion succeeds.
Otherwise, in case~(a'), we cover as much as possible of the top and bottom border using $r_1,r_2$ without moving the left intersection point of $r_1,r_2$ out of $r_3$.
We place $r_6$ and $r_7$ on the remaining strip; if the right intersection of $r_1,r_2$ is in $r_6$, we check whether recursion on the remaining region $\mathcal{A}$ is guaranteed to be successful.
Otherwise, in case~(a''), we consider using $r_6,r_7$ to cover the wider strip defined by the right intersection of $r_1,r_2$, and check whether we can recurse on the bounding box of the remaining region $\mathcal{A}$.

\psbstratref{strat:rectanglesSB-7ld}{4} Next, we consider the routine depicted in Fig.~\ref{fig:rectanglesSB-seven-disk-strategies}~(b).
For each possible choice $t_1,t_2,t_3$ of three disks from $r_1,\ldots,r_7$, we consider covering a strip of width $\lambda$ and maximum possible height $h$ at the top of $\mathcal{R}$.
We then compute the width of the widest possible rectangle $\mathcal{A}$ of height $1-h$ for which we can guarantee successful recursion using disks $r_8,\ldots,r_n$, placing it at the right border of the remaining area.
We place two disks $b_1,b_2$ covering a horizontal strip of maximum height at the bottom of the remaining area and check whether the last two disks $\ell_1,\ell_2$ can cover the entire remaining region.

\psbstratref{strat:rectanglesSB-7ld}{5} In the routine depicted in Fig.~\ref{fig:rectanglesSB-seven-disk-strategies}~(c), we place $r_2$ in the bottom left corner of $\mathcal{R}$, covering its inscribed square.
We place $r_1$ covering the same width on top of $r_2$; if this would exceed a height of $1$, we instead cover a vertical strip of maximum possible width with $r_1,r_2$.
We then cover two horizontal strips of width $\lambda - \sqrt{2}r_2$ and maximum possible height using disks $r_3,r_4$ and $r_5,r_6$, discarding the routine if such a placement is infeasible.
The remaining region can be subdivided into two rectangles $\mathcal{A}$ above $r_3,\ldots,r_6$ and $\mathcal{B}$ above $r_1,r_2$.
If $r_7$ can be placed such that it covers the left border of $\mathcal{B}$, consider the rectangle $\mathcal{R}_7 \subset \mathcal{B}$ covered by this placement.
If $\frac{r_7^2}{\abs*{\mathcal{R}_7}} \leq E$, we place $r_7$ in this way and reduce the size of $\mathcal{B}$ accordingly.
We compute the weight $w_{\mathcal{A}}$ necessary to recurse on $\mathcal{A}$; depending on the size of $r_7$ and $\mathcal{A}$, this may use Theorem~\ref{thm:mainRectangles} or Lemma~\ref{lem:rectanglesSB}~or~\ref{lem:size-bound-large}.
We build a collection $D_{\mathcal{A}}$ by adding the remaining disks in decreasing order of radius, until $w_{\mathcal{A}} + r_7^2 > W(D_{\mathcal{A}})\geq w_{\mathcal{A}}$.
our success criterion checks that the remaining disks have enough weight for this.
Using $D_{\mathcal{A}}$, we recurse on the widest possible rectangle $\mathcal{A}' \supseteq \mathcal{A}$.
Finally, our success criterion checks, using the bound $w_{\mathcal{A}} + r_7^2 > W(D_{\mathcal{A}})$ and $r_7$ as bound on the largest disk, whether we can guarantee successful recursion on the remainder $\mathcal{B'} \subseteq \mathcal{B}$ using the remaining disks.

\psbstratref{strat:rectanglesSB-7ld}{6} In the routine depicted in Fig.~\ref{fig:rectanglesSB-seven-disk-strategies}~(d), we start by covering a horizontal strip of width $\lambda$ and maximum height $h_{12}$ using disks $r_1,r_2$ and placing it at the bottom of $\mathcal{R}$.
We then compute the maximum height $h_{567}$ for which the following two conditions hold.
(1) We can place $r_6$ and $r_7$ at the left and right border of $\mathcal{R}$ such that they each cover a rectangular strip of height $h_{567}$.
(2) We can place $r_5$ between $r_6$ and $r_7$ such that, together with $r_1,r_2$, a horizontal strip of height $h_{12}+h_{567}$ is covered.
Finally we place $r_3$ and $r_4$ such that they cover the remainder of $\mathcal{R}$'s left and right boundary; if this covers everything, we are done.
If any of these placements are impossible or if the remaining uncovered region is not connected, we ignore this routine.
Otherwise, we check whether we can guarantee successful recursion on the bounding box $\mathcal{A}$ of the remaining region.

\psbstratref{strat:rectanglesSB-7ld}{7} In the routine depicted in Fig.~\ref{fig:rectanglesSB-seven-disk-strategies}~(e)--(e'), we begin by placing $r_2$ in the bottom-left corner of $\mathcal{R}$. covering its inscribed square.
We place $r_1$ right of $r_2$, covering a square of the same height $\sqrt{2}r_2$.
In the top-left corner of $\mathcal{R}$, we place $r_4$ and $r_3$ covering a strip of the same width as $r_1,r_2$ and maximum possible height.
A $T$-shaped region remains to be covered; parts of it are already covered by the first four disks.
For each disk $r_i$ among $r_5,r_6,r_7$, we proceed as follows.
First, we consider building a collection $D_{\mathcal{A}}$ from the remaining disks $r_i,\ldots,r_n$ that contains enough weight to guarantee successful recursion on the vertical strip $\mathcal{A}$.
If that works and there is enough remaining weight to successfully recurse on the remaining horizontal strip $\mathcal{B}$ (see Fig.~\ref{fig:rectanglesSB-seven-disk-strategies}~(e')), we are done.
Otherwise, we consider covering a piece of maximal width of the horizontal strip using $r_i$; if that is impossible, we disregard the routine.

If, during this operation, we place $r_7$ in such a way that it intersects the right boundary of $\mathcal{R}$ (see Fig.~\ref{fig:rectanglesSB-seven-disk-strategies}~(e)), the horizontal strip is completely covered and the vertical strip is subdivided into two pieces $\mathcal{A},\mathcal{B}$.
In this case, we apply \textsc{Greedy Splitting} to the remaining disks, resulting in two collections $D_{\mathcal{A}}, D_{\mathcal{B}}$ with $\abs*{W(D_{\mathcal{A}}) - W(D_{\mathcal{B}})} \leq r_7^2$.
We use this to bound the cost of the split and check whether we can guarantee successful recursion on $\mathcal{A}$ and $\mathcal{B}$ using $D_{\mathcal{A}}$ and $D_{\mathcal{B}}$.

\psbstratref{strat:rectanglesSB-7ld}{8} Finally, in the routine depicted in Fig.~\ref{fig:rectanglesSB-seven-disk-strategies}~(f), we begin by covering a strip of width $\lambda$ and maximum possible height at the top of $\mathcal{R}$ using disks $r_2,\ldots,r_5$.
Below that strip, at the left border of $\mathcal{R}$ we place $r_1$ covering its inscribed square.
If this placement covers the entire left boundary of $\mathcal{R}$, we instead place $r_1$ in the lower left corner, maximizing the width covered by $r_1$ while still covering the entire left border of $\mathcal{R}$, and check whether we can guarantee successful recursion on the remaining rectangle.
Otherwise, we place $r_6$ below $r_1$, covering the remainder of $\mathcal{R}$'s left border while maximizing the width of the covered rectangle.
We subdivide the remaining uncovered region into two rectangles:
$\mathcal{A}$ to the right of $r_1$ and $\mathcal{B}$ below $r_1$; see Fig.~\ref{fig:rectanglesSB-seven-disk-strategies}~(f).
After placing $r_6$, we build a collection $D_{\mathcal{A}}$ by adding disks in decreasing order of radius until we can recurse on $\mathcal{A}$;
if we can build such a collection and there is enough remaining weight to successfully recurse on $\mathcal{B}$, we are done.
Otherwise, we also consider placing $r_7$ below $r_1$ covering $\mathcal{B}$ completely, and then check for successful recursion on $\mathcal{A}$.
If that does not work, we disregard this routine.

\revsocg{
\subsubsection{Concluding the Proof}
As outlined in Section~\ref{sec:interval-arithmetic}, we implemented the success criteria of the routines described in this section using interval arithmetic.
Running this implementation on the space induced by $\lambda \in [1,2.5]$ that is left after applying Lemma~\ref{lem:rectanglesSB-handling-large-skew} yields no critical hypercuboids after inspecting more than $2^{46}$ hypercuboids in total.
This proves that, for any $\lambda \in [1,2.5]$ and any valid $r_1,\ldots,r_7$, at least one of our success criteria holds and thus, at least one of our routines works, thus concluding the proof for Lemma~\ref{lem:rectanglesSB}.
}%
}

\def\proofdetailswallbuilding{%
\subsection{Proof of Lemma~\ref{lem:rectangleSB-wall-building}}\label{sec:proof-lemma-wall-building}
In this section, we give the proof of Lemma~\ref{lem:rectangleSB-wall-building}; see Fig.~\ref{fig:rectangleSB-wall-building}.
\restatethm{\lemmawallbuilding*}{lem:rectangleSB-wall-building}
\begin{proof}
	We use the following simple algorithm to cover a strip, selecting the dimension $b$ in the process; in the following, we assume the strip to be vertical as depicted in Fig.~\ref{fig:rectangleSB-wall-building}.
	We begin by placing the first disk $q_1$ covering a square of side lengths $\sqrt{2} \cdot q_1$.
	By this placement, $q_1$ covers area $2q_1^2$, i.e., it has covering coefficient $\frac{1}{2} < E$.
	In decreasing order of radius, we keep placing disks on top of the previously placed disks such that they each cover a rectangle of width $\sqrt{2} \cdot q_1$.
	As long as $E > \frac{1}{2}$, Condition~(2) guarantees that each $q_i$ can cover a rectangle of width $\sqrt{2} \cdot q_1$ and height $h_i = \sqrt{4q_i^2 - 2q_1^2} > 0$.
	Moreover, we can prove the following Proposition~(4):
	For each disk $q_i$ placed covering a rectangle of dimensions $\left(\sqrt{2}\cdot q_1\right) \times h_i$ in this manner, we have $q_i^2 \leq E \cdot \sqrt{2}q_1 \cdot h_i$.
	In other words, the disks placed in this way cover area with coefficient at most $E$.
	In order to prove Proposition~(4), we first observe that the covering coefficient of a disk covering a rectangle decreases monotonically with increasing skew of the rectangle.
	Therefore, and because $q_i \leq q_1$, to verify Proposition~(4), it suffices to consider a disk $q_i$ of minimum allowed radius according to Condition~(2).
	In that case, we have \begin{align*}
		q_i^2 &= q_1^2\left(4E^2 - 2E\sqrt{4E^2-1}\right) = \left(E \cdot \sqrt{2}q_1\right) \cdot \sqrt{2}q_1 \cdot \left(2E-\sqrt{4E^2-1}\right)\\
		&= \left(E \cdot \sqrt{2}q_1\right) \cdot \sqrt{2}q_1 \cdot \sqrt{8E^2-4E\sqrt{4E^2-1}-1}\\
		&= \left(E \cdot \sqrt{2}q_1\right) \cdot \sqrt{4q_1^2\left(4E^2-2E\sqrt{4E^2-1}\right)-2q_1^2} = E \cdot \sqrt{2}q_1 \cdot h_i,%
	\end{align*} as claimed by Proposition~(4).
	We continue stacking disks until the covered region has height $\ell' \geq L(\ell)$; this eventually happens because of Condition~(3) and Proposition~(4).
	Because of Condition~(1), we know that at this point, $\ell' \leq \ell$ holds; no disk can cover more than $\sqrt{2} \cdot q_1$ height.

	If $\ell' = \ell$, we are done; otherwise, we proceed as follows.
	Starting from $\sqrt{2}q_1$, we reduce the width $b$ of the strip, adapting the height of the rectangle covered by each disk accordingly, until the covered height is exactly $\ell$; see Fig.~\ref{fig:rectangleSB-wall-building}.
	We know that before reducing the width, the coefficient of our cover is at most $E$; moreover, the width of each $q_i$'s rectangle is at least its height.
	It only remains to be proved that the covering coefficient stays at most $E$ after reducing $b$.
	Again, we prove this for each disk individually.
	In other words, we prove that the ratio between its weight and the area of its corresponding rectangle is at most $E$.
	Because the coefficient of a disk covering an inscribed rectangle depends on the skew of the rectangle, this is equivalent to proving that the covered rectangle does not become too high for any of the disks.
	Because all rectangles have the same width $b$, it suffices to show that the rectangle corresponding to $q_1$ does not become too high.
	Towards that goal, we first bound the factor by which we have to increase the height of $q_1$'s rectangle.
	Assume we reduce $b$ by some amount $\sqrt{2}q_1 > \delta_b > 0$ and consider the factor $X(\delta_b,q_i)$ by which the height of $q_i$'s rectangle increases.
		\begin{align*}
			X(\delta_b, q_i) &= \sqrt{\frac{4q_i^2-\left(\sqrt{2}q_1-\delta_b\right)^2}{4q_i^2-2q_1^2}} = \sqrt{1+\frac{\left(2\sqrt{2}q_1-\delta_b\right)\delta_b}{4q_i^2-2q_1^2}}\text{, and}\\
			\frac{\partial X(\delta_b, q_i)}{\partial q_i} &= \frac{1}{2\sqrt{1+\frac{\left(2\sqrt{2}q_1-\delta_b\right)\delta_b}{4q_i^2-2q_1^2}}} \cdot \frac{-8q_i\cdot\left(2\sqrt{2}q_1-\delta_b\right)\delta_b}{\left(4q_i^2-2q_1^2\right)^2} < 0,
		\end{align*}
	so increasing the height of $q_1$'s rectangle by some factor increases the total covered height by at least that factor.
	In total, we have to increase the covered height by a factor of at most $\frac{\ell}{L(\ell)} = \sqrt{1+\sqrt{1-\frac{1}{4E^2}}}$.
	Therefore, in the worst case, we have \begin{align*}%
		h_1 &= \sqrt{2}q_1 \cdot \sqrt{1+\sqrt{1-\frac{1}{4E^2}}},\ b = \sqrt{4q_1^2 - h_1^2} = \sqrt{4q_1^2 - 2q_1^2\left(1+\sqrt{1-\frac{1}{4E^2}}\right)},\\
		&= \sqrt{2}q_1 \cdot \sqrt{1-\sqrt{1-\frac{1}{4E^2}}}\textup{, and }h_1 \cdot b = 2q_1^2 \cdot \sqrt{\frac{1}{4E^2}} = \frac{q_1^2}{E},
	\end{align*}
	which implies that $q_1$ covers its rectangle with covering coefficient $E$.
\end{proof}%
}

\FloatBarrier
\ifthenelse{\boolean{applemmaworstcasesrectangles}}{}{
	\prooflemmaworstcasesrectangles
}
\ifthenelse{\boolean{applemmasizeboundlarge}}{}{
	\prooflemmasizeboundlarge
}
\ifthenelse{\boolean{appproofsizebounded}}{}{
	\proofdetailsweightboundedcovering
}
\ifthenelse{\boolean{appproofwallbuilding}}{}{
	\proofdetailswallbuilding
}
\ifthenelse{\boolean{appproofmainthm}}{}{
	\proofdetailsmaintheorem
}

\FloatBarrier
\section{Conclusion}
We have given a tight characterization of the critical covering
density for arbitrary rectangles. This gives rise to numerous
followup questions and extensions.

As discussed (and shown in Fig.~\ref{fig:worst-cases-rectangles}), the worst-case values
correspond to instances with only 2 or 3 relatively large disks;
if we have an upper bound $R$ on the size of the largest disk,
this gives rise to the critical covering area $A^*_R(\lambda)$ for
$\lambda\times 1$-rectangles.
Both from a theoretical and a practical point of view, getting some tight bounds on $A^*_R(\lambda)$ would be interesting and useful.
Our results of Lemma~\ref{lem:size-bound-large} and Lemma~\ref{lem:rectanglesSB} indicate possible progress in that direction; just like for unit disks, tighter results will require considerably more effort.

Establishing the critical covering density for disks and triangles is also open. We are optimistic
that an approach similar to the one of this paper can be used for a solution.

Finally, \emph{computing} optimal coverings by disks appears to be quite difficult.
However, while deciding whether a given collection of disks can be packed into a unit square is known to be NP-hard~\cite{DFL2010circle},
the complexity of deciding whether a given set of disks can be used to cover a unit square is still open.
Ironically, it is the higher practical difficulty of covering by disks that makes it challenging to apply a similar idea in a straightforward manner.

\bibliography{references}
\end{document}